\documentclass[a4paper,draft,12pt,openright]{article} 
\usepackage[english]{babel} 
\usepackage{amsmath,amssymb,amsthm,amsfonts,amscd,authblk,color,cancel,diagbox,enumerate,epsfig,eucal,latexsym,lmodern,mathrsfs,mathtools,multirow,skull,simplewick,subcaption,tikz-cd}
\usepackage[normalem]{ulem}
\usepackage[
final=true,                    
plainpages=false,              
pdfpagelabels=true,            
pdfencoding=auto,              
unicode=true,                  
hypertexnames=true,            
naturalnames=true              
]{hyperref}
\usepackage[all,cmtip]{xy}
\usetikzlibrary{matrix,calc}
\usepackage[autostyle, english = american]{csquotes}
\MakeOuterQuote{"} 

\usepackage[multiuser]{fixme}
\FXRegisterAuthor{nd}{annd}{Nico}
\FXRegisterAuthor{sm}{smnd}{Sonia}
\FXRegisterAuthor{vm}{vmnd}{Valter}

\newtheoremstyle{TheoremStyle}
{3pt}
{3pt}
{\slshape}
{}
{\sc}
{:}
{.5em}
{}

\theoremstyle{TheoremStyle}
\newtheorem{theorem}{\bf Theorem}
\newtheorem{corollary}[theorem]{\bf Corollary}
\newtheorem{proposition}[theorem]{\bf Proposition}
\newtheorem{lemma}[theorem]{\bf Lemma}
\newtheorem{definition}[theorem]{\bf Definition}
\newtheorem{remark}[theorem]{\bf Remark}

\newcommand{\bR}{{\mathbb{R}}}

\newcommand{\bC}{{\mathbb{C}}}

\newcommand{\bN}{{\mathbb{N}}}

\title{Feynman path integrals on compact Lie groups with bi-invariant Riemannian metrics}
\author[a,b]{N. Drago\thanks{\href{mailto:nicolo.drago@unige.it}{nicolo.drago@unige.it}}}
\author[a]{S. Mazzucchi\thanks{\href{mailto:sonia.mazzucchi@unitn.it}{sonia.mazzucchi@unitn.it}}}
\author[a]{V. Moretti\thanks{\href{mailto:valter.moretti@unitn.it}{valter.moretti@unitn.it}}}
\affil[a]{Dipartimento di Matematica, Universit\`a di Trento and INFN-TIFPA\\Via Sommarive 14, I-38123 Povo (Trento), Italy}
\affil[b]{Dipartimento di Matematica, Universit\`{a} di Genova and INdAM, Via Dodecaneso 35, I-16146 Genova, Italy}
\begin{document}
\maketitle

\begin{abstract}
\noindent
In this work we  consider a suitable generalization of the Feynman path integral on a specific class of Riemannian manifolds consisting of compact Lie groups with bi-invariant Riemannian metrics. The main tools we use are the Cartan development map, the notion of oscillatory integral and the Chernoff approximation theorem. We prove that, for a  class of functions of a dense subspace of the relevant Hilbert space, the Feynman map produces the solution of the Schr\"odinger equation, where the Laplace-Beltrami operator coincides with the second order Casimir operator of the group. 
\end{abstract}
\tableofcontents

\section{Introduction}

Since their introduction, Feynman path integrals have always been both a powerful quantization  tool and a source of challenging mathematical problems. They appeared for the first time in \cite{FeyT,Fey}, where an alternative Lagrangian formulation of time evolution in quantum mechanics was introduced. According to Feynman's proposal, the solution of the
Schr\"odinger equation
\begin{equation} \label{Schroedinger}
\begin{dcases}
    i\hbar\frac{\partial}{\partial t}\psi(t,x)
    =-\frac{\hbar^2}{2m}\Delta \psi (t,x)+V(x)\psi(x)
    \\
    \psi (0,x)=\psi _0(x), \qquad \psi_0\in C^\infty_0(\bR^{d})
    \end{dcases}\,,
\end{equation}
should be given  by an heuristic integral of the following form:
\begin{equation} \label{feynman}
\psi (t,x)=\quad {}^{``}\quad C^{-1}\int _{ \Gamma
}e^{\frac{i}{\hbar}S (\gamma)}\psi _0(\gamma (0))\mathrm{d}\gamma\quad{}^{"}
 \end{equation}
where $\Gamma$ denotes a set of paths $\gamma :[0,t]\to \bR^{d}$ with fixed end point $\gamma (t)=x$, $\mathrm{d}\gamma $ stands for a Lebesgue-type measure on $\Gamma$ while the function $S:\Gamma \to \bR$ denotes the classical action functional of the system, namely $$S (\gamma)=S^0(\gamma)-\int _0^tV(\gamma(s))\mathrm{d}s, \qquad S^0(\gamma)=\frac{m}{2}\int _0^t\vert\dot\gamma(s)\vert^2\mathrm{d}s. $$
Finally, the symbol  $C$ in \eqref{feynman} plays the role of a normalization constant.
In what follows we will set $m=1$ without loss of generality.

Feynman's formula \eqref{feynman}, as it stands, lacks of a sound mathematical meaning. Indeed the Lebesgue-type measure $\mathrm{d}\gamma $ on the infinite dimensional path space $\Gamma$ as well  the - actually infinite - normalization constant $C$ cannot be properly defined (see e.g. \cite{Ma-book} for a detailed discussion of these and related issues). The common interpretation of \eqref{feynman} is in terms of the limit of a suitable approximating sequence. By considering an equally spaced partition of the interval $[0,t]$ into $n$ subintervals $[jt/n,(j+1)t/n)]$, $j=1,\dots ,n-1$, and by restricting to the space of piecewise-linear paths 
with constant velocity along the partition subintervals, then the (heuristic infinite-dimensional) integral on the right hand side of \eqref{feynman} can be approximated by a finite-dimensional integral over the endpoints $x_j\equiv \gamma (jt/n)$ of the piecewise linear paths:
\begin{equation}\label{approx-fey0}
(2\pi i \hbar t/n)^{-nd/2}\int_{\bR^{nd}} e^{\frac{i}{\hbar}\sum_{j=1}^{n}\Big(\frac{(x_j-x_{j-1})^2}{2(t/n)^2}-V(x_j)\Big)\frac{t}{n}}\psi_0(x_0)\mathrm{d}x_0\dots \mathrm{d}x_{n-1}\,.
\end{equation}
According to Feynman intuition, when $n\to\infty $ the  sequence \eqref{approx-fey0}  converges to the solution $\psi (t,x)$ of Equation \eqref{Schroedinger}. This result can be proved (see e.g. \cite{Fuj,NicTra}) under suitable assumptions on the potential $V$, yet leaving open  the issue of the rigorous construction of path integrals \eqref{feynman} within Lebesgue integration theory. This problem
soon attracted the attention of the mathematical community, triggering the study  of the relation between  partial differential equations, stochastic processes and probability measures on path spaces. In particular, Feynman's idea inspired the proof of the {\it Feynman-Kac formula} \cite{Kac49,Kac51}, i.e. a representation for the solution of the heat equation in terms of an integral with respect to the Wiener probability measure over the space of continuous paths. The early attempts to extend Feynman-Kac formula to the Schr\"odinger equation and to realize an elusive "Feynman measure" $\mu_F$ in terms of a Wiener measure with complex covariance ended in  1960 with an important no-go result \cite{Cam}, showing that, unlike heat equation and Wiener measure, Feynman heuristic formula \eqref{feynman} cannot be rigorously defined in terms of a Lebesgue integral on the (infinite-dimensional) space of paths $(\bR^d)^{[0,t]}$. Indeed, denoting $K_t(x,y):=(2\pi i \hbar t )^{-d/2}e^{\frac{i}{2\hbar t}(x-y)^2}$ the fundamental solution of the Schr\"odinger equation \eqref{Schroedinger} with $V=0$, a generalization of Kolmogorov existence theorem to the case of complex  measures \cite{Tho} allows to prove that the finite-additive complex measure $\mu$ defined on the algebra $ {\mathcal A}$ of cylinder sets in $(\bR^d)^{[0,t]}$ of the form
$$E_{t_1,\ldots, t_n; B_1,\ldots,B_n}:=\{\gamma\in (\bR^d)^{[0,t]}\colon \gamma (t_1)\in B_1, \ldots, \gamma (t_n)\in B_n\}, $$
for some $n\geq 1$, $t_1,\ldots, t_n\in [0,t]$, $B_1,\ldots, B_n $ Borel sets in $\bR^d$, 
as
\begin{equation}\label{cyl-mes}
    \mu(E_{t_1,\ldots, t_n; B_1,\ldots,B_n})=\int_{B_1}\dots\int_{B_n}K_{t_n-t_{n-1}}(x_{n-1},x_n)\dots K_{t_2-t_1}(x_1,x_2) K_{t_1}(x,x_1)\mathrm{d}x_1\dots \mathrm{d}x_n\,,
\end{equation}
cannot be extended to a $\sigma$-additive measure on  the $\sigma$-algebra $\sigma({\mathcal A})$ generated by $\mathcal{A}$. 
In order to deal with the lack of an underlying measure, different approaches to the definition of formula \eqref{feynman} have been proposed \cite{AlHKMa,DeFaPoStre,Fuj,JoLa,Klauder,Ma-book,Muldowney2012}. A common feature of most of them is the replacement of the concept of Lebesgue-type integral with respect to a $\sigma$-additive measure with the more general concept of a linear functional $L: D(L)\to \bC $ on a domain $D(L)$ of "integrable functions" (see \cite{AlMa16} for an detailed discussion of this topic). In order to reproduce all the properties suggested by formulae \eqref{feynman}  and \eqref{cyl-mes}, the domain $D(L)$ should contain  the {\it cylinder functions}, i.e. those  functions $f:(\bR^d)^{[0,t]}\to\bC$ of the form 
\begin{equation}\label{cil-f}f(\gamma):=g(\gamma(t_1),\ldots,\gamma(t_n)), \qquad \gamma \in (\bR^d)^{[0,t]}\end{equation} for some $n\in \bN$, $t_1,\ldots,t_n\in [0,t]$ and suitable classes of Borel function $g: \bR ^d\times \dots \times \bR ^d\to \bC$, where the restrictions on $g$ depend on the particular construction procedure of the functional.
More importantly 
the action of the functional $L$ on the function \eqref{cil-f} must be given by a (finite-dimensional) integral of the form: 
\begin{equation}\label{funz-cyl}L(f)=\int_{\bR^d\times\dots \times \bR^d}g(x_1,\ldots,x_n)
K_{t_n-t_{n-1}}(x_{n-1},x_n)\dots K_{t_2-t_1}(x_1,x_2) K_{t_1}(x,x_1)\mathrm{d}x_1\ldots \mathrm{d}x_n.\end{equation}

While formula \eqref{feynman}  is  extensively studied in the case where the configuration space is the Euclidean space $\bR^d$, its generalizations to a $d$-dimensional Riemannian manifold $M$ with a metric $\boldsymbol{g}$ remains essentially an open problem. In this case the Schr\"odinger equation becomes
\begin{align}\label{SchroedingerM}  
    \begin{dcases}
        i\hbar\frac{\partial\psi}{\partial t}
        =-\frac{\hbar^2}{2}\Delta_{\boldsymbol{g}} \psi +V\psi
        \\
        \psi (0,x)=\psi _0(x)
    \end{dcases} 
\end{align}
where $\Delta_{\boldsymbol{g}}$  denotes the Laplace-Beltrami operator associated to the metric $\boldsymbol{g}$. Remarkably, in the physics literature \cite{Schu} the heuristic Feynman's formula \eqref{feynman} is replaced by the following
\begin{equation} \label{feynmanM}
\psi (t,x)=\quad {}^{``}\quad C^{-1}\int _{ \Gamma
}e^{\frac{i}{\hbar}S (\gamma)}e^{i\hbar k\int_0^t R(\gamma(s))\mathrm{d}s}\psi _0(\gamma (0))\mathrm{d}\gamma\quad{}^{"}
 \end{equation}
containing an additional term of the form $\hbar^2 k\int_0^t R(\gamma(s))\mathrm{d}s $ that has to be added to the classical action, where $R$ is the scalar curvature of the manifold and $k$ is a numerical constant whose value  actually depends on the approximation scheme (usually $k=\frac{1}{12}$ or $k=\frac{1}{6}$ \cite{DeWitt,Schu}).
In the case of the heat equation on a Riemannian manifold $M$ and the corresponding Feynman-Kac formula, there exist several interesting results addressing the problem of the construction of the path integral and the interpretation of the scalar curvature correction term \cite{AnDri,Bar,BarPfa}.
In particular, according to Ref. \cite{AnDri}, its appearance seems to be linked to the geometry the Hilbert manifold $\Gamma$, and the term $\rho(\gamma)\equiv e^{-\frac{1}{6}\int_0^tR(\gamma(s)ds}$ is interpreted as a Jacobian factor between two volume measures on  (the finite-dimensional approximations of) $\Gamma$ associated to different metrics. This interpretation doesn't seem to be generalizable to  the case of Feynman's formula where the oscillatory term $\rho(\gamma)=e^{i\hbar k\int_0^t R(\gamma(s))\mathrm{d}s}$ is a complex valued function, which  cannot be understood as a Jacobian term. In particular, 
 when dealing with the Schr\"odinger equation and its corresponding Feynman formula, only few {\em rigorous }mathematical results have been obtained \cite{Fuk,Tha1,Tha2}.
More specifically, in \cite{Fuk} the author proves that for compact manifolds $M$ the time-slicing approximation of formula \eqref{feynman} converges to the solution of Equation \eqref{SchroedingerM} with $V$ replaced by $V+\frac{\hbar^2}{12}R$.
In \cite{Tha1,Tha2} no scalar curvature correction term appears and the author realizes representation \eqref{feynman} for the solution of Equation \eqref{SchroedingerM} in terms of a Feynman-Kac formula constructed out of a particular stochastic process with values in the complexification of $M$.
Due to the particular techniques used, those results are restricted to the case where $M$ is a compact connected semisimple Lie group or a symmetric space.

In the present paper, we  study the rigorous mathematical construction of representation formula \eqref{feynman} for the solution of  the Schr\"odinger equation on manifolds by means of  the {\it infinite dimensional oscillatory integral approach} \cite{AlBr,ELT,Ma-book}, which relies on a generalization of the definition and the main properties of classical oscillatory integrals on $\bR^n$ \cite{Hor1} to the case where the integration domain is an infinite dimensional real separable Hilbert space. In the case of $M=\mathbb{R}^d$ this approach allows to define Feynman integral \eqref{feynman} in terms of a well defined continuous functional on a Banach algebra of functions on a suitable Hilbert space of paths $\gamma:[0,t]\to\mathbb{R}^d$, by preserving at the same time Feynman's original sequential construction.
In particular, it allows the implementation of an infinite dimensional version of the classical stationary phase method \cite{AlHK77} and the corresponding application  to the study of the semiclassical asymptotic behaviour of the solution of Schr\"odinger equation \eqref{Schroedinger} in the limit where the reduced Planck constant $\hbar$ is regarded as a small parameter, thus creating a direct link between classical and quantum description. In addition, infinite dimensional oscillatory integrals have proven to be particularly flexible in the Euclidean case, allowing to provide a rigorous mathematical definition of the heuristic Feynman formula \eqref{feynman} for a large class of potentials $V$ \cite{AlCaMa,AlMa05,Ma-book}. However, the generalization of this results to the case where $\mathbb{R}^d$ is replaced by a Riemannian manifold $M$ is up to now an open problem, since only very preliminary results and conjectures can be found in the literature (see Ref. \cite{ELT81} and the discussion in Section \ref{SECCONJ}). Indeed, in this case the construction should rely upon the {\em Cartan development map}, which on the one hand provides an elegant technique for transferring  the theory from paths in $\mathbb{R}^d$ to the case of curved spaces, but on the other hand results to be rather difficult to handle, as it produces rather implicit formulae that do not allow for simple explicit computations.

In the present work we shall focus on the construction of the representation formula \eqref{feynman} via the {\it Feynman maps} \cite{ELT,ELT81},  a particular infinite dimensional oscillatory integral that is closer to Feynman's original construction and that turned out to work nicely even in rather tricky cases \cite{AlCaMa}.
Our construction is inspired by the Euclidean case \cite{AlHK77,AlHKMa}, which benefits from techniques of harmonic analysis.
For this reason we will limit our analysis to the case of a compact Lie group $G$ endowed with a bi-invariant metric $\boldsymbol{g}$.
In particular,  if $G$ is  a connected compact semi-simple Lie group, $\boldsymbol{g}$ must coincide to the   Killing form of $G$ up to a positive constant factor. 

Informally, our main results may be summarised as follows:\\

\noindent {\bf Theorem}:
\textit{Let $G$ be a compact Lie group with bi-invariant metric $\boldsymbol{g}$. 
    Let $\mathcal{H}_{x,t}(G)$ be the space of absolutely continuous curves $\gamma\colon [0,t]\to G$ such that $\gamma(0)=x$ and $\|\gamma\|_{\mathcal{H}_{x,t}^2(G)}:=\int_0^t\boldsymbol{g}(\dot{\gamma}(s),\dot{\gamma}(s))\mathrm{d}s<+\infty$.
    Then there exists a linear map 
    \begin{align*}
      f\mapsto\mathcal{F}_{\mathcal{H}_{(x,t)}(G)}(f)
      =\widetilde{\int}_{\mathcal{H}_{x,t}(G)}
      e^{\frac{i}{2\hbar}\|\dot{\gamma}\|_{\mathcal{H}_{x,t}(G)}^2}
      f(\gamma)\mathrm{d}\gamma\,,
    \end{align*}
    defined on the algebra of functions $f\colon\mathcal{H}_{(x,t)}(G)\to\mathbb{C}$ for which
    \begin{align*}
        f(\gamma)=\phi_1(\gamma(t_1))\cdots\phi_k(\gamma(t_k))\,,
    \end{align*}
    where $k\in\mathbb{N}$, $0\leq t_1\leq \ldots\leq t_k\leq t$ and $\phi_1,\ldots,\phi_k\colon G\to\mathbb{C}$ are finite energy functions, \textit{cf.} Definition \ref{Def: finite energy functions}.
    The resulting map, which is constructed through the Cartan map ---\textit{cf.} Definition \ref{Def: Cartan map} ---, is a limit of oscillatory integrals over spaces of increasing finite dimension.
    Moreover, for $f$ as above, it holds
    \begin{align*}
        \mathcal{F}_{\mathcal{H}_{x,t}(G)}(f)
        =\bigg[U(t_1)\phi_1U(t_2-t_1)\phi_2\cdots U(t_{k-2}-t_{k-1})\phi_{k-1}U(t_k-t_{k-1})\phi_k\bigg](x)
    \end{align*}
    where $U(t):=e^{\frac{i\hbar t}{2}\overline{\Delta_{\boldsymbol{g}}}}$ is the unitary group generated by the closure of the Laplace-Beltrami operator $-\Delta_{\boldsymbol{g}}$ associated to the metric $\boldsymbol{g}$.
    In particular, the map $\mathcal{F}_{\mathcal{H}_{x,t}(G)}$ provides a representation of the unitary $U(t)$.\\
    Finally, { the map $\mathcal{F}_{\mathcal{H}_{x,t}(G)}$ can be extended to a class of relevant non-cylinder functions. In particular, }if $\psi_0,V\colon G\to\mathbb{C}$ are finite energy functions and $m\in \mathbb{N}$, the map $\mathcal{F}_{\mathcal{H}_{x,t}(G)}$ can be applied to the function $f\colon\mathcal{H}_{x,t}(G)\to\mathbb{C}$ defined by
    \begin{align*}
        f(\gamma)
        :=\psi_0(\gamma(t))\left(\int_0^tV(\gamma(s))\mathrm{d}s\right)^m\,,
    \end{align*}
    and $\mathcal{F}_{\mathcal{H}_{x,t}(G)}(f)$ provides the $m$-th term of the convergent Dyson perturbative series for the  solution to the Schr\"odinger equation \eqref{SchroedingerM} with potential $V$ and initial data $\psi_0$.} \\

We refer to Theorems \ref{FinalTheoremSchroedinger1}-\ref{teo-cyl-funct}-\ref{Thm: Feynman map with potential} for a more precise discussion of these results.
We point out that the finite energy assumption restricting the class of admissible cylinder functions $f\colon\mathcal{H}_{x,t}(G)\to\mathbb{C}$ is a natural generalization of the hypothesis considered in the Euclidean setting, \textit{cf.} Remark \ref{Rmk: finite energy assumption}.
Our result provides the first rigorous construction on non-Euclidean structures of infinite dimensional oscillatory integrals within the theory developed in \cite{AlBr,AlHK77,AlHK77,ELT} and further developed in \cite{AlMa05,AlCaMa}, paving the way for further applications, such as, e.g., the study of the semiclassical asymptotics of the solution of Eq. \eqref{SchroedingerM} in the limit $\hbar\downarrow 0$ via  the  infinite dimensional version of the stationary phase method developed in \cite{AlHK77,AlBr}. Furthermore, it  creates a link among  different approaches to the problem, such as the analytic continuation of Wiener integrals \cite{Tha1,Tha2} and the time slicing construction \cite{Fuk},  cf. remark \ref{Rmk: comparison with literature}. As proposed in \cite{ELT81}, the definition of the functional $\mathcal{F}_{\mathcal{H}_{x,t}(G)}$ relies on two basic ideas such as the theory of oscillatory integrals on infinite dimensional Hilbert spaces and the Cartan development map. In particular in  our case, the restriction to Lie groups with bi-invariant metrics allows, on the one hand, the derivation of explicit and tractable formulae for the action of the Cartan maps on the space of paths under consideration and, on the other hand, the exploitation of the non-commutative harmonic analysis on the Lie group.   It is worth noting that our approach to the mathematical definition of heuristic Feynman's formula \eqref{feynman} for the solution of Equation \eqref{SchroedingerM} provides an intrinsic construction that does not require the introduction of any {\em explicit }scalar curvature correction term,  thanks to the exploitation of the Cartan map, which encodes the geometry of the underlying manifold. Besides, we also show (see Remark \ref{remark-scalar-curvature}) how the scalar curvature still comes into play in our context when alternative measures on the finite dimensional approximations of the path space are chosen. In particular,   the  complex correction term $e^{i\hbar k\int_0^t R(\gamma(s))\mathrm{d}s} $  in Eq. \eqref{feynmanM} can still result from a real Jacobian factor between two different reference measures.

The paper is organized as follows. In Section \ref{Sec: preliminary definitions and results} we set the notation and recall some results on Lie groups and invariant metrics. In Section \ref{sez-Fey-M} we provide the definition of Feynman map on the Euclidean space $\bR^d$ and its generalization to  a Riemannian manifold $M$. In Section \ref{sez-fey-G} we restrict ourselves to the case where $M$ is a compact Lie group $G$ endowed with a bi-invariant metric   and provide some explicit formulas for the corresponding Feynman map, proving that   it can be regarded as a linear functional  satisfying condition \eqref{funz-cyl} on a suitable class of cylinder functions. 
Finally in Section \ref{sez-Schr-V} we consider Equation \eqref{SchroedingerM} with $V\neq 0$ and provide a perturbative solution.

\section{Some notions and results of Lie-group theory}
\label{Sec: preliminary definitions and results}
In the rest of the paper, if $H$ is a Hilbert (or Banach) space, $\mathfrak{B}(H)$ denotes the Banach algebra of bounded operators $H\to H$.

$G$ henceforth denotes a  real $d$-dimensional Lie group \cite{Warner} with identity $e$.   We explicitly assume $d < +\infty$.  The differentiable structure of  $G$ is assumed to be the unique smooth ($C^\infty$) structure and all geometric structures on $G$ are supposed to be smooth accordingly. $\mathfrak{g}$ will denote the tangent space at the unit element $T_eG$, while 
$\Gamma(TG)$ denotes the module of smooth vector fields on $G$. 

We recall here a few very well known facts on Lie groups and associated Lie algebras \cite{Warner,Serre}  to establish some relevant definitions and notations used throughout.

For all $x\in G$ we will denote by
\begin{align*}
    L_x\colon
    G\ni y
    \mapsto L_x(y):=xy\in G\,,
    \qquad
    R_x\colon
    G\ni y
    \mapsto R_x(y):=yx\in G\,,
\end{align*}
 the \textbf{left-translation} and \textbf{right-translation} respectively. These notions are defined for general groups $G$, but we are interested in the case of a Lie group. In that case, the  maps $G\ni x\mapsto L_x\in\operatorname{Diff}(G)$ and $G\ni x\mapsto R_x\in\operatorname{Diff}(G)$ are, respectively, a group representation of $G$ and a  group representation of\footnote{As is known, the {\bf opposite group} $(G^{\scriptsize \mbox{op}}, \circ^{\scriptsize \mbox{op}})$ of a group $(G, \circ)$ is the unique group structure constructed on the \textit{set} $G$ with the product $x  \circ^{\scriptsize \mbox{op}} y :=  y\circ x$ for $x,y\in G$.} $G^{\scriptsize \mbox{op}}$ in terms of smooth diffeomorphisms of $G$. For a given pair $x,y\in G$, we shall denote  by
\begin{align*}
    (\mathrm{d}L_x)_y\colon T_yG\to T_{xy}G\:,
    \qquad
    (\mathrm{d}R_x)_y\colon T_yG\to T_{yx}G
\end{align*}
the differentials of the left-action and  the right-action.
These maps are vector space isomorphisms by construction.

A vector field $X\in\Gamma(TG)$ is \textbf{left-invariant}, respectively \textbf{right-invariant}   if
$$X(xy)=(\mathrm{d}L_x)_yX(y)\quad \mbox{or}\quad X(xy)=(\mathrm{d}R_y)_xX(x)\:,\quad \forall x,y \in G\:. $$ 
With the symbol
  $\mathfrak{g}^L \subset \Gamma(TG)$ (resp. $\mathfrak{g}^R \subset \Gamma(TG)$) we will  denote the subspace of left-invariant (resp. right-invariant) smooth vector fields.
 The space $\mathfrak{g}^L$  ($\mathfrak{g}^R$) is   isomorphic to $\mathfrak{g}= T_eG$ since the value of a left-invariant (respectively, right-invariant) vector field at $e$ uniquely defines it.
 For later convenience we will denote by
 \begin{align*}\mathfrak{g}\ni X \mapsto \widetilde{X}\in \mathfrak{g}^L \subset \Gamma(TG)\:, \quad \mbox{where $\widetilde{X}(x):= (dL_x)_eX$.}\label{ISOTILDE}\end{align*}
 the canonical isomorphism between $\mathfrak{g}^L$ and $\mathfrak{g}$.
Observe that $\widetilde{X}(e)=X$ and 
 $\mathfrak{g}\ni X \mapsto \widetilde{X}(x) \in T_xG$ is also a vector space isomorphism for every given $x\in G$,  thus the left-invariant vector fields provide a basis of the tangent space at every point of $G$.

It is easy to prove that the Lie commutator of a pair of left-invariant (right-invariant) vector fields is left-invariant (respectively, right-invariant).
As a consequence,  the real $d$-dimensional  
vector space $\mathfrak{g}^L$ (respectively, $\mathfrak{g}^R$) equipped with the standard  Lie commutator of vector fields  $[\cdot,\cdot]: \Gamma(TG)\times \Gamma(TG) \to \Gamma(TG)$ is a Lie algebra.
 The canonical isomorphism between $\mathfrak{g}$ and $\mathfrak{g}^L$ introduced above induces  a Lie commutator  $[\cdot, \cdot] : \mathfrak{g}\times  \mathfrak{g}\to  \mathfrak{g}$, which 
 is defined by
\begin{align}
    [X,Y]:=[\widetilde{X},\widetilde{Y}]_e
    \qquad\forall X,Y\in\mathfrak{g}\,.
\end{align}
As it is well known,  the Lie algebra 
$(\mathfrak{g}, [\cdot,\cdot])$ is called the {\bf Lie algebra of $G$}.

 Since the Lie commutator of $\mathfrak{g}$ uniquely defines a tensor of order $(1,2)$, called the {\bf structure tensor}, we can profitably use the tensor technology to describe its action.
Let $X_1, \ldots, X_{d} \in \mathfrak{g}$ be a basis of $\mathfrak{g}$
and let us introduce the notation $X^{*k}$ for the elements of the dual basis.
Every element $X\in \mathfrak{g}$ can be therefore written in components $X= \sum_{k=1}^{d} x^k X_k$ and the action of the Lie commutator can be written as
$[X,Y]^k = {c_{ij}}^k x^iy^j$ referring to that basis,
where we adopted, and we henceforth do, Einstein's convention of summation over repeated indices from $1$ to $d$. The components of the structure tensor sometimes known as the {\bf structure constants} of $G$ are
\begin{align}
{c_{ij}}^k = \langle [X_i,X_j], X^{*k} \rangle\:.
\end{align}
By definition, the structure tensor is anti-symmetric in the lower indices 
${c_{ij}}^k = -{c_{ji}}^k$.

A Riemannian metric $\boldsymbol{g}$ on a Lie group is said to be  {\bf left-invariant} or 
{\bf right-invariant}
if, respectively, 
\[\boldsymbol{g}_y(X_y,Y_y)= \boldsymbol{g}_{xy}((dL_x)_yX_y, (dL_x)_yY_y) \:, \quad \forall X_y,Y_y \in T_yG\:, \forall x,y \in G\:,\]
or
\[\boldsymbol{g}_y(X_y,Y_y)=  \boldsymbol{g}_{yx}((dR_x)_yX_y, (dR_x)_yY_y)\:, \quad \forall X_y,Y_y \in T_yG\:, \forall x,y \in G\:.\]
Every diffeomorphism $L_x$ or, respectively $R_x$, is in that case also an \textit{isometry}.
The metric is {\bf bi-invariant} if it is both left- and right-invariant. 
The following result provides a complete characterization of the Lie groups that admit a bi-invariant metric.
\begin{proposition}
A connected Lie group admits a bi-invariant Riemannian metric if and only if it is isomorphic to the product of a compact Lie group and $\mathbb{R}^n$ with standard Lie group product structure.
\end{proposition} 

\begin{proof} Lemma 7.5 in \cite{Milnor-76}.\end{proof}

Once $G$ is endowed with a Riemannian structure, there are two different  notions of exponential map around $e$.
\begin{itemize}
\item One is the usual $\exp : \mathfrak{g}\ni X  \mapsto \exp(X) \in G$  \cite{Warner}, whith $\exp(X):= \gamma_X(1)$ where 
$\gamma : \mathbb{R} \to G$ is the unique (immersed) {\bf one-parameter  subgroup} of $G$ with tangent vector $X$ at $e$.
By definition $t\mapsto\gamma_X(t)$ is the maximal integral curve of $\widetilde{X}$ passing through  $e$ at $t=0$. These integral curves are always complete.

\item The other is the standard metric exponential map \cite{KN} $\exp_e^{\boldsymbol{g}} : U^{\boldsymbol{g}} \to G$ defined in an open  star-shaped neighborhood  $U^{ \boldsymbol{g}}$ of the origin of $T_eG$ and taking values on $G$. 
 By definition, $\exp_e^{\boldsymbol{g}}(X):= \gamma^{\boldsymbol{g}}_X(1)$ where 
$\gamma^{\boldsymbol{g}}_X : \mathbb{R} \to G$ is the unique geodesic with initial vector $X$ at $\gamma^{\boldsymbol{g}}_X(0) = e$ and maximal domain $I\ni 0$
which also includes $t=1$. This latter condition imposes restrictions on the possible $X$, namely, on the domain $U^{ \boldsymbol{g}}$ which may not coincide with the whole tangent space.
\end{itemize}
These two exponential maps coincide only when the metric $\boldsymbol{g}$ is bi-invariant, \textit{cf.} Appendix \ref{APPENDIXPROOFS}. Moreover, the following proposition describes some interesting properties of Lie groups endowed with bi-invariant metrics.

\begin{proposition}\label{PROPO6}  Let $G$ be a Lie group equipped with a bi-invariant  metric $\boldsymbol{g}$.
Then:
\begin{itemize}
\item[(a)] A smooth curve $\gamma : I \to G$, where $I\ni 0$ is an open interval and  $\gamma(0)= x\in G$, is  a $\boldsymbol{g}$-geodesic with maximal domain $I$ if and only if  it is complete ($I=\mathbb{R}$) and can be written in the form
\begin{align}
\gamma(t) = L_x\exp(tX)\quad \mbox{for every $t\in \mathbb{R}$ and some $X\in \mathfrak{g}$\,.}
\end{align}
In particular, $\exp_e^{\boldsymbol{g}}=\exp$, which is thus defined on the whole tangent space $T_eG$.
\item[(b)] The family of $\boldsymbol{g}$-geodesics with maximal domain coincides  with the family of the maximal  integral curves  of the left invariant vector fields of $G$.  Therefore  every such $\boldsymbol{g}$-geodesic $\gamma$ can be written in the form 
\begin{align}
    \mathbb{R} \ni t\mapsto\Phi^{\widetilde{X}}_t(x) \in G
    \qquad \widetilde{X}\in \mathfrak{g}^L\,,\,
    x\in G\,,
\end{align}
 where $\Phi^Z$ denotes the flow of a vector field $Z$.
The same geodesic  can  be written also as 
\begin{align}\mathbb{R} \ni t\mapsto R_{\exp(tX)}x\in G \:.\end{align}
\end{itemize}
If ${\boldsymbol{g}}$ is left-invariant but not right-invariant, then there is at least one integral curve of a left-invariant vector field which is not a ${\boldsymbol{g}}$-geodesic and thus properties (a),(b) are not valid.
\end{proposition}

\begin{proof} See Appendix \ref{APPENDIXPROOFS}.
\end{proof}

\begin{remark}
According to Proposition \ref{PROPO6}, if $\boldsymbol{g}$ is bi-invariant,
then for all $x\in G$ and $X_x\in T_xG$, the $\boldsymbol{g}$-geodesic $\gamma_{x,X_x}$ starting at $x$ with velocity $X_x$ and maximal domain is given by
\begin{align}
    \gamma_{x,X_x}(t)
    =x\exp\left[t(\mathrm{d}L_{x})_e^{-1}X_x\right]  = \Phi^{\widetilde{(dL_x)_e^{-1}X_x}}_t(x) = R_{\exp\left[t(\mathrm{d}L_{e})_x^{-1}X_x\right]} x\:,\quad \forall t\in \mathbb{R}\:.
\end{align}
\end{remark}

 As is well known (see, e.g., \cite{Cohn}), if $G$ is a topological locally compact group, a unique left-invariant positive $\sigma$-additive regular Borel measure which is finite on compact sets exists thereon  up to constant positive factors. The same fact holds for the right-invariant  measure. These are the left and the right {\bf Haar measures}. If the Lie group $G$ admits a left (right) invariant metric, the induced volume form must coincide (up to positive constant factors) with the left (resp. right) Haar measure, since the volume form satisfies the requirements above.  In the general case of  topological groups, these two measures coincide  
if and only if the group is \textit{unimodular}. In particular compact topological (Lie in particular) groups and \textit{semi-simple Lie groups} are unimodular as is well known \cite{Cohn}.
We have the following more general result due to Milnor.

\begin{proposition} \label{PROP13} If the Lie group $G$ is connected, the  unimodularity condition of $G$ 
is equivalent to the condition
$$tr(ad(X))=0 \quad \forall X \in \mathfrak{g}\, ,$$
which, in terms of structure constants, reads
${c_{ik}}^k=0$.
\end{proposition}
\begin{proof}
    See Appendix \ref{APPENDIXPROOFS}.
\end{proof}
\begin{corollary} \label{remHAAR} If $G$ admits a bi-invariant Riemannian metric
then it is unimodular  and thus the volume form $\mu_{\boldsymbol{g}}$ induced by the metric coincide, up to positive constant
factors, with the (bi-invariant) Haar measure $\mu_G$ on $G$.
\end{corollary}
\begin{proof}
    See Appendix \ref{APPENDIXPROOFS}.
\end{proof}

We now discuss the definition and the relevant properties of Laplace-Beltrami and Casimir operators on $G$. From now on we shall
interpret the vector fields on $G$ as smooth differential operators. In particular, a left invariant vector field $\widetilde{X}$ defined by an element  $X\in \mathfrak{g}$ is therefore a differential operator $\widetilde{X}: C^\infty(G;\mathbb{C})\to C^\infty(G;\mathbb{C})$.\\

\begin{definition}
Consider 
a Lie group $G$  equipped with a bi-invariant Riemannian metric $\boldsymbol{g}$.
Given a basis  $X_1,\ldots, X_{d}$  of $\mathfrak{g}$, define the $g_e^{ab}$ as the coefficients of the inverse of the matrix of coefficients $(g_e)_{ab}:= \boldsymbol{g}(X_a,X_b)$.\\
The second order differential operator
\begin{align}
X_{\boldsymbol{g}}^2 := g_e^{ab}\widetilde{X}_a \widetilde{X}_b : C^\infty(G;\mathbb{C})\to C^\infty(G;\mathbb{C}) \label{CASIMIR}
\end{align}
is called the (second order) {\bf Casimir operator of $(G,\boldsymbol{g})$}.
\end{definition}

Notice that the inverse metric is taken at $e$. It is not difficult to prove that the definition is intrinsic, i.e., it does not depend on the chosen basis of $\mathfrak{g}$. There exists a more abstract definition of $X_{\boldsymbol{g}}^2$ (see, e.g. \cite{Barut-Raczka-86}) based on the notion of  \textit{universal enveloping algebra}, but the above concrete definition is sufficient for the goals of this work.

\begin{proposition} \label{PROPCASIMIR} Assuming that the Lie group $G$ is equipped with a bi-invariant metric $\mathfrak{g}$,
the associated Casimir operator (\ref{CASIMIR}) satisfies
\begin{align}
X_{\boldsymbol{g}}^2  \widetilde{Y} = \widetilde{Y} X_{\boldsymbol{g}}^2\:, \quad \forall \widetilde{Y} \in \mathfrak{g}^L\:.
\end{align}
\end{proposition}

\begin{proof} See Appendix \ref{APPENDIXPROOFS}.
\end{proof}

\begin{remark} The operator $X_{\boldsymbol{g}}^2$ can be defined also if the metric is not bi-invariant. In this case however Proposition \ref{PROPCASIMIR} does not hold in general.
\end{remark}

The metric $\boldsymbol{g}$ on $G$ permits to define another important second-order differential operator, the {\bf Laplace-Beltrami} operator, in local coordinates,
\begin{align}
\Delta_{\boldsymbol{g}} f= g^{ab}\nabla_a (\mathrm{d}f)_b\:, \quad f \in C^\infty(G;\mathbb{C}) \label{LB}
\end{align}
where $\nabla$ is the Levi-Civita connection associated to $\boldsymbol{g}$.

The question arising at this juncture concerns the interplay of $X_{\boldsymbol{g}}^2$ (defined by (\ref{CASIMIR}) also if $\mathfrak{g}$ is not bi-invariant) and $\Delta_{\boldsymbol{g}}$. We now prove  that these two operators coincide if $\boldsymbol{g}$ is bi-invariant.

\begin{proposition}\label{PROPX2DELTA}
Let us consider a Lie group $G$ equipped with a bi-invariant Riemannian metric $\boldsymbol{g}$. The Casimir operator $X_{\boldsymbol{g}}^2$ (\ref{CASIMIR}) and the Laplace-Beltrami operator $\Delta_{\boldsymbol{g}}$  (\ref{LB}) associated to $\boldsymbol{g}$ satisfy 
$X_{\boldsymbol{g}}^2 = \Delta_{\boldsymbol{g}}$.
\end{proposition}

\begin{proof} See Appendix \ref{APPENDIXPROOFS}.
\end{proof}

\begin{corollary} \label{COROLLARIO12}  
Under the assumptions of Proposition \ref{PROPX2DELTA}, the Laplace-Beltrami operator $\Delta_{\boldsymbol{g}}$ commutes with both  all left-invariant vector fields and all right-invariant vector fields viewed as smooth differential operators.
\end{corollary}

\begin{proof} See Appendix \ref{APPENDIXPROOFS}.
\end{proof}

Let us consider now the self-adjointness properties of $X_{\boldsymbol{g}}^2$ in the natural Hilbert space provided by the structure $(G,{\boldsymbol{g}})$ where $\boldsymbol{g}$ is a bi-invariant Riemannian metric.

It is clear that, on a smooth Riemannian manifold $(M,{\boldsymbol{g}})$, the Laplace-Beltrami operator  $\Delta_{\boldsymbol{g}}$ is symmetric if defined as $\Delta_{\boldsymbol{g}}: C_c^\infty(M;\mathbb{C}) \to L^2(M, \mu_{\boldsymbol{g}})$ and $\mu_{\boldsymbol{g}}$ being the volume form induced by the metric. This leads to the following result.\\

\begin{proposition}\label{PROPESA}  Consider a Lie group  $G$ equipped with a bi-invariant Riemannian metric ${\boldsymbol{g}}$.  Then $X^2_{\boldsymbol{g}} = \Delta_{\boldsymbol{g}}: C_c^\infty(G;\mathbb{C}) \to L^2(G, \mu_G)$ is essentially selfadjoint, where $\mu_G$ indicates the  Haar measure.
\end{proposition}
\begin{proof}
See Appendix \ref{APPENDIXPROOFS}.
\end{proof}

Let us consider now the spectrum of $\overline{\Delta_{\boldsymbol{g}}}$, where the bar henceforth  denotes the closure in the Hilbert space $L^2(G, \mu_G)$, i.e., the unique selfadjoint extension of $\Delta_{\boldsymbol{g}}$ with domain $C_c^\infty(G;\mathbb{C})$ on account of Proposition \ref{PROPESA}.\\

\begin{proposition} \label{PROPSPECD} Consider a compact Lie group $G$ equipped with a bi-invariant Riemannian metric ${\boldsymbol{g}}$. The spectrum\footnote{Notice that $\overline{aA+bI}=a\overline {A} +bI$.} of $-\overline{\Delta_{\boldsymbol{g}}}$ is positive and discrete. More precisely,
\begin{itemize}
\item[(1)] $\sigma(-\overline{\Delta_{\boldsymbol{g}}})$ is a countably infinite set of reals $0 < \lambda_0 < \lambda_1 < \ldots < \lambda_n \to +\infty$ as $n\to +\infty$.
\item[(2)]  Every eigenspace $H_\lambda$, where $\lambda \in \sigma(-\overline{\Delta_{\boldsymbol{g}}})$, has finite dimension $d_\lambda$.
\item[(3)]  Every eigenspace $H_\lambda$ is made of $C^\infty$ functions so that the eigenvectors of $-\overline{\Delta_{\boldsymbol{g}}}$ are also eigenfunctions of the differential operator  $-\Delta_{\boldsymbol{g}}$.
\item[(4)]
The orthogonal Hilbert decompositions hold
\begin{align}&L^2(G, \mu_G) = \bigoplus_{\lambda \in \sigma(-\overline{\Delta_{\boldsymbol{g}}})} H_\lambda\:,\label{DECH}\\
&-\overline{\Delta_{\boldsymbol{g}}} f= \sum_{n=0}^{+\infty} \lambda_n P_nf\:, \quad \qquad \forall f \in D(\sigma(-\overline{\Delta_{\boldsymbol{g}}}))\:, \label{DECH2}
\end{align}
where $P_n$ is the orthogonal projector on $H_{\lambda_n}$ and
\begin{align}D(-\overline{\Delta_{\boldsymbol{g}}})= \left\{f\in L^2(G,\mu_G) \:\left|\: \sum_n \lambda_n^2 \|P_{n}f\|^2<+\infty\right.\right\}\:.\end{align}
\item[(5)] For  every $z \in \mathbb{C}$ with $Re(z) \geq 0$,  the operator $e^{-z \overline{\Delta_{\boldsymbol{g}}}}$ --
defined by spectral calculus -- belongs to $\mathfrak{B}(L^2(G,\mu_G))$. Moreover, it  is  compact and trace-class for $Re(z)>0$.
\end{itemize}
 
\end{proposition}
\begin{proof}
See Appendix \ref{APPENDIXPROOFS}.
\end{proof}

If $A: D(A) \to H$ is an operator (generally unbounded) in a Hilbert space $H$, a vector $\psi \in H$ is said to be {\bf analytic} for $A$ if $\psi \in \bigcap_{n\in\mathbb{N}} D(A^n)$ and $\sum_{n=0}^{+\infty} t^n\|A^n\psi\|/n!$ converges for some $t>0$. From spectral calculus (see, e.g., \cite{Moretti-17}) it follows that $e^{s A} \psi = \sum_{n=0}^{+\infty} \frac{s^n}{n!}A^n\psi$ for every complex $s$ with $|s|<t$. Finite linear combinations of eigenvectors of operators are automatically analytic.
Therefore the eigenvectors of $-\overline{\Delta_{\boldsymbol{g}}}$ and their finite linear combinations are obviously analytic vectors of that operator.

We eventually discuss the relation between the eigenspaces of $\overline{\Delta_{\boldsymbol{g}}}$ and the representation of $G$. We start by  recalling a well known definition that will play a central role in our construction.
\begin{definition}\label{def-right-regular-representation} If $G$ is a locally-compact topological group and $\mu_G$ is the right-invariant Haar measure,
then the right-action $G\ni x \mapsto R_x\in Aut(G)$ provides a \textit{strongly continuous unitary} representation  of $G$ on $L^2(G,\mu_G)$, called the {\bf right regular representation}:
\begin{align} G \ni x \mapsto \pi_R(x) \in \mathfrak{B}(L^2(G,\mu_G)) \quad \mbox{where $[\pi_R(x)f](y)=f(yx)$ for $f\in L^2(G,\mu_G)$}\:.\label{PIR} \end{align}
\end{definition}

The fact that $\pi_R$ is a unitary  strongly continuous  representation easily arises from the given definition. If $G$ is compact the  \textit{Peter-Weyl theorem} implies that the Hilbert space decomposes into a Hilbert sum of $\pi_R$-invariant and irreducible subspaces of finite dimension. This result, in the case of a Lie group equipped with a bi-invariant metric, can be proved autonomously providing also further information about the structure of $\pi_R$ and its interplay with the Casimir operator, i.e., the Laplace-Beltrami operator. 

\begin{proposition}\label{PROPRAPR}  Consider a compact Lie group $G$ equipped with a bi-invariant Riemannian metric ${\boldsymbol{g}}$. 
Then:
\begin{itemize}
\item[(1)] The finite-dimensional eigenspaces $H_\lambda$ of $\overline{\Delta_{\boldsymbol{g}}}$ in the Hilbert decomposition (\ref{DECH}) are invariant under the action of $\pi_R$, which correspondingly decomposes into finite-dimensional subrepresentations as 
\begin{align} \pi_R=\bigoplus_{\lambda\in\sigma(-\overline{\Delta_{\boldsymbol{g}}})}\pi_R^\lambda\,, 
\quad \mbox{where}\quad \pi_R^\lambda(x) := \pi_R(x)|_{H_\lambda} \to H_\lambda\quad \forall x\in G\:.\label{OPLUSP}\end{align}
In turn, every subspace $H_\lambda$ is a finite orthogonal sum of $\pi_R$-invariant and irreducible finite-dimensional subspaces of $L^2(G,\mu_G)$.
\item[(2)]  Take $X \in \mathfrak{g}$ and let $\mathbb{R} \ni t \mapsto \exp(tX)$ be the generated strongly-continuous unitary one-parameter subgroup, so that (due to Stone's theorem)  $$\mathbb{R} \ni t \mapsto \pi_R(\exp(tX)) = e^{-i t\: X^R}\:,$$ for a unique selfadjoint operator $X^R : D(X^R) \to L^2(G,\mu_G)$. Then,
\begin{itemize}
\item[(a)] $H_\lambda \subset D(X^R)$ and $X^R(H_\lambda)\subset H_\lambda$ for every $\lambda \in \sigma(-\overline{\Delta_{\boldsymbol{g}}})$ and $X\in \mathfrak{g}$,
\item[(b)] $-iX^R|_{H_\lambda} = \widetilde{X}|_{H_\lambda}$,
\item[(c)] $e^{-i t\: X^R} = \bigoplus_{\lambda \in \sigma(-\overline{\Delta_{\boldsymbol{g}}})} e^{-i t\: X^R|_{H_\lambda}}$ if $t\in \mathbb{R}$.
\end{itemize}
As a consequence, the vectors of each space $\bigoplus_{\lambda \in \Lambda} H_\lambda$, for $\Lambda \subset \sigma(-\overline{\Delta_{\boldsymbol{g}}})$ bounded, are {\bf analytic
vectors of} $\pi_R$, i.e., 
analytic vectors of  $X^R$  for every $X\in \mathfrak{g}$.

\item[(3)] If $z\in \mathbb{C}$, with  $Re\, z \geq 0$, then  $e^{-z \overline{\Delta_{\boldsymbol{g}}}}$ commutes with $\pi_R$ and leaves invariant every $H_\lambda$, where trivially $e^{-z \overline{\Delta_{\boldsymbol{g}}}}|_{H_\lambda}= e^{-z\lambda}I|_{H_\lambda}$.
\end{itemize}
\end{proposition}

\begin{proof}
See Appendix \ref{APPENDIXPROOFS}.
\end{proof}

We finally present a crucial but less known result based on general harmonic analysis on compact (Lie) groups \cite{RT}. Some definitions are necessary.  If $G$ is a compact Lie group, in the following
$C(G)$ denotes the commutative unital $C^*$ algebra of the complex continuous functions on $G$ with natural pointwise operations and norm $\|\cdot\|_\infty$.
It is clear that $C(G)$ is dense in $L^2(G,\mu_G)$.

\begin{definition}\label{Def: finite energy functions}
Consider a compact Lie group $G$ equipped with a bi-invariant Riemannian metric ${\boldsymbol{g}}$. If $H_\lambda$ denotes  the (finite dimensional) $\lambda$-eigenspaces of $-\overline{\Delta_{\boldsymbol{g}}}$, the space of {\bf finite-energy} vectors $F_G\equiv Span\{H_\lambda\:|\: \lambda \in \sigma(-\overline{\Delta_{\boldsymbol{g}}}) \}$ is the dense subspace of $L^2(G,\mu_G)$
  made of \textit{finite} complex linear combinations of  elements of the spaces $H_\lambda$. \\
\end{definition}
In view of the previous results, $F_G$ is made of smooth functions which are analytic vectors for $\overline{\Delta_{\boldsymbol{g}}}$ and  $\pi_R$. Furthermore $F_G$ is invariant under  $\pi_R$, every generator $X^L$, $\overline{\Delta_{\boldsymbol{g}}}$ and the  one parameter groups of unitaries it generates. We have also the following interesting fact. 
\begin{proposition}\label{PROPALG}  Referring to Definition \ref{Def: finite energy functions},  $F_G$ is a $\|\cdot\|_\infty$-dense unital subalgebra of $C(G)$.
\end{proposition}
\begin{proof} 
See Appendix \ref{APPENDIXPROOFS}.
\end{proof}

\section{Feynman path integrals on a Riemannian manifold}\label{sez-Fey-M}
In the present section we introduce the definition of the {\it Feynman map}, i.e. the linear functional providing a mathematical definition of the Feynman path integral \eqref{feynman}.
\subsection{The Feynman map for  $\mathbb{R}^d$ }
\label{Subsec: the Feynman map for Rd and oscillatory integrals}.

For fixed $t>0$ and $d\in \bN$, let us consider the {\bf Cameron Martin space} $H_t(\mathbb{R}^d)$, i.e. the Sobolev space\footnote{The $L^2$ condition $\int_0^t \|\gamma(t)\|^2 dt<+\infty$ is automatically satisfied since the considered curves are (absolutely) continuous functions defined on the compact set $[0,t]$.}
  of absolutely continuous paths $\gamma :[0,t]\to\mathbb{R}^d$ such that  $\gamma(0)=0$ and $\int_0^t\|\dot\gamma(s)\|^2ds<\infty$, where $\dot \gamma $ denotes the weak derivative of $\gamma$.  $H_t(\mathbb{R}^d)$  is actually an Hilbert space with the   inner product
$$\langle \gamma _1,\gamma _2\rangle_{H_t(\mathbb{R}^d)} =\int ^t _0 \dot\gamma _1(s)\cdot\dot\gamma
  _2(s)ds \, .$$
 
Fixed $n\in \mathbb{N}$, consider the equally spaced partition $0=t_0<t_1<\dots<t_n=t$ of the interval $[0,t]$, where $t_j=jt/n$, $j=0,\dots, n$. Let $P_n:H_t(\mathbb{R}^d)\to H_t(\mathbb{R}^d)$ be the orthogonal projection operator onto the subspace $P_nH_t(\mathbb{R}^d)\simeq\mathbb{R}^{nd}$ of piecewise linear paths, i.e. those paths having constant velocity when restricted to the partition subintervals:  
$$P_n\gamma (s):=\gamma(t_{j-1})+\frac{\gamma(t_{j})-\gamma(t_{j-1})}{t_{j}-t_{j-1}}(s-t_{j-1})\, , \qquad t_{j-1}\leq s\leq t_{j}\, .$$

Denoting by $\mathrm{d}P_n\gamma $ the Borel volume measure on $P_n H_t(\mathbb{R}^d)$ associated to the metric inherited by $H_t(\mathbb{R}^d)$ itself and considering a complex map $f:H_t(\mathbb{R}^d)\to \mathbb{C}$, let us define  the finite-dimensional \textit{oscillatory integral} (with henceforth $i^{1/2}:= e^{i\pi/4}$)
\begin{equation}\label{FDOI-1}
   (2\pi i \hbar )^{-nd/2}\int_{ P_n H_t(\mathbb{R}^d)} e^{\frac{i}{2\hbar}\|\gamma\|_{H_t(\mathbb{R}^d)}^2}f_{P_n}(\gamma) \mathrm{d}P_n\gamma\, ,
\end{equation}
where $f_{P_n}:P_n H_t(\mathbb{R}^d)\to \mathbb{C}$ stands for the restriction of $f$ to $P_n H_t(\mathbb{R}^d)$. As it stands, the integral above does not always make sense if interpreted in the standard way.

Following H\"ormander \cite{Hor1} and Elworthy and Truman \cite{ELT}, it is meaningful to provide a definition of the integral \eqref{FDOI-1} allowing the integration of functions $f\colon\mathbb{R}^{nd}\to\mathbb{C}$ that do not necessarily belong to  $L^1(\mathbb{R}^{nd}, \mathrm{d}x)$, where $\mathrm{d}x$ is the Borel measure  on $\mathbb{R}^{nd}$ induced by the standard metric\footnote{The completion of that Borel measure is the standard \textit{Lebesgue measure} in $\mathbb{R}^{nd}$.}.
(The integral in \eqref{FDOI-1} can be reduced to this setting by considering an orthonormal basis of $P_nH_t(\mathbb{R}^d)$ so that $P_nH_t(\mathbb{R}^d)$ and $\mathrm{d}P_n$ can be identified with $\mathbb{R}^{nd}$ and $\mathrm{d}x$ respectively.)
More specifically,  we shall adopt the following definition \cite{AlBr,ELT,Ma-book}.
\begin{definition}\label{def-osc-int-findim}

Let $n\geq 1$ be a natural number and let  $f:\mathbb{R}^n\to \mathbb{C}$, $\Phi:\mathbb{R}^n\to \mathbb{R}$ be Borel measurable functions.
We say that 
$\int_{\mathbb{R}^n} e^{i\Phi(x)}f(x)\mathrm{d}x$
exists as an  \textbf{oscillatory integral} if
\begin{itemize}
    \item for every Schwartz test function  $\varphi \in S(\mathbb{R}^{n})$ such that $\varphi (0)=1$,  the regularized integral
$$\int_{\mathbb{R}^n} e^{i\Phi(x)}f(x)\varphi (\epsilon x)\mathrm{d}x$$ exists and it is finite  for all $\epsilon>0$;
\item
the limit
$$ \lim_{\epsilon \downarrow 0} \int_{\mathbb{R}^n} e^{i\Phi(x)}f(x)\varphi (\epsilon x)\mathrm{d}x\,, $$
exists, it is finite and independent of $\varepsilon$.
\end{itemize}
In this case the limit is denoted \begin{equation}\label{int-osc-fin-dim}
    \int_{\mathbb{R}^n}^o e^{i\Phi(x)}f(x)\mathrm{d}x\, .
\end{equation}
\end{definition}
Clearly, when $f\in L^1(\mathbb{R}^n, \mathrm{d}x)$ the oscillatory integral \eqref{int-osc-fin-dim} coincides with the standard (Lebesgue) integral  $\int_{\mathbb{R}^n} e^{i\Phi(x)}f(x)\mathrm{d}x$. 

In the following, we shall adopt the shortened notation $I_n(f)$ for the normalized oscillatory integral on $ P_n H_t(\mathbb{R}^d)$.
\begin{equation}
   I_n(f)=(2\pi i \hbar )^{-nd/2}\int^o_{ P_n H_t(\mathbb{R}^d)} e^{\frac{i}{2\hbar}\|\gamma\|^2_{H_t(\mathbb{R}^d)}}f_{P_n}(\gamma) \mathrm{d}P_n\gamma\, ,
\end{equation}
In particular, by direct inspection and thanks to the introduction of the normalizing constant $(2\pi i \hbar )^{-nd/2}$, we have $I_n(\mathbf{1})=1$, where $\mathbf{1}$ denotes the constant function equal to 1.
Indeed, for $\varphi\in S(\mathbb{R}^n)$ Parseval equality yields:
\[\int_{\mathbb{R}^n}\frac{e^{\frac{i}{2\hbar}\|x\|^2}}{(2\pi i \hbar )^{nd/2}}\varphi(\epsilon x)\mathrm{d}x =\int_{\mathbb{R}^n}e^{-\frac{i\hbar}{2}\epsilon^2\|x\|^2}\hat\varphi( x)\mathrm{d}x\,,\]
where $\varphi(x)=\int_{\mathbb{R}^n}e^{ixy}\hat\varphi (y)\mathrm{d}y$. By taking the limit for $\epsilon\downarrow 0$ and exploiting the condition $\varphi(0)=\int_{\mathbb{R}^n}\hat\varphi (y)\mathrm{d}y=1$ we eventually get the normalization property $I_n(\mathbf{1})=1$.
\begin{remark}
    By adopting on $P_n H_t(\mathbb{R}^d)$ the coordinates 
    \begin{equation}\label{coord-1}
        x_j:=\gamma(t_j)\, ,\qquad j=1,\dots , n
    \end{equation} (with $x_0=0$) and introducing the notation $\gamma_{\boldsymbol{x}}$, $\boldsymbol{x}=(x_1,\ldots,x_{n})$, for the path $\gamma \in P_n H_t(\mathbb{R}^d)$ satifying conditions \eqref{coord-1},  the finite dimensional oscillatory integral \eqref{FDOI-1} assumes the following form
    \begin{equation}\label{rap-coord-1}
        I_n(f)=(2\pi i \hbar t/n )^{-nd/2}\int_{\mathbb{R}^{nd}}^o
        e^{\frac{i}{2\hbar t/n}\sum_{j=1}^n\|x_j-x_{j-1}\|^2} f_X(x_1,\dots, x_{n})\mathrm{d}x_1\cdots \mathrm{d}x_{n}\, ,
    \end{equation}
    where $f_X:(\mathbb{R}^d)^n \to \mathbb{C}$ is defined as  $  f_X(x_1,\dots, x_{n}):= f(\gamma_{\boldsymbol{x}})$.
    
    Similarly, by using as coordinates the velocities $v_j$, $j=1,\dots, n$,
    \begin{equation}\label{coord-2}
        v_j:=\frac{\gamma(t_j)-\gamma(t_{j-1})}{t_j-t_{j-1}}\, ,\qquad j=1,\dots , n
    \end{equation} and introducing the notation $\gamma_{\boldsymbol{v}}$, $\boldsymbol{v}=(v_1,\ldots, v_n)$, for the path $\gamma \in P_n H_t(\mathbb{R}^d)$ satifying conditions \eqref{coord-2}, and $f_V$ for the mapping $f_V:\mathbb{R}^{nd} \to \mathbb{C}$ defined as  $f_V(v_1,\dots, v_n):=f(\gamma_{\boldsymbol{v}})$, the finite dimensional oscillatory integral \eqref{FDOI-1} can be computed as:
    \begin{equation}\label{rap-coord-2}
        I_n(f)=(2\pi i \hbar n/t )^{-nd/2}\int^o_{\mathbb{R}^{nd}}e^{\frac{i}{2\hbar}\frac{t}{n}\sum_{j=1}^n\|v_j\|^2} f_V(v_1,\dots, v_n)\mathrm{d}v_1\cdots \mathrm{d}v_n\, .
    \end{equation}
    
\end{remark}
We are now ready to give the main definition of this section.
\begin{definition}
Given a function $f:H_t(\mathbb{R}^d)\to \mathbb{C}$ we  define its {\bf Feynman map} $F_{H_t(\mathbb{R}^d)}(f)$ as the limit of the finite dimensional oscillatory integrals  
\begin{equation}\label{FMap-1}
    F_{H_t(\mathbb{R}^d)}(f):=\lim_{n\to \infty} I_n(f)\, ,
\end{equation}
provided that all the terms  of the approximating sequence $I_n(f)$ are well defined and the limit \eqref{FMap-1} exists in $\mathbb{C}$.\\
\end{definition}
An alternative more suggestive notation for the Feynman map $F_{H_t(\mathbb{R}^d)}(f)$ is the following \cite{AlHKMa,AlBr}
\begin{equation}\label{not-FM1}
    F_{H_t(\mathbb{R}^d)}(f)\equiv \widetilde{\int}_{H_t(\mathbb{R}^d)}e^{\frac{i}{2\hbar}\int_0^t|\dot\gamma (s)|^2ds}f(\gamma)\mathrm{d}\gamma\,. 
    \end{equation}
\begin{remark}
By using representation \eqref{rap-coord-1} for $I_n(f)$, the Feynman map can be actually computed as:
\begin{equation}\label{FMap-2}
    F_{H_t(\mathbb{R}^d)}(f)=\lim_{n\to \infty} (2\pi i \hbar t/n )^{-nd/2}\int^o_{\mathbb{R}^{nd}}e^{\frac{i}{2\hbar t/n}\sum_{j=1}^n\|x_j-x_{j-1}\|^2} f_X(x_1,\dots, x_{n})\mathrm{d}x_1\cdots \mathrm{d}x_{n}\,  ,
\end{equation}
resembling Feynman's original construction. At the same time, representation \eqref{rap-coord-2} leads to the equivalent form
\begin{equation}\label{FMap-3}
    F_{H_t(\mathbb{R}^d)}(f)=\lim_{n\to \infty} (2\pi i \hbar (t/n)^{-1} )^{-nd/2}\int^o_{\mathbb{R}^{nd}}e^{\frac{i}{2\hbar}\frac{t}{n}\sum_{j=1}^n\|v_j\|^2} f_V(v_1,\dots, v_n)\mathrm{d}v_1\cdots \mathrm{d}v_n\, .
    \end{equation}
    \end{remark}

    \begin{remark}
        By considering an initial datum $\psi_0$ and a potential $V$ of the form
        $$\psi_0(x)=\int_{\bR^d}e^{ixy}\mathrm{d}\mu_0(y)\,, \quad V(x)=\int_{\bR^d}e^{ixy}\mathrm{d}\mu_V(y)\,, \qquad x\in \bR^d\,,$$
        for some bounded complex Borel measures $\mu_0$ and $\mu_V$ on $\bR^b$, it is possible to prove \cite{AlBr,AlHKMa,ELT} that the Feynman map of the function $f:H_t(\mathbb{R}^d)\to \bC$ defined as
        $$f(\gamma):=e^{-\frac{i}{\hbar}\int_0^tV(\gamma(s))\mathrm{d}s}\psi_0(\gamma(t))$$
        provides the solution $\psi(t,x)$ of the Schr\"odinger equation \eqref{Schroedinger}. By adopting the notation \eqref{not-FM1} for $ F_{H_t(\mathbb{R}^d)}(f)$ one formally obtains the original Feynman's heuristic formula, namely
        $$\psi(t,x)=\widetilde{\int}_{H_t(\mathbb{R}^d)}e^{\frac{i}{\hbar}\int_0^t\left(\frac{1}{2}|\dot\gamma (s)|^2-V(\gamma(s))\right)\mathrm{d}s}\psi_0(\gamma(t))\mathrm{d}\gamma\,.$$
    \end{remark}
    
\subsection{The Feynman map for Riemannian manifolds }
\label{SECCONJ}
Let us consider now an arbitrary $d$-dimensional Riemannian manifold $(M,\boldsymbol{g})$. 
For $t\in\mathbb{R}$ and $x\in M$, let us denote by $C_x(M)$ the set of continuous paths on $M$ starting at $x$:
$$C_x(M):=\{\gamma:[0,t]\to M\,|\, \gamma \:\:\hbox{continuous and }\gamma (0)=x \}\, . $$
Similarly, with the symbol $C(T_xM)$ we shall denote the real vector space
$$C(T_xM):=\{\Gamma:[0,t]\to T_xM\,|\, \Gamma \:\:\hbox{continuous and }\Gamma (0)=0 \}\, . $$
Let us also  consider the vector  space ${\mathcal{H}}_{0,t}(T_xM)\subset C(T_xM)$ defined by:
\begin{align}
    {\mathcal{H}}_{0,t}(T_xM)
    :=\{\Gamma\in C(T_xM)\,|\,
    \Gamma\:\: \hbox{absolutely continuous}, 
    \int_0^t\boldsymbol{g}_x(\dot{\Gamma}(s),\dot{\Gamma}(s))\mathrm{d}s
    <+\infty\}\,,
\end{align}
where $\dot{\Gamma}$ denotes the weak derivative of $\Gamma$. 
${\mathcal{H}}_{0,t}(T_xM)$, endowed with the inner product $\langle \, , \, \rangle_{\mathcal{H}_{0,t}(T_xM)}$ defined by
\begin{align}
    \langle\Gamma_1,\Gamma_2\rangle_{\mathcal{H}_{0,t}(T_xM)}
    :=\int_0^t\boldsymbol{g}_x(\dot{\Gamma}_1(s),\dot{\Gamma}_2(s))\mathrm{d}s\,,
\end{align}
is a real Hilbert space \cite{AlHKMa,ELT,AnDri}.
In fact, if we identify $T_xM$ with $\mathbb{R}^{\dim(M)}$ through  the choice of a orthonormal basis, then $\mathcal{H}_{0,t}(T_xM)$ turns out to be isomorphic to the Cameron-Martin space 
$H_t(\mathbb{R}^d)$.

\noindent We can analogously define the  set  of finite energy paths on $M$   \cite{Kli1,AnDri}:
\begin{align}\label{HxtM}
    \mathcal{H}_{x,t}(M)
    :=\{\gamma\in C_x(M)\,|\,
    \gamma \hbox{ absolutely continuous}\,,\,
    \int_0^t\boldsymbol{g}(\dot{\gamma}(s),\dot{\gamma}(s))\mathrm{d}s
    <+\infty\}\,,
\end{align}
where $\gamma\in C_x(M)$ is said to be  absolutely continuous if for any $f\in C^\infty (M;\mathbb{R})$, the function $f\circ \gamma :[0,t]\to \mathbb{R}$ is   absolutely continuous. \\
Clearly, $\mathcal{H}_{x,t}(M)$ is not a vector space, but it turns out to be an Hilbert manifold (see, e.g. \cite{Kli1} for further details). Moreover there exists an interesting smooth and one-to-one map between $\mathcal{H}_{x,t}(M)$ and the Hilbert space ${\mathcal{H}}_{0,t}(T_xM)$, as we are going to discuss.
\begin{definition}\label{Def: Cartan map}
For $x\in M$ arbitrary but fixed, the {\bf Cartan development map}\footnote{Sometimes the Cartan development map is defined as the inverse  the function $\Psi_x$ defined by \eqref{def-Cartan-map}.} $\Psi_x$ \cite{KN} maps piecewise smooth curves
$\gamma_x: [0, t] \to M$  starting at $x=\gamma(0)$ to curves
$\Gamma_x:=\Psi_x(\gamma_x)\colon[0,t]\to T_xM$ defined by the requirements
\begin{align}\label{def-Cartan-map}
    \Gamma_x(0)=0\,,
    \qquad
    \dot{\Gamma}_x(s):=\wp[\gamma_x]^s_0\dot{\gamma}_x(s)\:,\quad  s\in [0,t]\:,
\end{align}
where $\wp[\gamma_x]^s_0\colon T_{\gamma_x(s)}M\to T_xM$ denotes the $\boldsymbol{g}$-parallel transport along $\gamma_x$ while we identified $T_{\Gamma_x(s)}T_xM\simeq T_xM$ for all $s\in[0,t]$.
\end{definition}

\begin{remark}\label{REMCM}
From the definition and the properties of the parallel transport associated with the Levi-Civita connection, it follows that: 
\begin{itemize}
\item[(1)] if $\gamma_x$ is a piecewise $\boldsymbol{g}$-geodesic, then $\Gamma_x$ is a piecewise straight line;
\item[(2)] if $\Gamma_x=\Psi_x(\gamma_x)$, then $ \int_0^t\boldsymbol{g}_x(\dot{\Gamma}_x(s),\dot{\Gamma}_x(s))\mathrm{d}s=\int_0^t\boldsymbol{g}(\dot{\gamma}(s),\dot{\gamma}(s))\mathrm{d}s$.
\end{itemize}
\end{remark}
On account of the  properties of $\Psi_x$ we have the following result whose proof can be found in \cite{EE,AnDri}.
\begin{proposition}\label{PROPonetoONE}
The Cartan map $\Psi_x$ uniquely continuously extends to a  one-to-one map (indicated with the same symbol) from $\mathcal{H}_{x,t}(M)$  to $ {\mathcal{H}}_{0,t}(T_xM)$ which is a diffeomorphism of infinite dimensional Hilbert manifolds. In particular,
$$ \int_0^t\boldsymbol{g}_x(\dot{\Gamma}_1(s),\dot{\Gamma}_2(s))\mathrm{d}s = \int_0^t\boldsymbol{g}(\dot{\gamma}_1(s),\dot{\gamma}_2(s))\mathrm{d}s$$
for $\Gamma_i=\Psi_x(\gamma_i)$, $\gamma_i\in\mathcal{H}_{x,t}(M)$, $i=1,2$.
\end{proposition}

The previous result suggests a natural way to rely the construction of the Feynman map of a function $f: \mathcal{H}_{x,t}(M)\to \mathbb{C}$ upon the Feynman map of the associated function $f\circ\Psi_x^{-1}: \mathcal{H}_{0,t}(T_xM)\to \mathbb{C}$ \cite{ELT,ELT81}. Indeed, according to Proposition \ref{PROPonetoONE},
the Cartan map $\Psi_x$ sends suitably regular curves $\gamma_x$ on $M$ starting at $x$ to suitably regular curves $\Gamma_x:=\Psi_x(\gamma_x)$ on $T_xM$ starting at $0$.
In particular, having fixed an orthonormal basis of $T_xM$ and having identified $T_xM$ with $\mathbb{R}^d$, $\Psi_x$ maps $\mathcal{H}_{x,t}(M)$ onto $H_t(\mathbb{R}^d)$. In the following we shall denote $\Psi_x^{-1}:H_t(\mathbb{R}^d)\to \mathcal{H}_{x,t}(M)$ its inverse.

\begin{definition}\label{def-Fey-map-M}
Given a function $f:\mathcal{H}_{x,t}(M) \to \mathbb{C}$, we shall define its {\bf Feynman map} $F_{\mathcal{H}_{x,t}(M)}(f) $ as
\begin{equation}\label{Def-FM-M}
    F_{\mathcal{H}_{x,t}(M)}(f) :=F_{H_t(\mathbb{R}^d)}(f\circ \Psi_x^ {-1}) \:.
\end{equation}
  \end{definition}
\begin{remark}
    In fact, the Cartan development map and its extension to $C_x(T_xM)$ called {\em stochastic development} play a fundamental role  in the construction of Brownian motion on a Riemannian manifold. See, e.g. \cite{EE,AnDri} for further details.
\end{remark}
By adopting the heuristic but suggestive notation \eqref{not-FM1}, the Feynman map $ F_{\mathcal{H}_{x,t}(M)}(f) $ can be denoted as:
\begin{equation}\label{not-FM2}
F_{\mathcal{H}_{x,t}(M)}(f) \equiv \widetilde{\int}_{\mathcal{H}_{x,t}(M)} e^{\frac{i}{2\hbar} \int_0^t\boldsymbol{g}(\dot{\gamma}(s),\dot{\gamma}(s))\mathrm{d}s   } f(\gamma)\mathrm{d}\gamma\,.\end{equation}
 The authors of \cite{ELT81} conjectured what follows.\\
 
\noindent {\bf Conjecture}. \textit{Consider a smooth  Riemannian manifold $(M,{\boldsymbol{g}})$ and a function $f:\mathcal{H}_{x,t}(M) \to \mathbb{C}$ of the form 
$$f(\gamma)=\psi_0(\gamma (t))e^{-\frac{i}{\hbar}\int_0^tV(\gamma(s))\mathrm{d}s}$$
with $\psi_0:M\to\mathbb{C}$ and $V:M\to \mathbb{R}$ suitably chosen.\\
Then the action $F_{\mathcal{H}_{x,t}(M)}(f)$ of the Feynman map on $f$ provides a representation for  the solution $\psi(t,x)$ of the Schr\"odinger equation \eqref{SchroedingerM}, i.e.
according to the formal notation \eqref{not-FM2}
\begin{equation}\label{con-Fey-M}
    \psi(t,x)=\widetilde{\int}_{\mathcal{H}_{x,t}(M)} e^{\frac{i}{\hbar} \int_0^t\left(\frac{1}{2}\boldsymbol{g}(\dot{\gamma}(s),\dot{\gamma}(s))-V(\gamma(s))\right)\mathrm{d}s   } \psi_0(\gamma (t))\mathrm{d}\gamma\,.\end{equation}}

The conjecture may be understood in the sense of strongly-continuous unitary one-parameter groups, assuming the operator on the right-hand side of (\ref{SchroedingerM}) is essentially selfadjoint in the Hilbert space $L^2(M,\mu_{\boldsymbol{g}})$ in a dense domain which includes the initial datum:
\begin{align} F_{\mathcal{H}_{x,t}(M)}(f) = \left(e^{-\frac{it}{\hbar}(\overline{{ -}\frac{\hbar^2}{2}\Delta_{\boldsymbol{g}} { +}V})}\psi_0\right)(x)\:.\label{CONJ}\end{align}

Notice that each such group is also a \textit{$C_0$-semigroup} for $t\in [0,+\infty)$, and this fact has important consequences in relation to \textit{Chernoff's approximation theorem} which may be used to prove (\ref{CONJ}) on compact Lie groups $G=M$.

\bigskip

The case where $M=\mathbb{R}^d$ has been extensively studied and the conjecture has been proved to be true for a large class of initial data $\psi_0$ and potentials $V$ (see, e.g.,  \cite{Ma-book} and references therein). 
In the following we are going to address conjecture \eqref{con-Fey-M} when $M=G$ is a compact Lie group with a bi-invariant Riemannian metric $\boldsymbol{g}$.
In section \ref{sez-fey-G} the conjecture is proved for the case $V=0$ by restricting to a specific class of functions $f\colon\mathcal{H}_t(G)\to\mathbb{R}$.
In section \ref{sez-Schr-V} we generalize these results to the case of non-vanishing $V$, although in this case the construction conjecture \eqref{con-Fey-M} is proved only at a "perturbative" level.
Here, "perturbative" refers to the fact that we cannot make sense of the Feynman integral of $\exp(-\frac{i}{\hbar}\int_0^tV(\gamma(s))\mathrm{d}s)$ directly. However, we can interpret the Feynman integral of each term in the expansion of $\exp(-\frac{i}{\hbar}\int_0^tV(\gamma(s))\mathrm{d}s)$, and moreover show that the resulting series converges to the desired result.

\section{The Feynman map for compact Lie groups with bi-invariant Riemannian metrics}\label{sez-fey-G}

The goal of this section is to prove the well-definiteness of the Feynman map defined through the Cartan map as explained in Section \ref{SECCONJ} once $M=G$ is a compact Lie group with bi-invariant metric.
Specifically, we will prove that $F_{\mathcal{H}_{x,t}(G)}(f)$ makes sense for a suitable class of functions $f\colon\mathcal{H}_{x,t}(G)\to\mathbb{C}$.
Moreover, within this setting the resulting Feynman map will fulfil Equation \eqref{CONJ}.

\subsection{Finite dimensional approximations of the Feynman map on Lie groups}

In what follows we will consider the case where $M=G$ is a compact Lie group with bi-invariant Riemannian metric $\boldsymbol{g}$.
We shall denote by $\gamma_{x,v}$ the geodesic starting at $x\in G$ and with initial velocity $\widetilde{v}_x$ where $\widetilde{v}\in\mathfrak{g}^L$ is the left-invariant vector field associated with $v\in\mathfrak{g}$ ---we recall that $\gamma_{x,v}(t)=x\exp(t v)$ for all $t\in\mathbb{R}$ on account of Proposition \ref{PROPO6}.
Similarly, given $\boldsymbol{v}=(v_1,\ldots,v_n)$, where $v_1,\ldots,v_n\in\mathfrak{g}$, and $\boldsymbol{t}=(t_1,\ldots,t_n)\in\mathbb{R}^n$, $0\leq t_1\leq\dots\leq t_n$, we will denote by $\gamma_{x,\boldsymbol{v},\boldsymbol{t}}$ the piecewise $\boldsymbol{g}$-geodesic $\gamma_{x,\boldsymbol{v},\boldsymbol{t}}: [0,t] \to G$ starting at $x\in G$ and defined by
\begin{align}\label{gamma-x-v}
    \gamma_{x,\boldsymbol{v},\boldsymbol{t}}(s)
    :=x\prod\limits_{\ell=1}^{j-1}\exp[(t_\ell-t_{\ell-1})v_\ell]
    \exp[(s-t_{j-1})v_j]
    \qquad t_{j-1}\leq s\leq t_j
    \quad\forall j\in\{1,\ldots,n\}\,,
\end{align}
where $t_0:=0$, $t_n:=t$, and the empty product is set to $e$.
The $n$-tuple $\boldsymbol{t}$ contains the time steps at which the piecewise geodesic changes tangent vectors.
Indeed
the tangent vector field $\dot{\gamma}_{x,\boldsymbol{v},\boldsymbol{t}}$ to $\gamma_{x,\boldsymbol{v},\boldsymbol{t}}$, according to (\ref{ISOTILDE}), is given by
\begin{align*}
    \dot{\gamma}_{x,\boldsymbol{v},\boldsymbol{t}}(s)
    =\widetilde{v}_j(\gamma_{x,\boldsymbol{v},\boldsymbol{t}}(s))
    \qquad t_{j-1}\leq s\leq t_j
    \quad\forall j\in\{1,\ldots,n\}\,,
\end{align*}
where the left-invariant vector field $\widetilde{v}_j$ is defined as $\widetilde{v}_j(z)=(\mathrm{d}L_z)_e v_j$.
In what follows we shall always consider an equally spaced $n$-tuple $\boldsymbol{t}$ so that $t_j-t_{j-1}=\delta t=t/n$.
For this reason in the forthcoming discussion we will adopt the shorthand notation $\gamma_{x,\boldsymbol{v}}:=\gamma_{x,\boldsymbol{v},\boldsymbol{t}}$.

\begin{remark}\label{Rmk: initial velocities for Cartan curve associated with piecewise geodesic; hat velocities and velocities are in bijection}
\noindent
\begin{enumerate}[(1)]
\item\label{Item: initial velocities for Cartan curve associated with piecewise geodesic}
Since $\gamma_{x,\boldsymbol{v}}$ is a piecewise geodesic, the curve $\Gamma_{x,\hat{\boldsymbol{v}}}:=\Psi_x(\gamma_{x,\boldsymbol{v}})$ in $T_xG$ is a piecewise straight line with velocities $\hat{\boldsymbol{v}}=(\hat{v}_1,\ldots,\hat{v}_n)$, $\hat{v}_1,\ldots,\hat{v}_n\in\mathfrak{g}$.
The latter are given by {$\hat{v}_j:= \dot{\Gamma}_{x,\hat{\boldsymbol{v}}}(t_{j-1})$} $j\in\{1,\ldots,n\}$.
Explicitly we have for all $j\in\{1,\ldots,n\}$:
\begin{align}\label{Eq: velocities for piecewise straight line}
    \nonumber
    \dot{\Gamma}_{x,\hat{\boldsymbol{v}}}(s)
   & =
    \wp[\gamma_{x,v_1}]^{t_1}_0
    \cdots
    \wp[\gamma_{x_{j-1},v_j}]^{s}_{t_{j-1}}
    \widetilde{v}_j
    (\gamma_{x,\boldsymbol{v}}(s))
    \qquad
    t_{j-1}< s< t_j
    \\
    \hat{v}_j
    &=
    \wp[\gamma_{x,v_1}]^{t_1}_0
    \cdots
    \wp[\gamma_{x_{j-2},v_{j-1}}]^{t_{j-1}}_{t_{j-2}}
    (\mathrm{d}L_{x_{j-1}})_ev_j
\end{align}
where $x_0,\ldots,x_n$ are defined by $x_j=x\exp[t_1v_1]\cdots\exp[t_jv_j]$ with $x_0:=x$.

\item\label{Item: hat velocities and velocities are in bijection}
For later convenience we observe that
$\hat{\boldsymbol{v}}$ is computed from $\boldsymbol{v}$ by applying isometries ---we are considering the natural product metric on $\mathfrak{g}^n$--- namely the differential of the left-action and the parallel transport.
In particular, the map $\boldsymbol{v}\mapsto\hat{\boldsymbol{v}}$ is 1-1.

Furthermore, it is worth to point out that the Jacobian of the map $(v_1,\ldots,v_n)\mapsto(\hat{v}_1,\ldots,\hat{v}_n)$ is one.
Indeed, for each $j\in\{1,\ldots,n\}$ we have that $\hat{v}_j$ depends only on $v_1,\ldots,v_j$.
Therefore, the corresponding Jacobian will be a block-wise lower triangular matrix and its determinant will correspond to the product of the Jacobians of the maps $v_j\mapsto \hat{v}_j$ where $v_1,\ldots,v_{j-1}$ are kept fixed.
Crucially, by direct inspection the latter Jacobians are identically equal to 1, since they are associated to isometries. Indeed, for any $j$ the map  $v_j\mapsto \hat{v}_j$ preserves inner products 
 as it is the composition of the differential of the right translation $\mathrm{d}L_{x_{j-1}}$  and the parallel transport of the resulting vector $\mathrm{d}L_{x_{j-1}}v_j$ along the fixed curve uniquely identified by $x$ and the vectors $v_1, \dots, v_{j-1}$.
\end{enumerate}
\end{remark}

Let us consider a function $f:\mathcal{H}_{x,t}(G)\to \mathbb{C}$ of the form
 $f(\gamma)=\psi_0(\gamma(t))$ for $\psi_0\in C^\infty(G; \mathbb{C})\subset  L^2(G,\mu_G)$.
By the definition of the Feynman map (see  \eqref{Def-FM-M} and \eqref{FMap-1}), with the choice of an orthonormal basis at $x$, we have:
\begin{equation}
    F_{\mathcal{H}_{x,t}(G)}(f)=F_{H_t(\mathbb{R}^{d})}(f\circ \Psi_x^ {-1})
\end{equation}
This identity can be formally and more intuitively written as:
\begin{equation}
   \widetilde{ \int}_{\mathcal{H}_{x,t}(G)} e^{\frac{i}{2\hbar}\|\gamma\|^2_{H_t(\mathbb{R}^{d})}} f(\gamma)\mathrm{d}\gamma=\widetilde{\int}_{H_t(\mathbb{R}^{d})} e^{\frac{i}{2\hbar}\|\Gamma\|^2_{H_t(\mathbb{R}^{d})}}
    f(\Psi_x^{-1}\Gamma)\mathrm{d}\Gamma\,.
\end{equation}
\begin{remark}
	Using the results of Section \ref{Subsec: the Feynman map for Rd and oscillatory integrals}, $F_{H_t(\mathbb{R}^{d})}(f\circ\Psi_x^{-1})$ enjoys the representation formula \eqref{rap-coord-2}.
	In particular
	\begin{align*}
		F_{\mathcal{H}_{x,t}(G)}(f)
		&=\lim_{n\to \infty}
		(2\pi i \hbar (t/n)^{-1} )^{-nd/2}\int_{\mathbb{R}^{nd}}^o
		e^{\frac{i}{2\hbar}\frac{t}{n}\sum_{j=1}^n\|\hat{v}_j\|^2} f(\Psi_x^{-1}\Gamma_{x,\hat{\boldsymbol{v}}})
		\mathrm{d}\hat{v}_1\cdots \mathrm{d}\hat{v}_n\,,
	\end{align*}
	where $\Gamma_{x,\hat{\boldsymbol{v}}}$ is the piecewise linear path on $T_xG$ with velocities $\hat{\boldsymbol{v}}:=(\hat{v}_1,\ldots,\hat{v}_n)$, \textit{cf.} Remark \ref{Rmk: initial velocities for Cartan curve associated with piecewise geodesic; hat velocities and velocities are in bijection}.
    With a slight abuse of notation the integral in the above formula is meant with respect to the components of $\hat{v}_1,\ldots,\hat{v}_n$ once an arbitrary but fixed orthonormal basis $X_1,\ldots,X_d\in\mathfrak{g}$ is fixed.
	However, in what follows we shall profit from a slightly different representation formula which is more useful for our purposes.
	In particular we will parametrize $\Gamma\in P_n H_{t}(\mathbb{R}^{d})$ in terms of the parameters $v_1,\ldots,v_n\in\mathfrak{g}$ of the associated path $\Psi_x^{-1}\Gamma$.
	In more details, any $\Gamma\in P_n H_{t}(\mathbb{R}^{d})$ can be written as $\Gamma=\Psi_x(\gamma_{x,\boldsymbol{v}})$, for a unique piecewise geodesic curve $\gamma_{x,\boldsymbol{v}}$ with parameters $v_1,\ldots,v_n\in\mathfrak{g}$.
	We recall that, by Remark \ref{Rmk: initial velocities for Cartan curve associated with piecewise geodesic; hat velocities and velocities are in bijection}, the map $\boldsymbol{v}\mapsto\hat{\boldsymbol{v}}$ is an isometric bijection.
	In particular we have $\|\Gamma\|_{H_{0,t}(\mathbb{R}^{d})}=\|\gamma_{x,\boldsymbol{v}}\|_{\mathcal{H}_{x,t}(G)}=\frac{t}{n}\sum_{j=1}^n\|v_j\|^2$, thus,
	\begin{align}
		\label{Eq: Feynman integral parametrized by Lie algebra velocities}
		F_{\mathcal{H}_{x,t}(G)}(f)
		&
		=\lim_{n\to \infty}
		(2\pi i \hbar (t/n)^{-1} )^{-nd/2}\int_{\mathbb{R}^{nd}}^o
		e^{\frac{i}{2\hbar}\frac{t}{n}\sum_{j=1}^n\|\hat{v}_j\|^2} f(\Psi_x^{-1}\Gamma_{x,\hat{\boldsymbol{v}}})
		\mathrm{d}\hat{v}_1\cdots \mathrm{d}\hat{v}_n
		\\
		&=\lim_{n\to \infty}
		(2\pi i \hbar (t/n)^{-1} )^{-nd/2}\int_{\mathbb{R}^{nd}}^o
		e^{\frac{i}{2\hbar}\frac{t}{n}\sum_{j=1}^n\|v_j\|^2} f(\gamma_{x,\boldsymbol{v}})
		\mathrm{d}v_1\cdots \mathrm{d}v_n\,,
	\end{align}
        where the Jacobian of the map $(v_1,\ldots,v_n)\mapsto(\hat{v}_1,\ldots,\hat{v}_n)$ is 1 as discussed in remark \ref{Rmk: initial velocities for Cartan curve associated with piecewise geodesic; hat velocities and velocities are in bijection}.
        In the above equation integration over $v_1,\ldots,v_n$ is again intended as integration with respect to the Lebesgue measure of the corresponding components in a fixed orthonormal basis $X_1,\ldots,X_d\in\mathfrak{g}$.
\end{remark}

First of all, taking Proposition \ref{PROPSPECD} into account, let us consider the most elementary case where
\begin{align*}
    f\colon\mathcal{H}_{x,t}(G)
    \to\mathbb{C}
    \qquad
    f(\gamma):=\varphi_\lambda(\gamma(t))\,,
\end{align*}
with $\varphi_\lambda$  an eigenfunction of the  operator $-\Delta_{\boldsymbol{g}}$, that is $-\Delta_{\boldsymbol{g}} \varphi_\lambda=\lambda \varphi_\lambda$.  We know that the eigenvectors of $-\overline{\Delta_{\boldsymbol{g}}}$ coincide with the eigenfunctions of $-\Delta_{\boldsymbol{g}}$ and that in particular they are smooth functions in view of Proposition \ref{PROPSPECD}.

According to the conjecture \eqref{CONJ},
the Feynman map should provide a representation of the operator
\begin{align}
    U(t):=e^{\frac{i\hbar t}{2}\overline{\Delta_{\boldsymbol{g}}}}\colon L^2(G,\mu_G)\to L^2(G,\mu_G)\,,
\end{align}
where $\mu_G$ denotes the Haar measure on $G$ (which also coincides with the $\boldsymbol{g}$-volume measure $\mu_{\boldsymbol{g}}$).
The expected result for the case $f(\gamma)=\varphi_\lambda(\gamma(t))$ is therefore
\begin{align}\label{EXPECTED}
    F_{\mathcal{H}_{x,t}(G)}(f)
    \stackrel{\scriptsize\mbox{expected}}{=}
    e^{-\frac{it\hbar}{2}\lambda}\varphi_\lambda(x)\,.
\end{align}

To explicitly compute $F_{\mathcal{H}_{x,t}(G)}(f)$ we shall provide a convenient formula for $\varphi_\lambda(\gamma_{x,\boldsymbol{v}}(t))$.
To this avail, we shall consider the right-regular representation $\pi_R\colon G\to\mathfrak{B}(L^2(G,\mu_G))$ introduced in Definition \ref{def-right-regular-representation} and refer to the content and  notation of Proposition \ref{PROPRAPR}.
In particular, for a given $X\in \mathfrak{g}$,  $X^R$ will denote the selfadjoint generator of the unitary one-parameter group  $\mathbb{R}\ni t \mapsto \pi_R(\exp(tX)) = e^{-itX^R}$. 

Let $X_1,\ldots,X_d\in\mathfrak{g}$ be an orthonormal basis of $\mathfrak{g}$.
Given the $n$-tuple $v_1,\ldots,v_n\in\mathfrak{g}$ associated with $\gamma_{x,\boldsymbol{v}}$ we decompose $v_k=\sum_{j_k}v_k^{j_k}X_{j_k}$ according to the chosen basis.
We then set $v_k\cdot X^R:=v_k^R:=\sum_{j_k}v_k^{j_k}X_{j_k}^R$, $k\in\{1,\ldots,n\}$.
Thus
\begin{align}
	\varphi_\lambda(\gamma_{x,\boldsymbol{v}}(t))
	&=\varphi_\lambda(x\exp[(t/n)v_1]\cdots\exp[(t/n)v_n])\nonumber 
	\\
	&=[\pi_R(\exp[(t/n)v_n]\cdots\exp[(t/n)v_1])\varphi_\lambda](x)\nonumber
	\\
	&=e^{-i\frac{t}{n}v_1\cdot X^R}
	\cdots e^{-i\frac{t}{n}v_n\cdot X^R}\varphi_\lambda(x) \label{calcolo}
\end{align}
where we adopted the informal notation $A\varphi_\lambda (x):= [A\varphi_\lambda](x)$ and we shall use it hereinafter.

On account of Proposition \ref{PROPRAPR}, every unitary group $e^{itX^R}$ leaves the eigenspace $H_\lambda$ invariant and its action on that space  can be computed by directly exponentiating the restriction  $X^R|_{H_\lambda}$, which is trivially bounded.
Since $\varphi_\lambda$ in (\ref{calcolo}) belongs to $H_\lambda$, the previous observation simplifies the computation of
\begin{align*}
   e^{-i\frac{t}{n}v_1\cdot X^R}
   \cdots e^{-i\frac{t}{n}v_n\cdot X^R}\varphi_\lambda\:. 
\end{align*}
Indeed, all computations are done in the \textit{finite dimensional} spaces $H_\lambda$ using the operators $e^{itX^R|_{H_\lambda}}$.
No issues regarding topology or operator domains take place in this way.

Everything trivially generalises to the case of  a function  $\varphi:G\to \bC$ 
belonging to an orthogonal sum $\bigoplus_{\lambda \in \Lambda} H_\lambda$, where $\Lambda \subset \sigma(-\overline{\Delta}_{\boldsymbol{g}})$ is bounded (i.e., finite).
 Fixing a basis, everything may be interpreted in matrix sense in a sufficiently large space $\mathbb{C}^N$.
 
A definition is natural at this juncture.
\begin{definition}\label{DEFMI}  Consider the right-regular representation of the Lie group $G$ with a bi-invariant metric ${\boldsymbol{g}}$ and with 
 unitary one-parameter groups  $\mathbb{R}\ni t \mapsto \pi_R(\exp(tX)) = e^{-itX^R}$.
The oscillatory integral 
\begin{equation}\label{MATRIX}
    \frac{1}{(2\pi\hbar i (t/n)^{-1})^{nd/2}}
    \int_{\mathbb{R}^{nd}}^o
    e^{
    \frac{i}{2\hbar}\frac{t}{n} \sum\limits_{j=1}^n\|v_j\|^2}
    [e^{-i\frac{t}{n} v_1\cdot X^R}
    \cdots
    e^{-i\frac{t}{n}v_n\cdot X^R}]\mathrm{d}v_1\cdots\mathrm{d}v_n\:,
    \end{equation}
    is understood  as a {\bf matrix-valued}  oscillatory integral when, on account of Proposition \ref{PROPRAPR}, the unbounded selfadjoint  operators $v_n\cdot X^R$  appearing therein are actually interpreted as their bounded restrictions -- henceforth improperly  called {\bf matrices} -- to a  Hilbert sum $\bigoplus_{\lambda \in \Lambda }H_{\lambda}$  with $\Lambda \subset \sigma(-\overline{\Delta}_{\boldsymbol{g}})$  bounded.
\end{definition}

Coming back to the main stream, we have found that $F_{\mathcal{H}_{x,t}(G)}(f)$ is given by
\begin{equation}\label{FM-2}
     \lim_{n\to \infty}
    \frac{1}{(2\pi\hbar i (t/n)^{-1})^{nd/2}}
    \int_{\mathbb{R}^{nd}}^o
    e^{
    \frac{i}{2\hbar}\frac{t}{n} \sum\limits_{j=1}^n\|v_j\|^2} 
    [e^{-i\frac{t}{n} v_1\cdot X^R}
    \cdots
    e^{-i\frac{t}{n}v_n\cdot X^R}]\mathrm{d}v_1\cdots\mathrm{d}v_n \varphi_\lambda(x)\,.
\end{equation}
where the integral of operators, before applying it to $\varphi_\lambda$, has the precise meaning in Definition \ref{DEFMI}.
By linearity, we can analogously treat a  linear combination of eigenfunctions \begin{align}\label{FGEN}f(\gamma) := \sum_{\lambda \in \Lambda} c_\lambda\varphi_\lambda(\gamma(t))\quad \mbox{where $\Lambda \subset \sigma(-\overline{\Delta_{\boldsymbol{g}}})$ is \textit{bounded}.}\end{align}  In that case the expected result (\ref{EXPECTED}) would be replaced by \
\begin{align}\label{EXPECTED2}
    F_{\mathcal{H}_{x,t}(G)}(f)
    \stackrel{\scriptsize\mbox{expected}}{=}\sum_{\lambda \in \Lambda} c_\lambda e^{-\frac{it\hbar}{2}\lambda}\varphi_\lambda(x)\,.
\end{align}

In summary, the problem of evaluating  $F_{\mathcal{H}_{x,t}(G)}(f)$,  with $f$ as in (\ref{FGEN}), boils down to the evaluation of the \textit{matrix-valued}  oscillatory integrals (\ref{MATRIX}) and  to the study of their limit for $n\to \infty$. 
   
\subsection{Evaluation of matrix-valued oscillatory integrals}
In order to proceed with the evaluation of the limit \eqref{FM-2} we present the following lemma, whose proof is immediate.
\begin{lemma}\label{lemma-trasf-int}
For every constant $\sigma >0$, if the left hand side of the identity below exists as an oscillatory integral according to Definition \ref{def-osc-int-findim}, then  the right-hand side exists as well and
\begin{equation}\label{eq-changevariables}
   (2\pi i \sigma )^{-n/2} \int^o_{\mathbb{R}^n}e^{\frac{i}{2\sigma}\|x\|^2}f(x)\mathrm{d}x=(2\pi i  )^{-n/2} \int^o_{\mathbb{R}^n}e^{\frac{i}{2}\|x\|^2}f(\sqrt \sigma x)\mathrm{d}x\, .
\end{equation}
\end{lemma}
Thanks to the lemma, for $t>0$ the computation of the matrix-valued oscillatory integral (\ref{MATRIX})
   boils down to  the computation of 
    \begin{equation}
        \frac{1}{(2\pi i )^{nd/2}}
    \int^o_{\mathbb{R}^{nd}}
    e^{
    \frac{i}{2} \sum\limits_{j=1}^n\|x_j\|^2}
    [e^{-i\sqrt{\hbar \frac{t}{n}} x_1\cdot X^R}
    \cdots
    e^{-i\sqrt{\hbar \frac{t}{n}}x_n\cdot X^R}]\mathrm{d}x_1\cdots\mathrm{d}x_n
    \end{equation}
    The following theorem shows how the matrix-valued oscillatory integral \eqref{MATRIX} can be transformed into an (absolutely convergent) Gaussian integral.
   \begin{theorem}\label{teo-osc-gau} Let $t>0$ and consider a finite dimensional representation of the Lie group $G$ with one parameter subgroups $\mathbb{R} \ni t \mapsto e^{itX^R}$. Then for any positive integer $n\geq 1$ the matrix valued-oscillatory integral
   $$ \frac{1}{(2\pi i )^{nd/2}}
    \int^o_{\mathbb{R}^{nd}}
    e^{
    \frac{i}{2} \sum\limits_{j=1}^n\|x_j\|^2}
    [e^{-i\sqrt{\hbar \frac{t}{n}} x_1\cdot X^R}
    \cdots
    e^{-i\sqrt{\hbar \frac{t}{n}}x_n\cdot X^R}]\mathrm{d}x_1\cdots\mathrm{d}x_n$$
   is equal to the Gaussian integral
   \begin{equation}
       \label{int-osc-rot}
    \frac{1}{(2\pi  )^{nd/2}}
    \int_{\mathbb{R}^{nd}}
    e^{-
    \frac{1}{2} \sum\limits_{j=1}^n\|x_j\|^2}
    [e^{-ie^{i\pi/4}\sqrt{\hbar \frac{t}{n}} x_1\cdot X^R}
    \cdots
    e^{-ie^{i\pi/4}\sqrt{\hbar \frac{t}{n}}x_n\cdot X^R}]\mathrm{d}x_1\cdots\mathrm{d}x_n\, .
    \end{equation} 
   \end{theorem} 
   \begin{proof}
       Let $\varphi\in S(\mathbb{R}^{nd})$ be a Schwartz function such that $\varphi (0)=1$ and, for any $\epsilon >0$, let us focus on the regularized integral $I_\varphi (\epsilon)$ defined as 
       \begin{equation}\label{int-reg1}
       I_\varphi (\epsilon):=\frac{1}{(2\pi i )^{nd/2}}
    \int_{\mathbb{R}^{nd}}
    e^{
    \frac{i}{2} \sum\limits_{j=1}^n\|x_j\|^2}
    [e^{-i\sqrt{\hbar \frac{t}{n}} x_1\cdot X^R}
    \cdots
    e^{-i\sqrt{\hbar \frac{t}{n}}x_n\cdot X^R}]\varphi(\epsilon x)\mathrm{d}x_1\cdots\mathrm{d}x_n\,. \end{equation}
    In fact, every component of the matrix-valued integral  \eqref{int-reg1} is  well-defined as an absolutely convergent  integral since the matrix $e^{-i\sqrt{\hbar \frac{t}{n}} v\cdot X^R}$ is unitary for every $v\in \mathbb{R}^{d}$. 
    Moreover, since the map
    \begin{equation}\label{distr-Phi}
        \mathbb{R}^{nd}\ni x\mapsto \Phi(x):= \frac{1}{(2\pi i )^{nd/2}}e^{
    \frac{i}{2} \sum\limits_{j=1}^n\|x_j\|^2}
    [e^{-i\sqrt{\hbar \frac{t}{n}} x_1\cdot X^R}
    \cdots
    e^{-i\sqrt{\hbar \frac{t}{n}}x_n\cdot X^R}]
    \end{equation} is continuous and bounded, it defines a (matrix valued) Schwartz distribution. By introducing the Fourier transform $\hat \varphi$ of the Schwartz test function $\varphi$, the integral $I_\varphi (\epsilon)$ can be equivalently written as
    \begin{equation}\label{int-phiepsilon}
        I_\varphi (\epsilon)=\frac{1}{(2\pi)^{nd}}\int_{\mathbb{R}^{nd}}\hat \Phi(k)\frac{\hat \varphi \left(\frac{k}{\epsilon}\right)}{\epsilon^{nd}}\mathrm{d}k\, ,
    \end{equation} 
    where $\hat \Phi$ denotes the Fourier transform of the distribution \eqref{distr-Phi}.
    This can be computed by introducing a suitable regularization as:
   \begin{equation}\label{FTPHI}
   \hat \Phi (k) =\lim_{N\to \infty}\int _{\mathbb{R}^{nd}}e^{ikx}\Phi(x) e^{-\frac{1}{2N}\sum\limits_{j=1}^n\|x_j\|^2}
         \mathrm{d}x\, .\end{equation}
   Indeed, the  sequence of (matrix valued) Schwartz distributions $\Phi_N$ associated to the continuous summable functions $ \Phi_N(x):=\Phi(x)e^{-\frac{1}{2N}\sum\limits_{j=1}^n\|x_j\|^2}$ converges to $\Phi$ in $S'$, hence equality \eqref{FTPHI}  follows from the sequential continuity of Fourier transform in the space of Schwartz distributions. In particular we have:
    \begin{align*}
         \hat \Phi_N (k) 
         =\int_{\mathbb{R}^{nd}} e^{ikx}[e^{-i\sqrt{\hbar \frac{t}{n}} x_1\cdot X^R}
    \cdots
    e^{-i\sqrt{\hbar \frac{t}{n}}x_n\cdot X^R}]\frac{e^{
   \left( \frac{i}{2} -\frac{1}{2N}\right)\sum\limits_{j=1}^n\|x_j\|^2}}{(2\pi i )^{nd/2}}\mathrm{d}x\,.
    \end{align*}
          The latter can be actually computed by introducing polar coordinates in $\mathbb{R}^{nd}$ in the following way: 
      \begin{equation}\label{Fou-Phihat}
      \hat \Phi _N (k) =  \int_{\mathbb{S}^{nd-1}}\int_0^{+\infty} e^{-\rho^2/2N} e^{i\rho \hat n \cdot k}\Phi (\rho\nu)\rho ^{nd-1} \mathrm{d}\rho \mathrm{d}S (\nu)\, ,
      \end{equation}
  where $S$ stands for the surface measure on the unitary spherical hypersurface $\mathbb{S}^{nd-1}\equiv\{x\in \bR^{nd}\colon \|x\|=1\}$, while the inner integral can be interpreted as the Fourier transform of the matrix-valued distribution $\psi$ on the real line defined as:
  \begin{align}
       \psi (x):&=\theta (x)x^{nd-1}\Phi(x \nu)e^{-x^2/2N} \nonumber \\
       &=\theta (x)x^{nd-1}[e^{-i\sqrt{\hbar \frac{t}{n}}x \nu_1\cdot X^R}
    \cdots
    e^{-i\sqrt{\hbar \frac{t}{n}}x \nu_n\cdot X^R}]\frac{e^{
   \left( \frac{i}{2} -\frac{1}{2N}\right)x^2}}{(2\pi i )^{nd/2}}\, . \label{PSIDEF}
  \end{align}
By dominated convergence we have 
$$ \hat \psi (k)=\lim_{R_0\to +\infty} \int_0^{R_0}e^{ikx}\psi (x)\mathrm{d}x\,.$$
   We can then use a suitable deformation of the integration contour in the complex plane, following the argument presented in greater detail in \cite{AlMa05}. More specifically, let $\widetilde{\psi} $ denote the analytic continuation in the complex plane of the map $\psi$  defined by the right-hand side of (\ref{PSIDEF}) \textit{omitting the singular factor $\theta(x)$}. 
   Let us consider now the   three paths in the complex plane, \textit{cf.} Figure \ref{Fig: 3 paths}:
   \begin{align*}
       \gamma_1&=\{z= r\, , 0\leq r\leq R_0\}\\
       \gamma_2&=\{z= R_0e^{i\theta}\, , \epsilon\leq \theta\leq \pi/4\}\\
       \gamma_3&=\{z= re^{i\pi/4}\, , 0\leq r\leq R_0\}
   \end{align*}
    \begin{figure}[h!]
			\centering
			\begin{tikzpicture}[scale=2.5]
				\draw[->] (0,0) -- (1.2,0);
				\draw[->] (0,0) -- (0,1);
				\draw[variable=\t,domain=0:45,samples=100]
				plot ({cos(\t)},{sin(\t)});
                \draw[->,variable=\t,domain=20:{20+25/2},samples=500]
				plot ({cos(\t)},{sin(\t)}) node[right]  {$\gamma_2$};
                \draw[->] (0,0) -- ({0.5*cos(0)},{0.5*sin(0)}) node[below right] {$\gamma_1$};
                \draw (0,0) -- ({cos(45)},{sin(45)});
                \draw[->] (0,0) -- ({0.5*cos(45)},{0.5*sin(45)}) node[above left] {$\gamma_3$};
			\end{tikzpicture}
			\caption{The paths $\gamma_1,\gamma_2,\gamma_3$.}
			\label{Fig: 3 paths}
		\end{figure}
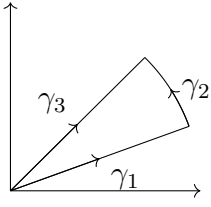
   Due to  the complex analyticity of the matrix-valued map $\widetilde{\psi}$ in the region bounded by the union of the three curves, we have
   $$ \int _{\gamma_1\cup\gamma_2\cup (-\gamma_3)} e^{ikz}\widetilde{\psi} (z) \mathrm{d}z=0\, .$$
   In particular this yields
   \begin{equation*}
       \lim_{R_0\to \infty}\int _{\gamma_1} e^{ikz}\widetilde{\psi} (z) \mathrm{d}z=\int_0^\infty e^{ikre^{i\pi/4}}\widetilde{\psi} (re^{i\pi/4})e^{i\pi/4}\mathrm{d}r-\lim_{R_0\to +\infty}\int_0^{\pi/4}e^{ikR_0e^{i\theta}}\widetilde{\psi} (R_0e^{i\theta})R_0e^{i\theta} \mathrm{d}\theta.
   \end{equation*}
The first integral on the right hand side is the Gaussian integral 
$$  \int_0^\infty e^{ikre^{i\pi/4}}r^{nd-1}[e^{-ie^{i\pi/4}\sqrt{\hbar \frac{t}{n}}x \nu_1\cdot X^R}
    \cdots
    e^{-ie^{i\pi/4}\sqrt{\hbar \frac{t}{n}}r \nu_n\cdot X^R}]\frac{e^{
    -\frac{1}{2} r^2}e^{-\frac{i}{2N}r^2}}{(2\pi  )^{nd/2}}\mathrm{d}r\, .$$
Concerning the second integral, by considering the operator norm $\|X^R_j\|$ of the matrices $X^R_j$, $j=1,\dots , d$ and by setting
\begin{equation}
    \label{def-l}l:=d\max_j \|X^R_j\|\,,
\end{equation}
we have the trivial bound $\|\nu_j\cdot X^R\|\leq l $, valid for any $\nu_j\in \mathbb{R}^{d}$ such that $\|\nu_j\|\leq 1$.
In particular each component of the matrix $\int_0^{\pi/4}e^{ikR_0e^{i\theta}}\tilde\psi (R_0e^{i\theta})R_0e^{i\theta} \mathrm{d}\theta$ will be bounded by
\begin{multline*}
    \frac{R_0^{nd}}{(2\pi)^{nd/2}}\int_0^{\pi/4}  e^{-kR_0\sin\theta} e^{R_0l\sqrt{\hbar t n}\sin\theta}e^{-\left(\sin (2\theta)+\frac{\cos(2\theta)}{N}\right)\frac{R_0^2}{2}}  \mathrm{d}\theta
    \\
    \leq\frac{R_0^{nd}}{(2\pi)^{nd/2}}
    \bigg(\frac{\pi}{4}\bigg)
    e^{-\frac{R_0^2}{2N}}
    e^{(|k|+l\sqrt{\hbar t n})R_0/\sqrt{2}}\,,
\end{multline*}
 which converges to 0 for $R_0\to \infty$.
 This yields:
 \begin{equation}\nonumber
      \hat \Phi_N (k)=\int_{\mathbb{R}^{nd}} e^{ie^{i\pi /4}kx}[e^{-ie^{i\pi /4}\sqrt{\hbar \frac{t}{n}} x_1\cdot X^R}
    \cdots
    e^{-ie^{i\pi /4}\sqrt{\hbar \frac{t}{n}}x_n\cdot X^R}]\frac{e^{
   - \frac{1}{2}\left(1+\frac{i}{2N}\right) \sum\limits_{j=1}^n\|x_j\|^2}}{(2\pi  )^{nd/2}}\mathrm{d}x\, .
 \end{equation}
 Finally, By taking the limit for $N\to \infty$ in \eqref{FTPHI} we can conclude that the matrix valued distribution $\hat \Phi$ is a smooth map given by the Gaussian integral
 \begin{equation}\label{rapp-fou}
      \hat \Phi (k)
      =\int_{\mathbb{R}^{nd}} e^{ie^{i\pi /4}kx}[e^{-ie^{i\pi /4}\sqrt{\hbar \frac{t}{n}} x_1\cdot X^R}
    \cdots
    e^{-ie^{i\pi /4}\sqrt{\hbar \frac{t}{n}}x_n\cdot X^R}]\frac{e^{
   - \frac{1}{2} \sum\limits_{j=1}^n\|x_j\|^2}}{(2\pi  )^{nd/2}}\mathrm{d}x\, .
 \end{equation}
 
 Coming back to the regularized integral \eqref{int-phiepsilon}, a simple change of variable argument leads to the following  expression:
 \begin{equation*}
        I_\varphi (\epsilon)=\frac{1}{(2\pi)^{nd}}\int_{\mathbb{R}^{nd}}\hat \Phi(\epsilon k)\hat \varphi \left(k\right)\mathrm{d}k\, .
    \end{equation*}
    By using the trivial identity $\frac{1}{(2\pi)^{nd}}\int \hat \varphi \left(k\right) \mathrm{d}k=\varphi(0)=1$ and the  dominated convergence theorem, which applies thanks to the representation \eqref{rapp-fou} for $\hat \Phi$, we eventually obtain 
   \begin{multline*}
       \frac{1}{(2\pi i )^{nd/2}}
        \int^o_{\mathbb{R}^{nd}}
        e^{
        \frac{i}{2} \sum\limits_{j=1}^n\|x_j\|^2}
        [e^{-i\sqrt{\hbar \frac{t}{n}} x_1\cdot X^R}
        \cdots
        e^{-i\sqrt{\hbar \frac{t}{n}}x_n\cdot X^R}]\mathrm{d}x_1\cdots\mathrm{d}x_n
       \\=\lim_{\epsilon\downarrow 0} I_\varphi (\epsilon)
       =\hat \Phi(0)
       =\int_{\mathbb{R}^{nd}}[e^{-ie^{i\pi /4}\sqrt{\hbar \frac{t}{n}} x_1\cdot X^R}
    \cdots
    e^{-ie^{i\pi /4}\sqrt{\hbar \frac{t}{n}}x_n\cdot X^R}]\frac{e^{
   - \frac{1}{2} \sum\limits_{j=1}^n\|x_j\|^2}}{(2\pi  )^{nd/2}}\mathrm{d}x\, .
   \end{multline*}
   \end{proof}

 By applying Theorem \ref{teo-osc-gau} and Fubini's Theorem we find the following useful corollary.
 
 \begin{corollary} With the same hypotheses as in Theorem \ref{teo-osc-gau}, the following \textit{Fubini-type identity} holds
 \begin{multline}\label{Fubini-type}
      \frac{1}{(2\pi i )^{nd/2}}
    \int^o_{\mathbb{R}^{nd}}
    e^{
    \frac{i}{2} \sum\limits_{j=1}^n\|x_j\|^2}
    [e^{-i\sqrt{\hbar \frac{t}{n}} x_1\cdot X^R}
    \cdots
    e^{-i\sqrt{\hbar \frac{t}{n}}x_n\cdot X^R}]\mathrm{d}x_1\cdots\mathrm{d}x_n\\
    =\prod _{j=1}^n\frac{1}{(2\pi i )^{d/2}} \int^o_{\mathbb{R}^{d}}e^{
    \frac{i}{2} 
    \|x_j\|^2}e^{-i\sqrt{\hbar \frac{t}{n}}x_j\cdot X^R}\mathrm{d}x_j\:.
 \end{multline}
 \end{corollary}
\subsection{Chernoff approximations of the Schr\"odinger group}
We will now study the limit of the finite dimensional oscillatory integrals \eqref{FM-2} and relate it to the solution of the Schr\"odinger equation \eqref{Schroedinger} with $V=0$ and $\psi_0= \sum_{\lambda \in \Lambda} c_\lambda\varphi_\lambda$, where $\Lambda \subset \sigma(-\overline{\Delta_{\boldsymbol{g}}})$ is bounded. A priori, this can be achieved by an application of the renown Chernoff's theorem  \cite{Chernoff,EN,BS}.
However, in our setting we are dealing with linear operators on finitely dimensional spaces. Therefore, it is more efficient to provide a concrete proof of Chernoff's theorem for the case at hand: The latter will have the merit to enlighten the quantitative bounds of our approximation.

 Let us consider now the strongly continuous Schr\"odinger unitary group (which is therefore a $C_0$-semigroup on the Hilbert space $ L^2(G,\mu_G)$ for $t\in \mathbb{R}_+$)
  $$U(t):=e^{\frac{it\hbar}{2}\overline{\Delta_{\boldsymbol{g}}}}\colon L^2(G,\mu_G)\to L^2(G,\mu_G)\:,$$ and, for the moment, a single eigenfunction $\varphi_\lambda$ of the Laplace-Beltrami operator $\Delta_{\boldsymbol{g}}$, i.e. $-\Delta_{\boldsymbol{g}}\varphi_\lambda=\lambda\varphi_\lambda $. We consider the Feynman map $F_{\mathcal{H}_{x,t}(G)}$ (\ref{Def-FM-M}) applied to the function $f: \mathcal{H}_{x,t}(G)\to \mathbb{C}$ given by
  $$ f(\gamma)=\varphi_\lambda (\gamma(t))\, , \qquad \gamma \in \mathcal{H}_{x,t}(G).$$
  As discussed above, $F_{\mathcal{H}_{x,t}(G)}(f) $ can be computed in terms of the limit of matrix-valued oscillatory integrals \eqref{FM-2} that, thanks the Fubini-type formula \eqref{Fubini-type}, leads to the following representation:
  \begin{equation}
     F_{\mathcal{H}_{x,t}(G)}(f)=\lim_{n\to \infty}( S(t/n))^n\varphi_\lambda (x)\, ,
  \end{equation}
  where 
  $(S(t))_{t\geq 0}$,  is the 1-parameter family of operators on $\oplus_{\lambda\in\Lambda} H_\lambda$, where $\Lambda$ is a finite subset of $\sigma(\overline{\Delta}_{\boldsymbol{g}})$, defined by
\begin{multline}\label{rapp-gaus-int}
    S(t)
    :=\frac{1}{(2\pi i )^{d/2}}
    \int^o_{\mathbb{R}^{d}}
    e^{
    \frac{i}{2} \|x\|^2}
    e^{-i\sqrt{\hbar t} x\cdot X^R}
    \mathrm{d}x
    \\
    =\frac{1}{(2\pi  )^{d/2}}
    \int_{\mathbb{R}^{d}}
    e^{-
    \frac{1}{2} \|x\|^2}
    e^{-ie^{i\pi/4}\sqrt{\hbar t} x\cdot X^R}
    \mathrm{d}x
    \in\mathfrak{B}(\oplus_{\lambda\in\Lambda}H_\lambda)\,,
\end{multline}
where in the second equality we used theorem \ref{teo-osc-gau}.
Notice that $\oplus_{\lambda\in\Lambda}H_\lambda$ is finite dimensional, therefore, $S(t)$ can be identified with a square matrix of suitable size (which increases with $\# \Lambda$).
The following lemma recollects the main properties of the map $t\mapsto S(t)$.

\begin{lemma}\label{Lem: Taylor expansion of S(t)}
    It holds $S(0)=I$, moreover,
    \begin{align}\label{Eq: Taylor expansion of S(t)}
        S(t)
        =I
        +\frac{i\hbar t}{2}\Delta_{\boldsymbol{g}}
        +R_2(t)\,,
        \qquad
        \|R_2(t)\|\leq C_2t^2\,.
    \end{align}
\end{lemma}
\begin{proof}
    The initial condition $S(0)=I$ is verified by direct inspection.
    For what concerns \eqref{Eq: Taylor expansion of S(t)}, the representation \eqref{rapp-gaus-int} leads to:
    \begin{multline}\label{rap-st-C4}
        S(t)=
        \int_{\mathbb{R}^{d}}\bigg[I-ie^{i\pi/4}\sqrt{\hbar t} x\cdot X^R-\frac{i\hbar}{2}(x\cdot X^R)^2 t
        \\
        -\frac{e^{i\pi/4}}{3!}\hbar ^{3/2}(x\cdot X^R)^3 t^{3/2}\bigg]\frac{e^{-\frac{\|x\|^2}{2}}}{(2\pi  )^{d/2}}\mathrm{d}x
        +R_2(t)\,,
  \end{multline}
  with 
  \begin{equation}\label{rap-Rt-C4}
      R_2(t)= \int_{\mathbb{R}^{d}}\left(\int_0^1\frac{(-ie^{i\pi/4}\sqrt{\hbar t} x\cdot X^R)^4}{3!} e^{-ie^{i\pi/4}\sqrt{\hbar t}u x\cdot X^R} (1-u)   \mathrm{d}u \right) \frac{e^{-\frac{\|x\|^2}{2}}}{(2\pi  )^{d/2}}\mathrm{d}x\, .
  \end{equation}
  Each term in \eqref{rap-st-C4} can be treated separately. In particular:
  \begin{align*}
      \int_{\mathbb{R}^{d}}x\cdot X^R\frac{e^{-\frac{\|x\|^2}{2}}}{(2\pi  )^{d/2}}\mathrm{d}x&
      =\sum_{j=1}^{d}X^R_j\int_{\mathbb{R}^{d}}x^j\frac{e^{-\frac{\|x\|^2}{2}}}{(2\pi  )^{d/2}}\mathrm{d}x=0\,,
      \\
      \int_{\mathbb{R}^{d}}(x\cdot X^R)^3\frac{e^{-\frac{\|x\|^2}{2}}}{(2\pi  )^{d/2}}\mathrm{d}x&
      =\sum_{j,k,l=1}^{d}X^R_jX^R_kX^R_l\int_{\mathbb{R}^{d}}x^jx^kx^l\frac{e^{-\frac{\|x\|^2}{2}}}{(2\pi  )^{d/2}}\mathrm{d}x=0
  \end{align*}
  while 
  \begin{align*}
      \int_{\mathbb{R}^{d}}(x\cdot X^R)^2\frac{e^{-\frac{\|x\|^2}{2}}}{(2\pi  )^{d/2}}\mathrm{d}x&=\sum_{j,k=1}^{d}X^R_jX^R_k\int_{\mathbb{R}^{d}}x^jx^k\frac{e^{-\frac{\|x\|^2}{2}}}{(2\pi  )^{d/2}}\mathrm{d}x
      \\
      &=\sum_{j,k=1}^{d}X^R_jX^R_k\delta_{jk}
      {
      =-X_{\boldsymbol{g}}^2
      =-\Delta_{\boldsymbol{g}}}\,,
  \end{align*}
 {where we used} proposition \ref{PROPX2DELTA} {together with the identity $X^R=i\tilde{X}$, \textit{cf} proposition \ref{PROPRAPR}-(2b)}.
  Overall we have
  $$ S(t)
  =I+\frac{i\hbar t}{2}\Delta_{\boldsymbol{g}}+R_2(t)\,.$$
  The remainder term $R_2(t) $ can be estimated by using representation \eqref{rap-Rt-C4}. In particular, recalling the definition \eqref{def-l} of the positive constant $l$, we have:
  \begin{align*}
      \|R_2(t)\|&\leq \frac{\hbar^2 t^2}{3!}\int_{\mathbb{R}^{d}}\int_0^1(1-u)\| x\cdot X^R\|^4\|e^{\frac{\sqrt 2 }{2}\sqrt{\hbar t}u x\cdot X^R}\|\frac{e^{-\frac{\|x\|^2}{2}}}{(2\pi  )^{d/2}}\mathrm{d}x\mathrm{d}u\\
      &\leq  \frac{\hbar^2 t^2}{3!} l^4 \int_{\mathbb{R}^{d}}\int_0^1(1-u)\|x\|^4 e^{\frac{\sqrt 2 }{2}\sqrt{\hbar t}u \|x\|l}\frac{e^{-\frac{\|x\|^2}{2}}}{(2\pi  )^{d/2}}\mathrm{d}x\mathrm{d}u\\
      &\leq \frac{\hbar^2 t^2}{12}\int_{\mathbb{R}^{d}}\|x\|^4 e^{\frac{\sqrt 2 }{2}\sqrt{\hbar t} \|x\|l}\frac{e^{-\frac{\|x\|^2}{2}}}{(2\pi  )^{d/2}}\mathrm{d}x\,.
  \end{align*}
  Hence, if $|t|\leq 1$ the remainder satisfies the estimate $ \|R_2(t)\|\leq C_2t^2$, with the constant $C_2$ given by
  $$C_2=\frac{\hbar^2}{12}\int_{\mathbb{R}^{d}}\|x\|^4 e^{\frac{\sqrt 2 }{2}\sqrt{\hbar } \|x\|l}\frac{e^{-\frac{\|x\|^2}{2}}}{(2\pi  )^{d/2}}\mathrm{d}x\,.$$
\end{proof}

 \begin{theorem}\label{teoCherAppr}
     Let $S:\mathbb{R}^+\to \mathfrak{B}(\bigoplus_{\lambda \in \Lambda} H_\lambda)$ be the map defined by the matrix-valued oscillatory integral 
     \begin{equation}\label{FormulaChernoffOperator}
         S(t):= 
   \frac{1}{(2\pi i )^{d/2}}
    \int^o_{\mathbb{R}^{d}}
    e^{
    \frac{i}{2} \|x\|^2}
    e^{-i\sqrt{\hbar t} x\cdot X^R}
    \mathrm{d}x\,.
     \end{equation}
     Then we have
     \begin{align}\label{Eq: Chernoff approximation - standard}
         \exists C>0\colon
         \sup_{t\in[0,T]}\|S(t/n)^n-e^{\frac{it\hbar}{2}\overline{\Delta_{\boldsymbol{g}}}}\|
         \leq C/n\,.
     \end{align}
     In other words, the limit $S(t/n)^n$ converges to the (restriction to $\oplus_{\lambda\in\Lambda}H_\lambda$ of the) 1-parameter group of unitary operators $U(t)=e^{\frac{it\hbar}{2}\overline{\Delta_{\boldsymbol{g}}}}$.
 \end{theorem}
\begin{proof}
    In what follows, we will consider the restriction of $\Delta_{\boldsymbol{g}}$ to $\oplus_{\lambda\in\Lambda}H_\lambda$.
    The latter is a trivially bounded operator with norm $\max_{\lambda\in\Lambda}\lambda$.
    Since $\oplus_{\lambda\in\Lambda}H_\lambda$ is finitely dimensional, we also drop the distinction between $\Delta_{\boldsymbol{g}}$ and $\overline{\Delta_{\boldsymbol{g}}}$.
    By Lemma \ref{Lem: Taylor expansion of S(t)} we have $\|S(t)-I\|\leq ct$ for some positive constant $c>0$.
    This implies that $\|S(t/n)-I\|<1$ for $n$ large enough and all $t\in [0,T]$.
    Thus, the matrix logarithm of $S(t/n)^n$ makes sense and we have
    \begin{align*}
        S(t/n)^n=\exp[\log(S(t/n)^n)]
        =\exp[n\log S(t/n)]\,.
    \end{align*}
    Moreover, the Taylor expansion of the matrix logarithm and Lemma \ref{Lem: Taylor expansion of S(t)} leads to
    \begin{align*}
        \log S(t/n)
        &=S(t/n)-I
        +R_2'(t/n)
        &\exists C_2'>0\colon\|R_2'(t/n)\|\leq C_2't^2/n^2
        \\
        &=\frac{i\hbar t}{2n}\Delta_{\boldsymbol{g}}
        +R_2(t/n)+R_2'(t/n)\,.
        &\textrm{Lemma }\ref{Lem: Taylor expansion of S(t)}
    \end{align*}
    It follows that
    \begin{align*}
        S(t/n)^n=\exp\Big[
        \frac{i\hbar t}{2}\Delta_{\boldsymbol{g}}
        +R_1(t^2/n)
        \Big]
        \qquad
        \exists C_1>0\colon \|R_1(t^2/n)\|
        =C_1t^2/n\,.
    \end{align*}
    We now estimate
    \begin{multline*}
        S(t/n)^n-e^{\frac{it\hbar}{2}\Delta_{\boldsymbol{g}}}
        =\int_0^1\frac{\mathrm{d}}{\mathrm{d}\varepsilon}
        e^{\frac{it\hbar}{2}\Delta_{\boldsymbol{g}}
        +\varepsilon R_1(t^2/n)}
        \mathrm{d}\varepsilon
        \\
        =\int_0^1
        \sum_{n\geq 1}
        \frac{1}{n!}
        \sum_{k=0}^{n-1}
        \Big(\frac{it\hbar}{2}\Delta_{\boldsymbol{g}}
        +\varepsilon R_1(t^2/n)\Big)^k
        R_1(t^2/n)
        \Big(\frac{it\hbar}{2}\Delta_{\boldsymbol{g}}
        +\varepsilon R_1(t^2/n)\Big)^{n-1-k}
        \mathrm{d}\varepsilon\,.
    \end{multline*}
    Thus, we have
    \begin{align*}
        \|S(t/n)^n-e^{\frac{it\hbar}{2}\Delta_{\boldsymbol{g}}}\|
        &\leq C_1\frac{t^2}{n}\int_0^1
        \sum_{n\geq 1}
        \frac{1}{(n-1)!}
        \Big\|\frac{it\hbar}{2}\Delta_{\boldsymbol{g}}
        +\varepsilon R_1(t^2/n)\Big\|^{n-1}
        \mathrm{d}\varepsilon
        \\
        &\leq
        C_1\frac{T^2}{n}e^{\frac{T\hbar}{2}\|\Delta_{\boldsymbol{g}}\|+\frac{C_1T^2}{n}}
        \leq\frac{C}{n}\,,
    \end{align*}
    which provides the claimed uniform bound over $t\in [0,T]$.
\end{proof}

\begin{remark}\label{Remark-properties-S(t)}
    \noindent
    \begin{enumerate}[(i)]
    \item
    Theorem \ref{teoCherAppr} is nothing but a finite-dimensional version of the renown Chernoff's theorem.
    The formulation in the infinite-dimensional setting requires an exponential bound of the norm of $S(t)$, namely there should exists $A>0$ and $\omega\in\mathbb{R}$ such that $\|S(t)\|\leq Ae^{\omega t}$.
    In the present setting this assumption is satisfied because
     \begin{align*}
        \|S(t)\|
        \leq  \frac{1}{(2\pi  )^{d/2}}
       \int_{\mathbb{R}^{d}}
       e^{-
        \frac{1}{2} \|x\|^2}
        e^{\frac{\sqrt{2}}{2}\sqrt{\hbar t} l\|x\|}
        \mathrm{d}x
       =\left(\sqrt{\frac{2}{\pi  }}\int_{0}^\infty
       e^{-
        \frac{1}{2} r^2}e^{\frac{\sqrt{2}}{2}\sqrt{\hbar t} lr}\mathrm{d}r\right)^d
        =2^de^{\frac{dl^2\hbar}{4}t}\,,
    \end{align*}
    where $l$ as been defined in equation \eqref{def-l}.
    This estimate was not used for the proof of Theorem \ref{teoCherAppr}, nevertheless, the latter property plays a prominent role also in the finite-dimensional case.
    Indeed, under this assumption it can be shown that
    \begin{align}\label{Eq: Chernoff approximation - improvement excluding 0}
         \forall\varepsilon>0\,,\,
         \exists C>0\colon
         \sup_{t\in[\varepsilon,T]}\|S(t/n)^n-e^{\frac{it\hbar}{2}\overline{\Delta_{\boldsymbol{g}}}}\|
         \leq C/n^2\,,
     \end{align}
     which provides an improvement of the rate of convergence.
     To prove \eqref{Eq: Chernoff approximation - improvement excluding 0} we first observe that, considering $t\in[\varepsilon,T]$, we may assume $A<1$ in the estimate $\|S(t)\|\leq Ae^{\omega t}$ at the price of increasing $\omega$.
     Indeed, let $0<A'<1$: Then
     \begin{align*}
         \|S(t)\|\leq Ae^{\omega t}
         =A'e^{\omega t+\log A/A'}
         \leq A'e^{\omega 't}
         \qquad
         \omega'
         =\omega+\varepsilon^{-1}\log A/A'\,.
     \end{align*}
     Thus, in what follows we will assume $A<1$.
     Lemma \ref{Lem: Taylor expansion of S(t)} and the expansion of the exponential $e^{\frac{it\hbar}{2}\Delta_{\boldsymbol{g}}}$ implies
     \begin{align*}
         S(t)-e^{\frac{it\hbar}{2}\Delta_{\boldsymbol{g}}}
         =R_2(t)
         \qquad
         \exists C_2>0\colon\|R_2(t)\|\leq C_2t^2\,.
     \end{align*}
     Next, we observe that
     \begin{align*}
          S(t/n)^n
         -e^{\frac{it\hbar}{2}\Delta_{\boldsymbol{g}}}
         =S(t/n)^n
         -(e^{\frac{it\hbar}{2n}\Delta_{\boldsymbol{g}}})^n
         =\sum_{j=0}^{n-1}
         S(t/n)^j
         R_2(t/n)
         (e^{\frac{it\hbar}{2n}\Delta_{\boldsymbol{g}}})^{n-1-j}
         \,.
     \end{align*}
     It follows that
     \begin{align*}
         \|S(t/n)^n
         -e^{\frac{it\hbar}{2}\Delta_{\boldsymbol{g}}}\|
         \leq\frac{C_2T^2}{n^2}\sum_{j=0}^{n-1}
         A^je^{\omega tj/n}
         \leq\frac{C_2T^2}{n^2}
         \frac{1-A^ne^{\omega T}}{1-Ae^{\omega T/n}}
         \leq\frac{C}{n^2}\,,
     \end{align*}
    as claimed.

    \item 
    For later convenience we point out that, with minor changes to the proof of Theorem \ref{teoCherAppr}, we may prove the following variation of \eqref{Eq: Chernoff approximation - standard}, \textit{cf.} \cite[Cor. 5.4]{EN} for the infinite dimensional case.
    Let $t>0$ and let $\{t_n\}$ be a positive null sequence and $\{k_n\}$ be an increasing
    sequence of integers such that $t_nk_n\underset{n\to\infty}{\longrightarrow}t$.
    Then 
    \begin{align}\label{Eq: Chernoff approximation - with arbitrary sequence}
        \lim_{n\to \infty }\|S(t_n)^{k_n}-e^{\frac{it\hbar}{2}\overline{\Delta_{\boldsymbol{g}}}}\|=0\,.
    \end{align}
    \end{enumerate}
\end{remark}

 Let us finally come back to the initial problem of the computation of 
\begin{multline}\label{FM-2X}
     F_{\mathcal{H}_{x,t}(G)}(f)
   =\lim_{n\to \infty}
    \frac{1}{(2\pi\hbar i (t/n)^{-1})^{nd/2}}
    \int^o
    e^{
    \frac{i}{2\hbar}\frac{t}{n} \sum\limits_{j=1}^n\|v_j\|^2} \\
    [e^{-i\frac{t}{n} v_1\cdot X^R}
    \cdots
    e^{-i\frac{t}{n}v_n\cdot X^R}]\mathrm{d}v_1\cdots\mathrm{d}v_n \sum_{\lambda \in \Lambda}c_\lambda \varphi_\lambda(x)\,.
\end{multline}
with  $f(\gamma) := \sum_{\lambda \in \Lambda} c_\lambda\varphi_\lambda(x)$   where, as mentioned above,  $\Lambda \subset \sigma(-\overline{\Delta_{\boldsymbol{g}}})$ is bounded.
Thus, we have  established identity (\ref{EXPECTED2}). That identity is the simplest case of the  conjecture stated in Section \ref{SECCONJ}. We have in fact proved the following theorem.\\
 
  \begin{theorem}\label{FinalTheoremSchroedinger1} Consider a compact Lie group $G$ equipped with a bi-invariant Riemannian metric ${\boldsymbol{g}}$ and let $\Delta_{\boldsymbol{g}}$ denote
the   Laplace-Beltrami operator (\ref{LB}).
    If $\varphi \in C^\infty(G;\mathbb{C})$ is a finite energy function (Definition \ref{Def: finite energy functions}), define  $f: \mathcal{H}_{x,t}(G)\to \mathbb{C}$ as
      $$ f(\gamma)=\varphi (\gamma(t))\, , \qquad \gamma \in \mathcal{H}_{x,t}(G),$$
      where the space of paths $\mathcal{H}_{x,t}(G)$ is defined in (\ref{HxtM}).
      Then the Feynman map $F_{\mathcal{H}_{x,t}(G)}$ (\ref{FM-2}) applied to the function $f$ gives the solution of the free Schr\"odinger equation with initial datum $\varphi$:
      $$F_{\mathcal{H}_{x,t}(G)}(f)= \left(
      e^{\frac{i\hbar t}{2}\overline{\Delta_{\boldsymbol{g}}}  }\varphi
      \right)(x)\,, \quad \forall x\in G, \forall t\in \mathbb{R}\:.$$
  \end{theorem}
  $\null$\\
  With a more formal notation, closer to the original intuition by Feynman, we have found that
      \begin{align}\nonumber
   \widetilde{\int}_{\mathcal{H}_{x,t}(G)} e^{\frac{i}{2\hbar}\int_0^t\boldsymbol{g}(\dot{\gamma}(s),\dot{\gamma}(s))\mathrm{d}s} \varphi(\gamma(t))\mathrm{d}\gamma =  \left(
      e^{\frac{i\hbar t}{2}\overline{\Delta_{\boldsymbol{g}}}  }\varphi
      \right)(x)\:.
\end{align}
 \begin{remark}
    Given the finite dimension of the space $\bigoplus_{\lambda \in \Lambda} H_\lambda$
  with $\Lambda$ finite and the smoothness of the eigenfunctions $\varphi_\lambda$, the Feynman map $F_{\mathcal{H}_{x,t}(G)}$ provides both a classical solution and a solution in $L^2(G,\mu_G)$ of the Schr\"odinger equation.
\end{remark}
\begin{remark}\label{Rmk: finite energy assumption}
    \noindent
    \begin{enumerate}[(i)]
        \item
        The assumptions on the function $\varphi \in C^\infty(G; \mathbb{C})$ in Theorem \ref{FinalTheoremSchroedinger1} seem to be rather restrictive, but in fact they are quite natural in the context of infinite dimensional oscillatory integrals whenever representation formulae analogous to \eqref{int-osc-rot} come into play. In particular, in the extensively studied case where $G=\mathbb{R}^{d}$, formulae similar to \eqref{int-osc-rot} can be proved by assuming $\varphi \in C^\infty(\mathbb{R}^{d}; \mathbb{C})$ to be Fourier transform of a compactly supported measure (see, e.g., \cite{AlCaMa,BoMa}).
        In the case of a compact Lie group, given the particular structure of the harmonic analysis there, this condition is equivalent to the requirement that $\varphi$ is a finite linear combination of eigenfunctions of  the Laplace-Beltrami operator or, equivalently in the case considered in the present work, the Casimir operator.
        \item
        Our approach exploits the assumption that the geodesics of the used metric coincide with the integral curves of the left-invariant vector fields -- this is necessary to implement the Cartan development map into a useful form.
        For a left-invariant metric that condition is equivalent to requiring that it is also right-invariant in view of the last statement of Proposition \ref{PROPO6}.
        Under the bi-invariance assumption, the Laplace-Beltrami operator coincides with the Casimir operator due to Proposition \ref{PROPX2DELTA}.
        In our approach these two operators must always coincide. 
    \end{enumerate}
\end{remark}

\begin{remark}\label{remark-scalar-curvature}
    In the formula \eqref{FormulaChernoffOperator}  for the Chernoff operators $S(t)$, the integration variables $x$ are coordinates on the Lie algebra $\mathfrak{g}$ and $\mathrm{d}x$ represents the Lebesgue measure $d{\mathcal{ L} }(x)$ there.
    If we restrict the domain of integration to an open set $U\subset\mathfrak{g}$ where $\exp$ is a diffeomorphism, we obtain an integral on the Lie group $G$ with respect to the push-forward measure of $d{\mathcal{ L} }(x)$  under the action of the exponential map.
    \begin{multline}\label{Eq: local oscillatory integration - motivating formula}
    \frac{1}{(2\pi i t\hbar )^{d/2}}
    \int^o_{U}
    e^{
    \frac{i}{2t\hbar} \|x\|^2}
    e^{-ix\cdot X^R}
    \mathrm{d}x
    \\
    =\frac{1}{(2\pi i t\hbar )^{d/2}}\int^o_{\exp(U)} e^{
    \frac{i}{2t\hbar} \|\exp^{-1}(x)\|^2}
    e^{-i \exp^{-1}(x)\cdot X^R}
    \mathrm{d}(\exp)_*{\cal L}(x)\,.
    \end{multline}
    This formula holds only locally, however, it provides a useful insight on the presence or absence of the scalar curvature in our construction.
    Indeed, the latter is a consequence of the choice of the reference measure on $G$ in the finite dimensional approximations of the measure $d\gamma$ on path space.
    This is ultimately related to the construction of the infinite dimensional oscillatory integral on curved space via the Cartan development map \cite{ELT81}, that is able to encode the geometry of the underlying manifold and faithfully translates  the theory developed for $\mathbb{R}^d$ \cite{AlBr,AlHK77,AlHKMa,ELT} into the non-euclidean setting.
    On the other hand, a curvature term may appear by suitably modifying \eqref{Eq: local oscillatory integration - motivating formula}. Indeed, by replacing $\mathrm{d}(\exp)_*{\cal L}(x)$ with the Haar measure $\mu_G$ on the group $G$, one obtains
    \begin{multline}
     \frac{1}{(2\pi i t\hbar )^{d/2}}\int^o_{\exp^{-1}(U)} e^{
    \frac{i}{2t\hbar} \|\exp^{-1}(x)\|^2}
    e^{-i\exp^{-1}(x)\cdot X^R}
    \mathrm{d}\mu_G(x)
    \\
    =\frac{1}{(2\pi i t\hbar )^{d/2}}\int^o_U e^{
    \frac{i}{2t\hbar} \|x\|^2}
    e^{-ix\cdot X^R}J(x)
    \mathrm{d}x\,,
   \end{multline}
   where $J(x)=\det\Big(\frac{1-e^{-\operatorname{ad}(x\cdot X)}}{\operatorname{ad}(x\cdot X)}\Big)$ is the Jacobian of the exponential map, \textit{cf.} \cite[\S II, Thm. 1.7]{Hel} and \cite{Watanabe}.
   This motivates the definition of a different Chernoff approximation, namely
\begin{equation}\label{tildeS}
     \tilde S(t):= 
     \frac{1}{(2\pi i t\hbar )^{d/2}}\int^o_{\mathbb{R}^d} e^{
    \frac{i}{2t\hbar} \|x\|^2}
    e^{-ix\cdot X^R}J(x)
    \mathrm{d}x
     =\frac{1}{(2\pi i )^{d/2}}\int^o_{\mathbb{R}^d} e^{
    \frac{i}{2} \|x\|^2}
    e^{-i\sqrt{ t\hbar}x\cdot X^R}J(\sqrt{ t\hbar}x)
    \mathrm{d}x\,.
   \end{equation} 
    For $t\downarrow 0$ one has the following asymptotic expansion:
    \begin{align*}
        \frac{1-e^{-\mathrm{ad}(\sqrt{\hbar t}x\cdot X)}}{\mathrm{ad}(\sqrt{\hbar t}x\cdot X)}
        &=1
        -\frac{\sqrt{\hbar t}}{2}\mathrm{ad}(x\cdot X)
        +\frac{\hbar t}{6}\mathrm{ad}^2(x\cdot X)
        +o(t)\,,
    \end{align*}
and, by using the asymptotic approximation formula $$\det(I+\epsilon B)=1+\epsilon \mathrm{Tr}[B]-\frac{\epsilon^2}{2}\left((\mathrm{Tr}[B])^2-\mathrm{Tr}[B^2]\right)+o(\epsilon^2)\,, $$ this yields:
\begin{align*}
J(\sqrt{ t\hbar}x)=\det\left(\frac{1-e^{-\mathrm{ad}(\sqrt{\hbar t}x\cdot X)}}{\mathrm{ad}(\sqrt{\hbar t}x\cdot X)}\right)
=1+\frac{\hbar t}{24}\mathrm{Tr}[\mathrm{ad}^2(x\cdot X)]+o(t)\,.
\end{align*}
where in the last step we applied Proposition \ref{PROP13}.
The term $\mathrm{Tr}[\mathrm{ad}^2(x\cdot X)]$ is strictly related with the Ricci tensor.
Indeed, given an orthonormal basis $\{e_i\}_{i=1}^d$ of $\mathfrak{g}$ one has, \textit{cf.} remark \ref{Rmk: Ricci curvature} in the Appendix:
\begin{align*}
    \mathrm{Tr}[\mathrm{ad}^2(x\cdot X)]
    &=\sum_{i=1}^d\boldsymbol{g}_e(e_i,[x\cdot X, [x\cdot X, e_i]])
    \\
    &=-\sum_{i=1}^d\boldsymbol{g}_e([x\cdot X, e_i],[x\cdot X, e_i])
    =-4\,\mathrm{Ric}(x\cdot X,x\cdot X)\,,
\end{align*}
thus giving
$$J(\sqrt{ t\hbar}x)=1-\frac{\hbar t}{6}\mathrm{Ric}(x\cdot X,x\cdot X)+o(t)\,. $$
By studying the asymptotic behavior of the integral \eqref{tildeS} for $t\downarrow 0$, and using the identities
$$ \int_{\mathbb{R}^d}^ox_j\frac{e^{
    \frac{i}{2} \|x\|^2}}{(2\pi i )^{d/2}}\mathrm{d}x=0\qquad \int_{\mathbb{R}^d}^ox_jx_k\frac{e^{
    \frac{i}{2} \|x\|^2}}{(2\pi i )^{d/2}}\mathrm{d}x=i\delta_{jk}\qquad \forall j,k=1,\dots, d,$$
we eventually get
\begin{align*}
    \tilde S(t)
   &= \int^o_{\mathbb{R}^d}
    e^{-i\sqrt{ t\hbar}x\cdot X^R}
    J(\sqrt{ t\hbar}x)
    \frac{e^{
    \frac{i}{2} \|x\|^2}}{(2\pi i )^{d/2}}
    \mathrm{d}x
    \\
    &=\int_{\mathbb{R}^d}^o 
    \left(I-i\sqrt{ t\hbar}x\cdot X^R-\frac{t\hbar}{2}(x\cdot X^R)^2-\frac{1}{6}t\hbar \mathrm{Ric}(x\cdot X,x\cdot X)\right)\frac{e^{
    \frac{i}{2} \|x\|^2}}{(2\pi i )^{d/2}}\mathrm{d}x +o(t)\\
    &=I{+}i\frac{t\hbar}{2}X^2_{\boldsymbol{g}}-i\frac{1}{6}t\hbar\sum_{j=1}^d\mathrm{Ric}(e_i,e_i) +o(t)\\
    &=I{+}i\frac{t\hbar}{2}\Delta_{\boldsymbol{g}}-i\frac{1}{6}t\hbar R
    +o(t),\end{align*}
where $R $ in the last line denotes the scalar curvature. This shows that the family of operators $\tilde S(t)$ is a Chernoff approximation for the unitary group $\tilde U(t)=e^{it\hbar\left(\frac{\overline{\Delta_{\boldsymbol{g}}}}{2}{-}\frac{1}{6}R\right)}$ on $L^2(G,\mu_G)$.
Similarly, this result shows that the family of operators $(\tilde S_R(t))_t$ defined by
\begin{align*}
    \tilde S_R(t):=\tilde S(t) e^{\frac{i\hbar}{6}Rt}\,.
\end{align*}
is a Chernoff approximation for the unitary group $U(t)=e^{it\hbar \frac{\overline{\Delta_{\boldsymbol{g}}}}{2}}$.

These results show that, analogously to the case of the heat equation studied in \cite{AnDri}, even in the Schr\"odinger equation case the  appearance of the scalar curvature correction term arises from different choices of the reference measure on the finite-dimensional approximations of the path space and it is ultimately related to the particular form of the Radon-Nikodym derivative between them.

\end{remark}
\subsection{Integration of cylinder functions}
In this section we will discuss the Feynman map for a wider class of functions $f\colon\mathcal{H}_{x,t}(G)\to\mathbb{C}$, namely cylinder functions which are obtained as a product of finite energy functions, \textit{cf.} Definition \ref{Def: finite energy functions}.
In a nutshell, the value that a cylinder function $f\colon\mathcal{H}_{x,t}(G)\to\mathbb{C}$ attains on a path $\gamma\in \mathcal{H}_{x,t}(G)$ depends only on the values of $\gamma$ at a finite (fixed) set of times.
Actually, let us consider a finite set of times $0\leq t_1<t_2<\ldots<t_k\leq t$ and a Borel map $g:G^k\to \mathbb{C}$ ---here $G^k$ denotes the $k-$fold cartesian product of $G$, i.e. $G^k:=\underbrace{G\times \ldots\times G}_{ \mbox{$k$ times}}$. The function $f:\mathcal{H}_{x,t}(G)\to \mathbb{C}$ defined as
 \begin{equation}\label{cylinder-function}
 f(\gamma):=g(\gamma(t_1),\ldots,\gamma(t_k) ), \qquad \gamma\in \mathcal{H}_{x,t}(G)\,,
 \end{equation}
 is called a \textbf{cylinder function}. 
In the following, we are going to prove the integration formula \eqref{funz-cyl} for a specific class of cylinder functions, namely those obtained as products of finite energy functions:
\begin{align*}
    f(\gamma)=\phi_1(\gamma(t_1))\cdots\phi_k(\gamma(t_k))\,,
    \qquad
    \phi_1,\ldots,\phi_k\in F_G\,.
\end{align*}
The restriction to this subclass is motivated by the results of the previous section, \textit{cf.} Remark \ref{Rmk: finite energy assumption}.
More precisely, our Ansatz is the following formula for the Feynman map of a function  of the form \eqref{cylinder-function}:
 \begin{align}\label{int-cyl1}
     F_{\mathcal{H}_{x,t}(G)}(f)=\int_{G^k}
     \prod_{j=0}^k K_{t_j-t_{j-1}}(x_{j-1},x_j)
     g(x_1, \ldots, x_k)\mathrm{d}\mu_G(x_1)\cdots \mathrm{d}\mu_G(x_k)\,,
 \end{align}
 where $g(x_1,\ldots,x_k):=\phi_1(x_1)\cdots\phi_k(x_k)$, $\phi_1,\ldots,\phi_k\in F_G$, while we set $t_0=0$ and $x_0=x$ and denote by $K_t\in {\cal D}(G\times G)'$ the \textit{Schwartz kernel} of the unitary operator $U(t)= e^{\frac{i\hbar t}{2}\overline{\Delta_{\boldsymbol{g}}}}$.
 We stress that the existence of $K_t$ follows from the fact that $U(t)$, viewed as a map ${\cal D}(G) \to L^2(G,\mu_G) \subset {\cal D}(G)'$, is sequentially continuous in the relevant  ${\cal D}-{\cal D}'$ topologies so that the Schwartz kernel theorem applies \cite[Thm. 8.2.12]{Hormander_1990}.
  
 On account of the rather explicit form of $g$, the right hand side of \eqref{int-cyl1} takes the following form:
\begin{multline}\label{int-cyl2}
    \int_{G^k}
    \prod_{j=1}^k K_{t_{j}-t_{j-1}}(x_{j-1},x_j)\phi_j(x_j)
    \mathrm{d}\mu_G(x_1)\cdots \mathrm{d}\mu_G(x_k)
     \\
     =\left(U(t_1)\phi_1\cdots U(t_{k-2}-t_{k-1})\phi_{k-1}U(t_k-t_{k-1})\phi_k\right)(x)
\end{multline}

Theorem \ref{teoCherAppr}, see also equation \eqref{Eq: Chernoff approximation - with arbitrary sequence}, allows to approximate the action of the operator $U(s)$, with $s\geq 0$, on a finite energy function $\phi\in F_G$ in the following way
\begin{align*}
    U(s)\phi
    =\lim_ {n\to\infty}S(t/n)^{\lfloor ns/t\rfloor}\phi\,,
    \qquad
    S(t):= 
   \frac{1}{(2\pi i )^{d/2}}
    \int^o_{\mathbb{R}^{d}}
    e^{
    \frac{i}{2} \|x\|^2}
    e^{-i\sqrt{\hbar t} x\cdot X^R}
    \mathrm{d}x\,,
\end{align*}
where $t\geq0$ is arbitrary.
In particular \eqref{int-cyl2} can be rewritten in the following form
 \begin{multline}\label{mult-Cher}
\left(U(t_1)\phi_1\cdots U(t_{k-2}-t_{k-1})\phi_{k-1}U(t_k-t_{k-1})\phi_k\right)(x)\\
=\lim_{n_1\to\infty}\lim_{n_2\to\infty}\dots \lim_{n_k\to\infty}\big(S(t/n_1)^{\lfloor n_1 t_1/t\rfloor}\phi_1 S(t/n_2)^{\lfloor n_2(t_2- t_1)/t\rfloor}\phi_2\cdots \\ \cdots \phi_{k-1}S(t/n_k)^{\lfloor n_k(t_k- t_{k-1})/t\rfloor}\phi_k\big)(x)\,,
 \end{multline}

In the derivation of formula \eqref{mult-Cher}, we have taken advantage of the fact that the set of finite energy functions is an algebra under multiplication and it is closed under the action of the unitary operators $U(t)$, $t\in {\mathbb R}$, \textit{cf.} Proposition \ref{PROPALG}.

By Theorem \ref{teoCherAppr},  the bound $S(t)\leq Ae^{\omega t}$ for all $t\geq 0$ (see Remark \ref{Remark-properties-S(t)}), and  a simple inductive argument, the right hand side of \eqref{mult-Cher} can be replaced by 
 \begin{multline}\label{mult-Cher-1}
\left(U(t_1)\phi_1\cdots U(t_{k-2}-t_{k-1})\phi_{k-1}U(t_k-t_{k-1})\phi_k\right)(x)\\
=\lim_{n\to\infty}\big(S(t/n)^{\lfloor n t_1/t\rfloor}\phi_1 S(t/n)^{\lfloor n(t_2- t_1)/t\rfloor}\phi_2\cdots  \phi_{k-1}S(t/n)^{\lfloor n(t_k- t_{k-1})/t\rfloor}\phi_k\big)(x).
 \end{multline}
\begin{remark}\label{rem-appr-mult-Cher}
    More generally, the representation formula \eqref{mult-Cher-1} can be replaced by
    \begin{multline*}
\left(U(t_1)\phi_1\cdots U(t_{k-2}-t_{k-1})\phi_{k-1}U(t_k-t_{k-1})\phi_k\right)(x)\\
=\lim_{n\to\infty}\big(S(t/n)^{k_n(t_1)}\phi_1 S(t/n)^{k_n(t_2- t_1)}\phi_2\cdots  \phi_{k-1}S(t/n)^{k_n(t_k- t_{k-1})}\phi_k\big)(x).
 \end{multline*}
 where, for each $s\geq 0$, $\{k_n(s)\}_n$ is an increasing sequence of integers such that
     $$\lim_{n\to \infty }\frac{t}{n}k_n(s)=s\,.$$
\end{remark}
Considering now the Feynman map applied to $f$, by its very definition we find
$$F_{\mathcal{H}_{x,t}(G)}(f)=\lim_{n\to \infty}
    (2\pi i \hbar (t/n)^{-1} )^{-nd/2}\int_{\mathbb{R}^{nd}}^o
    e^{\frac{i}{2\hbar}\frac{t}{n}\sum_{j=1}^n\|v_j\|^2}\prod_{j=1}^k\phi_j (\gamma_{x,\boldsymbol{v}}(t_j))
    \mathrm{d}v_1\cdots \mathrm{d}v_n\,. $$

Fixed $n\geq 1$, let us set for any $j=1,\ldots k$ the integers $m_{j,n}:=\lfloor nt_j/t\rfloor$ in such a way that the time $t_j$ belongs to the partition subinterval $[m_{j,n}t/n, (m_{j,n}+1)t/n)$. By Equation \eqref{calcolo} we have:
$$\phi_j (\gamma_{x,\boldsymbol{v}}(t_j))=\prod_{l=1}^{m_{j,n}}e^{-i\frac{t}{n}v_l\cdot X^R}
     e^{-i\left(t_j-m_{j,n}\frac{t}{n}\right)v_{m_{j,n}+1}\cdot X^R}\phi_j(x)\,,$$
hence:
\begin{multline*}
    F_{\mathcal{H}_{x,t}(G)}(f)=\lim_{n\to \infty}
    (2\pi i \hbar (t/n)^{-1} )^{-nd/2}\int_{\mathbb{R}^{nd}}^o
    e^{\frac{i}{2\hbar}\frac{t}{n}\sum_{j=1}^n\|v_j\|^2}\\ \prod_{j=1}^k\prod_{l=1}^{m_{j,n}}e^{-i\frac{t}{n}v_l\cdot X^R}
     e^{-i\left(t_j-m_{j,n}\frac{t}{n}\right)v_{m_{j,n}+1}\cdot X^R}\phi_j(x)
    \mathrm{d}v_1\cdots \mathrm{d}v_n\,.
    \end{multline*}
    By theorem \ref{teo-osc-gau}, the  oscillatory integral above is equal to the Gaussian integral
   \begin{multline*}\lim_{n\to \infty}
    (2\pi   )^{-nd/2}\int_{\mathbb{R}^{nd}}
    e^{-\frac{1}{2}\sum_{j=1}^n\|x_j\|^2} \\ \prod_{j=1}^k\prod_{l=1}^{m_{j,n}}e^{-ie^{i\pi/4}\sqrt{\hbar\frac{t}{n}}x_l\cdot X^R}
     e^{-ie^{i\pi/4}\left(t_j-m_{j,n}\frac{t}{n}\right)\sqrt{\hbar \frac{n}{t}}x_{m_{j,n}+1}\cdot X^R}\phi_j(x)
    \mathrm{d}x_1\cdots \mathrm{d}x_n\,\end{multline*}
which, by the representation \eqref{rapp-gaus-int} for the 1-parameter family of operators $(S(t))_{t\geq 0}$, is equal to
\begin{multline}\label{FM-appr1}\lim_{n\to \infty} S(t/n)^{m_{1,n}}\int_{\mathbb{R}^{d}}\frac{e^{-\frac{1}{2}\|x_{m_{1,n}+1}\|^2}}{ (2\pi   )^{d/2}} e^{-ie^{i\pi/4}\left(t_1-m_{1,n}\frac{t}{n}\right)\sqrt{\hbar \frac{n}{t}}x_{m_{1,n}+1}\cdot X^R}\phi_1 \\ e^{-ie^{i\pi/4}\left((m_{j,n}+1)\frac{t}{n}-t_1\right)\sqrt{\hbar \frac{n}{t}}x_{m_{1,n}+1}\cdot X^R} S(t/n)^{m_{2,n}-m_{1,n}-1}\int_{\mathbb{R}^{d}}\frac{e^{-\frac{1}{2}\|x_{m_{2,n}+1}\|^2}}{ (2\pi   )^{d/2}}  \\ e^{-ie^{i\pi/4}\left(t_2-m_{2,n}\frac{t}{n}\right)\sqrt{\hbar \frac{n}{t}}x_{m_{2,n}+1}\cdot X^R}\phi_2 e^{-ie^{i\pi/4}\left((m_{2,n}+1)\frac{t}{n}-t_2\right)\sqrt{\hbar \frac{n}{t}}x_{m_{2,n}+1}\cdot X^R} S(t/n)^{m_{3,n}-m_{2,n}-1} \cdots \\
\cdots S(t/n)^{m_{k,n}-m_{k-1,n}-1}\int_{\mathbb{R}^{d}}\frac{e^{-\frac{1}{2}\|x_{m_{k,n}+1}\|^2}}{ (2\pi   )^{d/2}} e^{-ie^{i\pi/4}\left(t_k-m_{k,n}\frac{t}{n}\right)\sqrt{\hbar \frac{n}{t}}x_{m_{k,n}+1}\cdot X^R}\\ \phi_k(x) \mathrm{d}x_{m_{1,n}+1}\dots \mathrm{d}x_{m_{k,n}+1}\,.
 \end{multline}
By introducing the sequence of operators $(T_n(s))_n$, $s\geq 0$, on $F_G\times F_G$ defined by
\begin{align}\label{def-Tn}
    T_n(s)(\phi,\psi):=\int_{\mathbb{R}^{d}}\frac{e^{-\frac{1}{2}\|u\|^2}}{ (2\pi   )^{d/2}} e^{-ie^{i\pi/4}\left(s-\lfloor ns/t \rfloor\frac{t}{n}\right)\sqrt{\hbar \frac{n}{t}}u\cdot X^R}\phi e^{-ie^{i\pi/4}\left((\lfloor ns/t \rfloor+1)\frac{t}{n}-s\right)\sqrt{\hbar \frac{n}{t}}u\cdot X^R} \psi du\, ,
\end{align}
representation \eqref{FM-appr1} can be equivalently written as
\begin{multline}\label{FM-appr1-bis}\lim_{n\to \infty} S(t/n)^{m_{1,n}}T_n(t_1)\phi_1 S(t/n)^{m_{2,n}-m_{1,n}-1}T_n(t_2)\phi_2
\\
\cdots S(t/n)^{m_{k,n}-m_{k-1,n}-1}S(t_k-m_kt/n)\phi_k(x)\,,
\end{multline}
where we adopted the convention $T_n(s)(\phi,\psi)\equiv T_n(s)\phi\psi$.

Coming back to the right hand side of \eqref{int-cyl1},  by Remark  \ref{rem-appr-mult-Cher}, we can write the following representation formula 
    \begin{multline}\label{FM-appr2}
\left(U(t_1)\phi_1\cdots U(t_{k-2}-t_{k-1})\phi_{k-1}U(t_k-t_{k-1})\phi_k\right)(x)\\
=\lim_{n\to\infty}\big(S(t/n)^{m_{1,n}}\phi_1 S(t/n)^{m_{2,n}-m_{1,n}-1}\phi_2\cdots  \phi_{k-1}S(t/n)^{m_{k,n}-m_{k-1,n}-1}\phi_k\big)(x).
 \end{multline}
 On the other hand, since for any pair $\phi,\psi$ of finite energy functions and for any  $s\geq 0$ the following holds
 \begin{equation}\label{lim-Tn}
     \lim_{n\to \infty} T_n(s)(\phi,\psi)(x)=\phi \psi (x)\,,
 \end{equation}
 it is easy to prove by an inductive argument that the limits \eqref{FM-appr1} and \eqref{FM-appr2} coincide, namely: 
  \begin{multline}\label{limit-equivalence}
  \lim_{n\to\infty}\Big(S(t/n)^{m_{1,n}}\phi_1 S(t/n)^{m_{2,n}-m_{1,n}-1}\phi_2\cdots  \phi_{k-1}S(t/n)^{m_{k,n}-m_{k-1,n}-1}\phi_k \\ -S(t/n)^{m_{1,n}}T_n(t_1)\phi_1 S(t/n)^{m_{2,n}-m_{1,n}-1}T_n(t_2)\phi_2 \cdots S(t/n)^{m_{k,n}-m_{k-1,n}-1}S(t_k-m_kt/n)\phi_k
  \Big)=0\,.
  \end{multline}
We have thus proved the following result.
\begin{theorem}\label{teo-cyl-funct}
    Let $G$ be a compact Lie group endowed with a bi-invariant metric ${\boldsymbol{g}}$ and let $\Delta_{\boldsymbol{g}}$ denote
the   Laplace-Beltrami operator and $(U(t))_{t\in \bR}$ the group of unitary operators given by $U(t)= e^{\frac{i\hbar t}{2}\overline{\Delta_{\boldsymbol{g}}}  }$. Let $f: \mathcal{H}_{x,t}(G)\to \mathbb{C}$ be a cylinder function of the form $$f(\gamma)=\prod_{j=1}^k\phi_j(\gamma(t_j))\,\qquad \gamma \in \mathcal{H}_{x,t}(G)\,,$$
for some $k\geq 1$, $t_1,\ldots , t_k\in[0,t]$, $\phi_j\in F_G$ for $j=1,\dots, k$. Then the Feynman map $F_{\mathcal{H}_{x,t}(G)}$  applied to the function $f$ is equal to
$$F_{\mathcal{H}_{x,t}(G)}(f)=\left(U(t_1)\phi_1\cdots U(t_{k-2}-t_{k-1})\phi_{k-1}U(t_k-t_{k-1})\phi_k\right)(x)\,.$$
\end{theorem}

  \section{A perturbative solution of the Schr\"odinger equation}\label{sez-Schr-V}
  Let us consider now the Schr\"odinger equation \eqref{SchroedingerM} with a general potential $V\neq 0$ and study the construction of a representation for its strong solution in $L^2(G,\mu_G)$, i.e. a representation formula of the form:
  \begin{align}\label{EF-potentialV}
       \psi (t, x)
       =F_{\mathcal{H}_{x,t}(M)}(f)
       =\widetilde{\int}_{\mathcal{H}_{x,t}(G)} e^{\frac{i}{2\hbar}\int_0^t\boldsymbol{g}(\dot{\gamma}(s),\dot{\gamma}(s))\mathrm{d}s} f(\gamma)\mathrm{d}\gamma\,,
  \end{align}
  with $f:\mathcal{H}_{x,t}(G)\to {\mathbb C}$ being the map defined as
  \begin{equation}\label{map-f}
  f(\gamma)=\psi_0(\gamma (t))e^{-\frac{i}{\hbar}\int_0^tV(\gamma(s))\mathrm{d}s}\,, \qquad \gamma \in\mathcal{H}_{x,t}(G)\,.
  \end{equation} 
  In order to define the relevant Feynman maps, in the following we shall assume that both the initial datum $\psi_0$ and the potential $V$ belong to the algebra $F_G$ of finite energy functions. In particular, the map $V:G\to\mathbb{R}$ is continuous and bounded, hence the operator sum $ { -}\frac{\hbar^2}{2}\overline{\Delta_{\boldsymbol{g}}} { +}V$ is selfadjoint on $D(\overline{\Delta_{\boldsymbol{g}}})\subset L^2(G,\mu_G) $ and generates a strongly continuous unitary group $e^{-\frac{it}{\hbar}(\overline{{ -}\frac{\hbar^2}{2}\Delta_{\boldsymbol{g}} { +}V})}$.  With this notation, formula  \eqref{EF-potentialV} can be rephrased as $F_{\mathcal{H}_{x,t}(M)}(f) = \left(e^{-\frac{it}{\hbar}(\overline{{ -}\frac{\hbar^2}{2}\Delta_{\boldsymbol{g}} { +}V})}\psi_0\right)(x)$, where the equality has to be understood in the $L^2(G,\mu_G)$-sense.
  \\
As we have seen in the previous section, the techniques developed so far can be directly applied only to the case where $V=0$ and the initial datum $ \psi_0:G\to {\mathbb C} $ belongs to the algebra $F_G$ of finite energy functions. More generally, Theorem \ref{teo-cyl-funct} provides the existence and an explicit representation formula for the Feynman map $F_{\mathcal{H}_{x,t}(M)}(f)$ only in the case of particular cylinder functions $f: \mathcal{H}_{x,t}(G)\to\mathbb{C}$  constructed out of finite energy functions, while  the mapping \eqref{map-f} does not fall within those cases. In addition, it is important to point out that 
even if we assume that the potential  $V:G\to {\mathbb R}$ belongs to the set $F_G$, in general its exponential $e^{i V}$ will not have this property. In order to circumvent these issues, we will consider on the one hand the power series expansion of the exponential $e^{-\frac{i}{\hbar}\int_0^tV(\gamma(s))\mathrm{d}s}$ in \eqref{map-f} and the corresponding   expansion of the Feynman map $F_{\mathcal{H}_{x,t}(M)}(f)$:
  \begin{equation}\label{form-pert}
      \sum_{m=0}^\infty \frac{(-i/\hbar)^m}{m!} \widetilde{\int}_{\mathcal{H}_{x,t}(G)} e^{\frac{i}{2\hbar}\|\gamma\|^2_{H_t(\mathbb{R}^{d})}} \left(\int_0^tV(\gamma(s))\mathrm{d}s\right)^m\psi_0(\gamma (t))\mathrm{d}\gamma\, ,
  \end{equation}
and, on the other hand, the perturbative Dyson expansion of the solution $ \psi (t)\equiv e^{-\frac{it}{\hbar}(\overline{{ -}\frac{\hbar^2}{2}\Delta_{\boldsymbol{g}} { +}V})}\psi_0$ in $ L^2(G,\mu_G)$:
\begin{equation}\label{Dyson-series}
      \psi(t,x)= \sum_{m=0}^\infty (-i/\hbar)^m \psi_m(t,x)\,,
  \end{equation}
  with 
  \begin{equation}\label{Dyson-term-m}
      \psi_m(t,x)=\int_{\Delta_m}U(s_1)VU(s_2-s_1 )V\cdots VU(t-s_m)\psi_0(x)\mathrm{d}s_1\cdots \mathrm{d}s_m\,.
  \end{equation}
  In the formula above  $\Delta_m$ denotes the $m-$dimensional simplex $\Delta_m=\{(s_1,\dots, s_m)\in { \mathbb R}^m\, :\, 0\leq s_1\leq \dots \leq s_m\leq t\}$,  $U(t)= e^{\frac{i\hbar t}{2}\overline{\Delta_{\boldsymbol{g}}}  }$ and, with a slight abuse of notation, $V$ denotes the multiplication operator on $L^2(G\mu_G)$ associated to the bounded potential $V$.
 \begin{remark}\label{remark-Dyson-convergence}
  The perturbative expansion \eqref{Dyson-series} is a well-known result of the theory of operator semigroups and their perturbations. Under the assumption of the boundedness of the map $V$, it is is convergent in $L^2(G,\mu_G)$, \textit{cf.} \cite[Thm. X69]{Reed_Simon_1975}, \cite[Ch. IX, Th 2.1]{Kato}. In addition, since $F_G$ is an algebra of smooth functions and the unitary operators $(U(t))_{t\in \bR}$ map $F_G$ onto itself, the function $$x\mapsto U(s_1)VU(s_2-s_1 )V\cdots VU(t-s_m)\psi_0(x)$$ is still an element of $F_G$ and the expansion \eqref{Dyson-series} converges pointwise and provides a classical solution  of Equation \eqref{SchroedingerM}, \textit{cf.} \cite[Thm. 1, Rmk. 6]{Tha1} . \end{remark}
  
In the following we are going to show that every term in the - at this level still formal - series expansion \eqref{form-pert} above is a well defined Feynman map and its value provides the corresponding term in the Dyson perturbative expansion for the solution of the Schr\"odinger equation \eqref{SchroedingerM}. In particular we shall prove the following Ansatz:
\begin{equation}\label{ident-Dyson-1}
    \psi_m(t,x)= \frac{1}{m!}\widetilde{\int}_{\mathcal{H}_{x,t}(G)} e^{\frac{i}{2\hbar}\|\gamma\|_{\mathcal{H}_{x,t}(G)}^2} \left(\int_0^tV(\gamma(s))\mathrm{d}s\right)^m\psi_0(\gamma (t))\mathrm{d}\gamma
\end{equation}
where $\psi_m(t,x)$ is given by \eqref{Dyson-term-m}.

 \begin{remark}  \label{rem-fubini}
  By the results of the previous section, in particular 
  Equations \eqref{int-cyl1} and \eqref{int-cyl2}, the integrand on the right hand side of \eqref{Dyson-term-m} is equal to the Feynman map of a cylinder function, namely:
  \begin{multline*}
      U(s_1)VU(s_2-s_1 )V\cdots VU(t-s_m)\psi_0(x)
      \\
      =\widetilde{\int}_{\mathcal{H}_{x,t}(G)} e^{\frac{i}{2\hbar}\|\gamma\|_{\mathcal{H}_{x,t}(G)}^2} V(\gamma(s_1))\cdots V(\gamma(s_m))\psi_0(\gamma (t))\mathrm{d}\gamma\, .
  \end{multline*}
Hence, the proof of \eqref{ident-Dyson-1} reduces to proving the following Fubini-type theorem
  \begin{multline*}
     \widetilde{\int}_{\mathcal{H}_{x,t}(G)} e^{\frac{i}{2\hbar}\|\gamma\|_{\mathcal{H}_{x,t}(G)}^2} \left(\int_0^tV(\gamma(s))\mathrm{d}s\right)^m\psi_0(\gamma (t))\mathrm{d}\gamma\\
     =\int_0^t\dots\int_0^t \widetilde{\int}_{\mathcal{H}_{x,t}(G)} e^{\frac{i}{2\hbar}\|\gamma\|_{\mathcal{H}_{x,t}(G)}^2} V(\gamma(s_1))\cdots  V(\gamma(s_m)) \psi_0(\gamma (t))\mathrm{d}\gamma \mathrm{d}s_1\dots \mathrm{d}s_m\, ,
  \end{multline*}
  since, by the symmetry of the integrand, the following identity holds:
  \begin{multline*}
      \frac{1}{m!}\int_0^t\dots\int_0^t \widetilde{\int}_{\mathcal{H}_{x,t}(G)} e^{\frac{i}{2\hbar}\|\gamma\|_{\mathcal{H}_{x,t}(G)}^2} V(\gamma(s_1))\cdots  V(\gamma(s_m)) \psi_0(\gamma (t))\mathrm{d}\gamma \mathrm{d}s_1\dots \mathrm{d}s_m\\
      =\int_{\Delta_m} \widetilde{\int}_{\mathcal{H}_{x,t}(G)} e^{\frac{i}{2\hbar}\|\gamma\|_{\mathcal{H}_{x,t}(G)}^2} V(\gamma(s_1))\cdots  V(\gamma(s_m) )\psi_0(\gamma (t))\mathrm{d}\gamma \mathrm{d}s_1\dots \mathrm{d}s_m\,.
  \end{multline*}
 \end{remark} 

\begin{remark}\label{Rmk: comparison with literature}
     We point out that in \cite{Tha1} a similar result is obtained by means of a different technique. As in our case, the author requires that  both the initial datum $\psi_0$ and the potential $V$ belong to the algebra $F_G$ (the set of {\it trigonometric polynomials} according to the terminology he adopted, see also the proof of Proposition \ref{PROPALG} in Appendix \ref{APPENDIXPROOFS}) and provides a mathematical definition for the terms in the formal perturbative expansion
     \begin{align*}
        \sum_{n=0}^\infty \frac{(-i/\hbar)^m}{m!}\int_0^t\dots\int_0^t
        \int_{\gamma(0)=x}
        e^{\frac{i}{2\hbar}\|\gamma\|_{\mathcal{H}_{x,t}(G)}^2} V(\gamma(s_1))\cdots  V(\gamma(s_m))\psi_0(\gamma (t))\mathrm{d}\gamma \mathrm{d}s_1\dots \mathrm{d}s_m\,.     
     \end{align*}
    In our case, this would follow directly by Theorem \ref{teo-cyl-funct}, as explained in Remark \ref{rem-fubini} above. In addition, in the theorem \ref{teo-Fey-Dyson-G} we prove that for any $m\in \bN$ the (non-cylinder) function  $f:{\mathcal{H}_{x,t}(G)}\to {\mathbb C}$ defined as:
\begin{equation}\label{map-f-Vn}
    f(\gamma):=\left(\int_0^tV(\gamma(s))\mathrm{d}s\right)^m\psi_0(\gamma (t))\,.
\end{equation}
is integrable according to Definition \ref{def-Fey-map-M}  and the Ansatz \eqref{ident-Dyson-1} holds true. 
\end{remark}   
The main technical result of this section is the following theorem.
\begin{theorem}\label{teo-Fey-Dyson-G}Let $G$ be a compact Lie group endowed with a bi-invariant metric ${\boldsymbol{g}}$ and let $\Delta_{\boldsymbol{g}}$ denote
the   Laplace-Beltrami operator and $\{U(t)\}_{t\in \bR}$ the group of unitary operators given by $U(t)= e^{\frac{i\hbar t}{2}\overline{\Delta_{\boldsymbol{g}}}  }$. Let $f: \mathcal{H}_{x,t}(G)\to \mathbb{C}$ be a function of the form \eqref{map-f-Vn}, for some finite energy functions $\psi_0,V\in F_G$ on $G$. Then the Feynman map $F_{\mathcal{H}_{x,t}(G)}$  applied to the function $f$ is equal to
    $$ m!\int_{\Delta_m}U(s_1)VU(s_2-s_1 )V\cdots VU(t-s_m)\psi_0(x)\mathrm{d}s_1\cdots \mathrm{d}s_m\, .$$
\end{theorem}
\begin{proof}
    By definition, the Feynman map of the function \eqref{map-f-Vn} is given by
\begin{align}\label{int-V-n-1}
    F_{\mathcal{H}_{x,t}(G)}(f)=\lim_{n\to \infty}
    (2\pi i \hbar (t/n)^{-1} )^{-nd/2}\int_{\mathbb{R}^{nd}}^o
    e^{\frac{i}{2\hbar}\frac{t}{n}\sum_{j=1}^n\|v_j\|^2}f (\gamma_{x,\boldsymbol{v}})
    \mathrm{d}v_1\cdots \mathrm{d}v_n\,,
\end{align}
provided that the limit exists, where $\gamma_{x,\boldsymbol{v}}$ is the piecewise-geodesic path on $G$ starting at $x\in G$ and associated parameters $\boldsymbol{v}=(v_1,\ldots,v_n)$ (see Equation \eqref{gamma-x-v}). 
By setting $\delta \equiv t/n$, the finite dimensional oscillatory integrals appering on the right hand side of \eqref{int-V-n-1} can be written in the following form
\begin{multline}\label{int-V-n-2}
  \int_{\mathbb{R}^{nd}}^o
    \frac{e^{\frac{i}{2\hbar}\delta\sum_{j=1}^n\|v_j\|^2}}{(2\pi i \hbar \delta ^{-1} )^{nd/2}}\left(\sum_{k=0}^{n-1}\int_{k\delta}^{(k+1)\delta}V(\gamma_{x,\boldsymbol{v}}(s))\mathrm{d}s\right)^m\psi_0(\gamma_{x,\boldsymbol{v}}(t))
    \mathrm{d}v_1\cdots \mathrm{d}v_n  \\
    =\sum_{k_1,\dots k_m=0}^{n-1}\int_{\mathbb{R}^{nd}}^o
    \frac{e^{\frac{i}{2\hbar}\delta\sum_{j=1}^n\|v_j\|^2}}{(2\pi i \hbar \delta ^{-1} )^{nd/2}}\psi_0(\gamma_{x,\boldsymbol{v}}(t))\prod_{l=1}^m\left(\int_{k_l\delta}^{(k_{l}+1)\delta}V(\gamma_{x,\boldsymbol{v}}(s))\mathrm{d}s\right)\mathrm{d}v_1\cdots \mathrm{d}v_n\,.
\end{multline}
By Equation \eqref{calcolo},
each oscillatory integral appearing in the sum on the right hand side of the identity above can be equivalently written as:
\begin{multline*}
\int_{\mathbb{R}^{nd}}^o
    \frac{e^{\frac{i}{2\hbar}\delta\sum_{j=1}^n\|v_j\|^2}}{(2\pi i \hbar \delta ^{-1} )^{nd/2}}\prod_{j=1}^{n}e^{-i\delta v_j\cdot X^R}\psi_0(x)\\ \prod_{l=1}^m\left(\int_{k_l\delta}^{(k_{l}+1)\delta}\prod_{j_l=1}^{k_l}e^{-i\delta v_{j_l}\cdot X^R} e^{-i\left(s_l-k_l\delta\right)v_{k_{l}+1}\cdot X^R}V(x)\mathrm{d}s_l\right) \mathrm{d}v_1\cdots \mathrm{d}v_n\, ,
    \end{multline*}
    which can be turned into an absolutely convergent Gaussian integral by exploiting the same argument used in the proof of Theorem \ref{teo-osc-gau}, thus obtaining:
\begin{multline*}
\int_{\mathbb{R}^{nd}}
    \frac{e^{-\frac{1}{2\hbar}\delta\sum_{j=1}^n\|v_j\|^2}}{(2\pi i \hbar \delta ^{-1} )^{nd/2}}\prod_{j=1}^{n}e^{-ie^{i\pi/4}\delta v_j\cdot X^R}\psi_0(x)\\ \prod_{l=1}^m\left(\int_{k_l\delta}^{(k_{l}+1)\delta}\prod_{j_l=1}^{k_l}e^{-ie^{i\pi/4}\delta v_{j_l}\cdot X^R} e^{-ie^{i\pi/4}\left(s_l-k_l\delta\right)v_{k_{l}+1}\cdot X^R}V(x)\mathrm{d}s_l\right) \mathrm{d}v_1\cdots \mathrm{d}v_n\, .
    \end{multline*}
    By Fubini theorem the latter is equal to:    
    \begin{multline*}
\int_{k_1\delta}^{(k_{1}+1)\delta}\dots\int_{k_m\delta}^{(k_{m}+1)\delta}\int_{\mathbb{R}^{nd}}
    \frac{e^{-\frac{1}{2\hbar}\delta\sum_{j=1}^n\|v_j\|^2}}{(2\pi  \hbar \delta ^{-1} )^{nd/2}}\prod_{j=1}^{n}e^{-ie^{i\pi/4}\delta v_j\cdot X^R}\psi_0(x)\\ \prod_{l=1}^m\left(\prod_{j_l=1}^{k_l}e^{-ie^{i\pi/4}\delta v_{j_l}\cdot X^R} e^{-ie^{i\pi/4}\left(s_l-k_l\delta\right) v_{k_{l}+1}\cdot X^R}V(x)\right) \mathrm{d}v_1\cdots \mathrm{d}v_n\,\mathrm{d}s_1\dots \mathrm{d}s_m \,.
    \end{multline*}
    By summing over all possible values of $k_1,\dots k_m$ and exploiting the symmetry of the integrand over permutation of the variables $s_1,\ldots, s_m$,one can easily see that the oscillatory integral \eqref{int-V-n-2} turns out to be equal to
    \begin{equation}\label{int-V-n-3}
         m!\int_{\Delta_m} g_n(s_1,\dots,s_m) \mathrm{d}s_1\dots \mathrm{d}s_m
    \end{equation}
    where 
     \begin{multline}\label{funz-g}
       g_n(s_1,\dots,s_m):= \int_{\mathbb{R}^{nd}}
    \frac{e^{-\frac{1}{2\hbar}\delta\sum_{j=1}^n\|v_j\|^2}}{(2\pi  \hbar \delta ^{-1} )^{nd/2}} \prod_{j=1}^{n}e^{-ie^{i\pi/4}\delta v_j\cdot X^R}\psi_0(x)\\ \prod_{l=1}^m\left(\prod_{j_l=1}^{\lfloor s_l/\delta\rfloor}e^{-ie^{i\pi/4}\delta v_{j_l}\cdot X^R} e^{-ie^{i\pi/4}\left(s_l-\lfloor s_l/\delta\rfloor\delta \right) v_{\lfloor s_l/\delta\rfloor+1}\cdot X^R}V(x)\right)
    \mathrm{d}v_1\cdots \mathrm{d}v_n \,.
    \end{multline}
 By  reorganizing the factors appearing in the integral above, taking into account that $s_1<s_2<\dots <s_m<t$ and  that for $n$ sufficiently large the following holds: 
 \begin{equation}\label{cond-times}
     \lfloor s_j/\delta\rfloor\neq \lfloor s_{j+1}/\delta\rfloor\qquad \forall j=1,\ldots, m
 \end{equation}
 we obtain:
 \begin{multline*}
     \prod_{j=1}^{n}e^{-ie^{i\pi/4}\delta v_j\cdot X^R}\psi_0(x)\prod_{l=1}^m\left(\prod_{j_l=1}^{\lfloor s_l/\delta\rfloor}e^{-ie^{i\pi/4}\delta v_{j_l}\cdot X^R} e^{-ie^{i\pi/4}\left(s_l-\lfloor s_l/\delta\rfloor\delta \right) v_{\lfloor s_l/\delta\rfloor+1}\cdot X^R}V(x)\right)\\
     =\prod_{j=1}^{\lfloor s_1/\delta\rfloor}e^{-ie^{i\pi/4}\delta v_j\cdot X^R}
     e^{-ie^{i\pi/4}\left(s_1-\lfloor s_1/\delta\rfloor\delta \right) v_{\lfloor s_1/\delta\rfloor+1}\cdot X^R}V e^{-ie^{i\pi/4}\left(\delta(\lfloor s_1/\delta\rfloor+1)-s_1 \right) v_{\lfloor s_1/\delta\rfloor+1}\cdot X^R}\\ 
     \prod_{j=\lfloor s_1/\delta\rfloor+2}^{\lfloor s_2/\delta\rfloor}e^{-ie^{i\pi/4}\delta v_j\cdot X^R}
     e^{-ie^{i\pi/4}\left(s_2-\lfloor s_2/\delta\rfloor\delta \right) v_{\lfloor s_2/\delta\rfloor+1}\cdot X^R}V \cdots\\
     \cdots \prod_{j=\lfloor s_{m-1}/\delta\rfloor+2}^{\lfloor s_m/\delta\rfloor}e^{-ie^{i\pi/4}\delta v_j\cdot X^R}
     e^{-ie^{i\pi/4}\left(s_m-\lfloor s_m/\delta\rfloor\delta \right) v_{\lfloor s_m/\delta\rfloor+1}\cdot X^R}V\\ e^{-ie^{i\pi/4}\left(\delta (\lfloor s_m/\delta\rfloor+1)-s_m \right) v_{\lfloor s_m/\delta\rfloor+1}\cdot X^R} \prod_{j=\lfloor s_{m}/\delta\rfloor+2}^{n}e^{-ie^{i\pi/4}\delta v_j\cdot X^R}  \psi_0(x)
 \end{multline*}
 Hence, in this case the function \eqref{funz-g}  can be written as:
  \begin{multline}\label{funz-g-1}
        g_n(s_1,\dots, s_m)= S(t/n)^{\lfloor n s_1/t\rfloor}T_n(s_1)V S(t/n)^{\lfloor n s_2/t\rfloor-\lfloor n s_1/t\rfloor-1}T_n(s_2)V\cdots\\ \cdots S(t/n)^{\lfloor n s_m/t\rfloor-\lfloor n s_{m-1}/t\rfloor-2}T_n(s_m)VS(t/n)^{n-\lfloor s_{m}/\delta\rfloor-2}\psi_0(x)\,, 
    \end{multline}
    where $T_n$ was defined by \eqref{def-Tn}.
Hence, by applying Equation \eqref{FM-appr2} and Equation \eqref{limit-equivalence}, we obtain the following convergence result
\begin{equation*}
    \lim_{n\to \infty}g_n(s_1,\dots, s_m)=U(s_1)VU(s_2-s_1 )V\cdots VU(t-s_m)\psi_0(x)\,.
\end{equation*}
 By dominated convergence theorem, 
  we eventually get Equation \eqref{ident-Dyson-1}.

 \end{proof}
 We finally have the main result, whose proof is a direct consequence of Remark \ref{remark-Dyson-convergence} and Theorem \ref{teo-Fey-Dyson-G}.
 \begin{theorem}\label{Thm: Feynman map with potential}
     Under the assumptions of Theorem \ref{teo-Fey-Dyson-G}, the solution of the Schr\"odinger equation \eqref{SchroedingerM} is given by the power series \eqref{form-pert}. The series converges in $L^2(G,\mu_G)$ and pointwise, thus providing both a solution in $L^2(G,\mu_G)$  and a classical solution of \eqref{SchroedingerM}.
 \end{theorem}
 
 \section{Conclusions}  

The present paper develops the construction of infinite-dimensional oscillatory integrals on the space of paths of a compact Lie group with a bi-invariant metric, following the theory established in \cite{AlHK77,AlHKMa,ELT,ELT81}. This functional integral provides a rigorous mathematical formulation of the Feynman path integral representation \eqref{feynmanM} for the perturbative solution of the Schr\"odinger equation \eqref{SchroedingerM}, in the case where both the initial datum and the potential belong to the class of {\em finite energy functions}. The restriction to Lie groups with bi-invariant metrics and to finite energy functions enables the derivation of explicit and tractable formulae for the action of the Cartan maps on the path space, thereby furnishing a tool for studying infinite-dimensional limits. Our results can be the starting point for the application of the semiclassical asymptotics of the solution of Eq. \eqref{SchroedingerM} in the limit $\hbar\downarrow 0$ via the general theory of stationary phase method for infinite dimensional oscillatory integrals developed in \cite{AlHK77,AlBr}. Moreover, the results of Section 5 could be extended to the case of time dependent potentials, by generalizing the techniques applied in \cite{AlMa04}

Our techniques, however, cannot be directly extended to general Riemannian manifolds, where explicit and manageable expressions for the action of the Cartan map on smooth paths are unavailable. Nevertheless, analogous results can be obtained for manifolds $M$ that admit an explicit characterization of geodesics together with suitable symmetry properties. A first example of such an extension is given in \cite{DrMaPi}, where Feynman and Wiener integrals for the solution of the Schr\"odinger and the  heat equations on the Weyl–Heisenberg group are investigated.

 \section*{Acknowledgments}  This work 
 has been written within the activities of INdAM-GNAMPA and INdAM-GNFM.
 
\appendix 
\section{Proof of some propositions}\label{APPENDIXPROOFS}

In this section we briefly discuss the proofs of some well-known results on Lie groups which were recalled and used in the paper.

\begin{proposition}\label{PROPGEOD}  If $G$ is a   Lie group with a bi-invariant Riemannian metric $\boldsymbol{g}$, then
\begin{align}
    \boldsymbol{g}([\widetilde{X},\widetilde{Y}], \widetilde{Z}) + \boldsymbol{g}(\widetilde{Y}, [\widetilde{X},\widetilde{Z}]) =0
    \qquad \forall \widetilde{X},\widetilde{Y},\widetilde{Z} \in \mathfrak{g}^L\,.
    \label{USEFUL}
\end{align}
\end{proposition}

\begin{proof}
    The proof is based on \cite[\S 11, Lemma 3]{ON}.
    To begin with we observe that, on account of the bi-invariance of $\boldsymbol{g}$ and the left-invariance of $\widetilde{X},\widetilde{Y},\widetilde{Z}$, equation \eqref{USEFUL} is equivalent to
    \begin{align}
        \label{Eq: infinitesimal bi-invariance at the identity}
        \boldsymbol{g}_e([X,Y],Z) + \boldsymbol{g}_e(Y,[X,Z]) =0\,,
    \end{align}
    where $X,Y,Z\in\mathfrak{g}$ are the elements of the Lie algebra associated with $\widetilde{X},\widetilde{Y},\widetilde{Z}$.
    To prove equation \eqref{Eq: infinitesimal bi-invariance at the identity} we set, for $x\in G$, $C_x\colon G\to G$ defined by $C_x(y):=xyx^{-1}=(R_{x^{-1}}\circ L_x)(y)$.
    Notice that, for all $y\in G$, the differential $(\mathrm{d}C_x)_y\colon T_yG\to T_{xyx^{-1}}G$ is given by $(\mathrm{d}C_x)_y=(\mathrm{d}R_{x^{-1}})_x\circ(\mathrm{d}L_x)_y$.
    We set $\mathrm{Ad}_x:=(\mathrm{d}C_x)_e\colon\mathfrak{g}\to\mathfrak{g}$.
    In particular we have $\mathrm{Ad}_x=(\mathrm{d}R_{x^{-1}})_x\circ(\mathrm{d}L_x)_e$.

    Since $\boldsymbol{g}$ is bi-invariant we have
    \begin{align*}
        \boldsymbol{g}_e(\operatorname{Ad}_xY,\operatorname{Ad}_xZ)
        &=\boldsymbol{g}_e((\mathrm{d}R_{x^{-1}})_x\circ(\mathrm{d}L_x)_eY,(\mathrm{d}R_{x^{-1}})_x\circ(\mathrm{d}L_x)_eZ)
        \\
        &=\boldsymbol{g}_x((\mathrm{d}L_x)_eY,(\mathrm{d}L_x)_eZ)
        =\boldsymbol{g}_e(Y,Z)\,.
    \end{align*}
    Moreover, we recall that the flow of $\widetilde{X}\in\mathfrak{g}^L$ is given by $R_{x(t)}$, where $x(t):=\exp(tX)$.
    It follows that
    \begin{align*}
        [X,Y]=[\widetilde{X},\widetilde{Y}]_e
        =\lim_{t\to 0}\frac{(\mathrm{d}R_{x(-t)})_{x(t)}\widetilde{Y}_{x(t)}-Y}{t}
        =\lim_{t\to 0}\frac{\mathrm{Ad}_{x(t)}Y-Y}{t}
        =\frac{\mathrm{d}}{\mathrm{d}t}\operatorname{Ad}_{x(t)}Y\Big|_{t=0}\,,
    \end{align*}
    where $\widetilde{Y}_{x(t)}=(\mathrm{d}L_{x(t)})_eY$ by left-invariance.
    Thus, we find
    \begin{align*}
        \boldsymbol{g}_e([X,Y],Z) + \boldsymbol{g}_e(Y,[X,Z])
        &=\boldsymbol{g}_e(\frac{\mathrm{d}}{\mathrm{d}t}\operatorname{Ad}_{x(t)}Y\Big|_{t=0},Z)
        +\boldsymbol{g}_e(Y,\frac{\mathrm{d}}{\mathrm{d}t}\operatorname{Ad}_{x(t)}Z\Big|_{t=0})
        \\
        &=\frac{\mathrm{d}}{\mathrm{d}t}\boldsymbol{g}_e(\mathrm{Ad}_{x(t)}Y,\operatorname{Ad}_{x(t)}Z)\Big|_{t=0}
        =0\,,
    \end{align*}
    where in the last line we used the $\operatorname{Ad}_x$-invariance of $\boldsymbol{g}_e$.
\end{proof}

\begin{corollary}\label{COROLLARIO4}
    Let $G$ be a Lie group with left-invariant Riemannian metric $\boldsymbol{g}$ and let $\nabla$ be the associated Levi-Civita connection.
    Then
    \begin{align}\label{Eq: identity for left-invariant Levi-Civita connection}
        \boldsymbol{g}(Z,\nabla_XY)
        =\frac{1}{2}\boldsymbol{g}(Z,[X,Y])
        +\frac{1}{2}\mathcal{L}_Z(\boldsymbol{g})(X,Y)
        \qquad
        \forall X,Y,Z\in\mathfrak{g}^L\,,
    \end{align}
    where $\mathcal{L}_Z$ denotes the Lie derivative along $Z$.
    Moreover, the following identity
    \begin{align}
        \nabla_X Y = \frac{1}{2}[X,Y]
        \qquad
        \forall X,Y\in\mathfrak{g}^L
        \label{centralID}\:,
    \end{align}
    holds if and only if $\boldsymbol{g}$ is bi-invariant.
\end{corollary}

\begin{proof}    For all $X,Y,Z\in\Gamma(G)$ we have
    \begin{multline*}
        X\boldsymbol{g}(Y,Z)
        +Y\boldsymbol{g}(X,Z)
        -Z\boldsymbol{g}(Y,Z)
        \\
        =\boldsymbol{g}(\nabla_XY+\nabla_YX,Z)
        +\boldsymbol{g}(Y,[X,Z])
        +\boldsymbol{g}(X,[Y,Z])\,.
    \end{multline*}
    Moreover, if $X,Y\in\mathfrak{g}^L$ then $\boldsymbol{g}_p(X_p,Y_p)=\boldsymbol{g}_e(X_e,Y_e)$, therefore $Z\boldsymbol{g}(X,Y)=0$ for all $Z\in\Gamma(TG)$.
    Thus, if $X,Y,Z\in\mathfrak{g}^L$, the previous equality reduces to
    \begin{align*}
        \boldsymbol{g}(\nabla_XY+\nabla_YX,Z)
        =-\boldsymbol{g}(Y,[X,Z])
        -\boldsymbol{g}(X,[Y,Z])
        =\mathcal{L}_Z(\boldsymbol{g})(X,Y)\,.
    \end{align*}
    Since $\nabla_XY-\nabla_YX=[X,Y]$, the latter equality proves \eqref{Eq: identity for left-invariant Levi-Civita connection}.
    Moreover, Equation \eqref{centralID} holds true if and only if $\mathcal{L}_Z(\boldsymbol{g})(X,Y)=0$ for all $X,Y,Z\in\mathfrak{g}^L$.
    Since $\mathfrak{g}^L\simeq\mathfrak{g}\simeq T_pG$, $p\in G$, the latter condition is equivalent to $\mathcal{L}_Z(\boldsymbol{g})=0$ for all $Z\in\mathfrak{g}^L$.
    This implies that $\boldsymbol{g}$ is invariant under the flow $\Phi^Z$ induced by any left-invariant vector field $Z\in\mathfrak{g}^L$.
    Since $\Phi^Z_t(x)=x\exp(tZ)=R_{\exp(tZ)}x$, the condition $(\Phi^Z_t)^*\boldsymbol{g}=\boldsymbol{g}$ for all $Z\in\mathfrak{g}^L$ is equivalent to right-invariance of $\boldsymbol{g}$.
    Thus, Equation \eqref{centralID} holds if and only if $\boldsymbol{g}$ is bi-invariant.
    \end{proof}
\begin{remark}\label{Rmk: Ricci curvature}
    Out of equation \eqref{centralID} one obtains neat expressions for the Riemann curvature tensor and the Ricci tensor.
    In particular the curvature tensor of the connection is given by
    \begin{align*}
	R(X,Y)Z=
	\Big[\nabla_{[X,Y]}
	-(\nabla_X\nabla_Y
        -\nabla_Y\nabla_X)\Big]Z
	=\frac{1}{4}[[X,Y],Z]\,,
\end{align*}
where $X,Y,Z\in\mathfrak{g}^L$.
The Ricci tensor is given by
\begin{align*}
	\operatorname{Ric}(X,Y)
	=\sum_{i=1}^d\boldsymbol{g}(R(X,e_i)Y,e_i)
	=\frac{1}{4}\sum_{i=1}^d\boldsymbol{g}([X,e_i],[Y,e_i])\,,
\end{align*}
where $\{e_i\}_{i=1}^d$ is an orthonormal basis of $\mathfrak{g}$ while in the last equality we used the bi-invariance of $\boldsymbol{g}$.
In particular, the scalar curvature is constant and given by
\begin{align*}
	R=\frac{1}{4}\sum_{i,j=1}^d\boldsymbol{g}([e_j,e_i],[e_j,e_i])\,.
\end{align*}
\end{remark}

\noindent {\bf Proof of Proposition \ref{PROPO6}}.
(a) 
Take $x\in G$ and $X\in \mathfrak{g}$. We consider $\gamma(t):=L_x\exp(tX)$ where $t\in \mathbb{R}$.
By construction,  $\Dot{\gamma}(t)=(\mathrm{d}L_x)_{\exp(tX)}\widetilde{X}(\exp(tX))=\widetilde{X}(\gamma(t))$ and therefore
\begin{align*}
    \nabla_{\Dot{\gamma}(t)}\dot{\gamma}(t)
    =\nabla_{\widetilde{X}}\widetilde{X}|_{\gamma(t)}
    =\frac{1}{2}[\widetilde{X},\widetilde{X}]|_{\gamma(t)}
    =0\,,
\end{align*} where we exploited (\ref{centralID}).
We have established that every curve $\mathbb{R}\ni t \mapsto L_x \exp(tX)$ is a  maximal geodesic with domain given by the whole real line exiting $x$ at $t=0$ with initial tangent vector $\widetilde{X}(x)$. Notice that, varying $X\in T_eG$, the vectors $\widetilde{X}(x)$
vary in the whole tangent space $T_xG$ because $(dL_x)_e$ is an isomorphism of vector spaces.
In view of the (existence and) uniqueness theorem of the Cauchy problem for the geodesic equation in $TG$, this result also entails that all $\boldsymbol{g}$-geodesics have the form above.  

(b)  Since $\Dot{\gamma}(t)=\widetilde{X}(\gamma(t))$ for every geodesic $\gamma$,
the proof of (a) also demonstrated that geodesics of a bi-invariant metric  are integral curves of left invariant vector fields.
In other words, if $\Phi^Z$ denotes the flow of a vector field $Z$ (thus it is complete if $Z$ is left-invariant), 
a $\boldsymbol{g}$-geodesic starting at $x\in G$ for $t=0$ can be written as
$\mathbb{R} \ni t\mapsto\Phi^{\widetilde{X}}_t(x)$.  Conversely, a complete integral curve $\mathbb{R} \ni t\mapsto\Phi^{\widetilde{X}}_t(x)$ of a left-invariant vector field $\widetilde{X}$ passing through $x\in G$ at $t=0$ can be re-written in the form $L_x\exp(tX)$ (in view of the uniqueness property of the Cauchy problem for the integral curves of $\widetilde{X}$)
and thus it is a $\boldsymbol{g}$-geodesic with maximal domain for (a). The last identity arises from
    $\Phi^{\widetilde{X}}_t(x)
    =x\exp(tX)=R_{\exp(tX)}x$.

Let us prove the last statement of Proposition \ref{PROPO6}.
Suppose that every complete integral curve $\gamma(t):=L_x\exp(tX)$ of every left-invariant field $\widetilde{X}$ is a ${\boldsymbol{g}}$-geodesic. Therefore, taking (\ref{Eq: identity for left-invariant Levi-Civita connection}) into account, for every $\widetilde{Z}\in \mathfrak{g}^L$,
\begin{align*}
  0=  {\boldsymbol{g}}(\widetilde{Z},\nabla_{\Dot{\gamma}(t)}\dot{\gamma}(t))
    ={\boldsymbol{g}} (\widetilde{Z},\nabla_{\widetilde{X}}\widetilde{X})
    =\frac{1}{2}{\boldsymbol{g}}(\widetilde{Z},[\widetilde{X},\widetilde{X}]) + \mathcal{L}_{\widetilde{Z}}(\boldsymbol{g})(\widetilde{X},\widetilde{X})\:.
\end{align*}
Therefore
$(\mathcal{L}_{\widetilde{Z}})_x(\boldsymbol{g})(\widetilde{X},\widetilde{X})=0$. This identity is valid at each point and for every left-invariant field $\widetilde{X}$.  By polarisation of ${\boldsymbol{g}}$, we immediately achieve
$$\mathcal{L}_{\widetilde{Z}}(\boldsymbol{g})(\widetilde{X},\widetilde{Y})=0\:, \quad \forall \widetilde{X},\widetilde{Y},\widetilde{Z} \in \mathfrak{g}^L\:.$$
As already observed in proof  of Corollary \ref{COROLLARIO4} this requirement is equivalent to the right-invariance of ${\boldsymbol{g}}$. \hfill $\Box$\\

\noindent {\bf Proof of Proposition \ref{PROP13}}. The first statement was established in
\cite[Lem. 6.3]{Milnor-76}. Regarding the second statement, observe that the condition  (\ref{USEFUL}) which holds when there is a bi-invariant metric $\boldsymbol{g}$ can be written $c_{sjk}+c_{skj}=0$ so that we have 
$g^{sk}c_{sjk}+g^{sk}c_{skj}=0$, which boils down to 
$2 {c_{kj}}^k=0$. This is Milnor's condition since due to the antisymmetry in the lower indices. \hfill $\Box$\\

\noindent {\bf Proof of Corollary \ref{remHAAR}}.
${c_{ik}}^k=0$ is satisfied as a consequence of (\ref{USEFUL}) written in components and $\mu_{\boldsymbol{g}}$ is bi-invariant by construction.
\hfill $\Box$\\

\noindent {\bf Proof of Proposition \ref{PROPCASIMIR}}.
The thesis is evidently equivalent to 
$[X^2_{\boldsymbol{g}}, \widetilde{X}_s] =0\quad s=1,\ldots, d$
where $[\cdot,\cdot]$ indicates the commutator of differential  operators on $C^\infty(G;\mathbb{C})$.
The written identity can be expanded to
$g^{ij}_e {c_{is}}^k\widetilde{X}_k\widetilde{X}_j + g^{ij}_e {c_{sj}}^k\widetilde{X}_i \widetilde{X}_k=0$ that is implied by  
$c_{jsk}+c_{skj}=0$.
In turn, in view of the antisymmetry in the first two indices, the identity above is the same as
$c_{sjk}+c_{skj}=0$.
This latter identity is nothing but (\ref{USEFUL}) in components, which is true because $\boldsymbol{g}$ is bi-invariant. \hfill $\Box$\\

\noindent {\bf Proof of Proposition \ref{PROPX2DELTA}}.
Consider a basis $X_1,\ldots,X_{d}$ of $\mathfrak{g}$.
Since the vector fields $\widetilde{X}_k$ are left invariant and the metric is bi-invariant, then for every $x\in G$ it holds $\boldsymbol{g}_x(\widetilde{X}_r(x),\widetilde{X}_s(x))=\boldsymbol{g}_e(X_r,X_s)$, that is, $(g_x)_{ab} X^a_r(x)X^b_s(x) = (g_e)_{rs}$ where $X_r^a(h)$ are the components of $\widetilde{X}_r(x)$ in a local chart around $x$. As a consequence
$g^{hk}_e  X^a_h(x)  X^b_k(x) = g^{ab}_x$ is also valid.
For $f\in C^\infty(G; \mathbb{C})$, in local coordinates 
$$X^2_{\boldsymbol{g}} f= g^{hk}_e \widetilde{X}_h \widetilde{X}_k f =
g^{hk}_e \nabla_{\widetilde{X}_h } \nabla_{\widetilde{X}_k } f= g^{hk}_e  X^a_h \nabla_a X^b_k \nabla_b f$$
$$= g^{hk}_e  X^a_h (\nabla_a X^b_k) \nabla_b f + g^{hk}_e  X^a_h  X^b_k \nabla_a (d f)_b \:.$$
Since $g^{hk}_e  X^a_h(x)  X^b_k(x) = g^{ab}_x$, the found result can be rearranged to
$$X^2_{\boldsymbol{g}} f = g^{hk}_e  X^a_h (\nabla_a X^b_k) \nabla_b f + \Delta_{\boldsymbol{g}}f = g^{hk} \nabla_{\nabla_{\widetilde{X}_h}\widetilde{X}_k} f + \Delta_{\boldsymbol{g}}f = \left(g_e^{hk}\nabla_{\widetilde{X}_h}\widetilde{X}_k\right) f +  \Delta_{\boldsymbol{g}}f\:.$$
Taking advantage of (\ref{centralID}), which holds because $\boldsymbol{g}$ is bi-invariant,  we have
that $$ \left(g_e^{hk}\nabla_{\widetilde{X}_h}\widetilde{X}_k\right)= \frac{1}{2} g_e^{hk} {c_{hk}}^l\widetilde{X}_l =0$$
since ${c_{hk}}^l=-{c_{kh}}^l$ whereas $g^{hk}=g^{kh}$. In summary
$X_{\boldsymbol{g}}^2f = \Delta_{\boldsymbol{g}}f$ concluding the proof. \hfill $\Box$\\

\noindent {\bf Proof of Corollary \ref{COROLLARIO12}}.
$\Delta_{\boldsymbol{g}}$ commutes with the left-invariant vector fields because this is a property of $X^2_{\boldsymbol{g}}$.
Regarding right-invariant vector fields, we can construct an analogous Casimir operator $X^2_R$ using right-invariant vector fields and referring to the Lie algebra $(\mathfrak{g}, [\cdot,\cdot]')$ where $[X,Y]'=-[X,Y]$.
This latter Lie algebra is isomorphic to $(\mathfrak{g},[\cdot,\cdot])$ by $X\mapsto -X$.
The statement of Proposition \ref{PROPCASIMIR} can also be proved for $X^2_R$ and the Lie structure arising from $[\cdot,\cdot]'$.
Indeed, for a given basis $X_1,\ldots, X_{d}$ of $\mathfrak{g}$ the structure constants of $[\cdot,\cdot]'$ are ${c'_{ij}}^k= -{c_{ij}}^k$, thus, $c'_{sjk}+c'_{skj}=0$ holds and according to the very proof of Proposition \ref{PROPCASIMIR}, this condition suffices to prove that $X^2_R$ commutes with all right-invariant vector fields.
Finally, referring to the \textit{same} basis of $\mathfrak{g}$, $X_1,\ldots, X_{d}$, used in the proof of Proposition \ref{PROPX2DELTA}, the same statement  can be proved to be true also for $X^2_R$ and the Lie structure $[\cdot,\cdot]'$ induced by the right invariant fields with the same proof,
 obtaining that $\Delta_{\boldsymbol{g}}= X^2_R$. Since $\Delta_{\boldsymbol{g}}= X^2_{\boldsymbol{g}}$  the proof ends.  \hfill $\Box$\\

\noindent {\bf Proof of Proposition \ref{PROPESA}}.
$(G,\boldsymbol{g})$ is complete because geodesics are integral curves of left invariant fields (Proposition \ref{PROPO6}), i.e. one parameter-subgroups which are complete by definition.   The volume form $\mu_{\boldsymbol{g}}$ induced by ${\boldsymbol{g}}$ coincides with the Haar measure $\mu_G$ (Corollary \ref{remHAAR}) which is bi-invariant because the group admits a bi-invariant metric.
The operator $\Delta_{\boldsymbol{g}}$ on complete Riemannian manifolds defined on smooth compactly supported functions is essentially self-adjoint on the $L^2$ space constructed out of the volume form generated by ${\boldsymbol{g}}$ as  can be proved is several different ways, for instance using the general approach by Chernoff \cite{ChernoffSA}.  \hfill $\Box$\\

\noindent {\bf Proof of Proposition \ref{PROPSPECD}}.
\ Properties (1)-(3)  are classic results for the Laplace-Beltrami operator with suitable domain
$\Delta_{\boldsymbol{g}}: C^\infty_c(M; \mathbb{C})\to L^2(M, \mu_{\boldsymbol{g}})$ in a smooth compact Riemannian manifold   \cite{Chavel}, so that they immediately generalises to our specific case. The  decompositions in (5) arises immediately from the spectral theorem for selfadjoint operators with pure-point spectrum (see, e.g. \cite{Moretti-17}). Property (5) easily  follows from the spectral decomposition of $\overline{\Delta_{\boldsymbol{g}}}$, taking (1), (2), and (4)  into account.
\hfill $\Box$\\

\noindent {\bf Proof of Proposition \ref{PROPRAPR}}.
(1) If $\Delta_{\boldsymbol{g}}$ is invariant under the isometries of the bi-invariant metric ${\boldsymbol{g}}$, in particular the right translations. If $\psi \in H_\lambda$, then it is smooth for Proposition \ref{PROPESA} and $$\overline{\Delta_{\boldsymbol{g}}} \pi_R(x) \psi = \Delta_{\boldsymbol{g}} \pi_R(x) \psi =\pi_R(x) \Delta_{\boldsymbol{g}} \psi = \lambda \pi_R(x)\psi\:,$$ so that $\pi_R(x) (H_\lambda) \subset H_\lambda$. As  $\pi_R$ is continuous and the Hilbert decomposition (\ref{DECH}) holds,  this leads to the decomposition of $\pi_R$ as in the thesis. Notice that, since  $H_\lambda$ is finite dimensional, it is irreducible under the unitary representation  $\pi_R$ or it is a direct (orthoghonal) sum of such finite dimensional representations.\\
(2) Let $\psi\in H_\lambda$. By hypothesis, $(\pi_R(\exp(tX)) \psi)(x) = \psi(x \exp(tX))$.
As a consequence
$$\lim_{t\to 0} \frac{(\pi_R(\exp(tX)) \psi)(x) - \psi(x)}{t} = (\widetilde{X}\psi)(x)\:.$$
Notice that, since all vectors $(\pi_R(\exp(tX)) \psi$ belongs to the closed subspace $H_\lambda$ which is finite dimensional, the limit is also valid in $L^2$ sense. The Stone theorem \cite{Moretti-17} implies that $\psi \in D(X^R)$ and $-iX^R \psi= X(\psi)$. Notice that, since $H_\lambda$ is closed because finite dimensional and  $\frac{(\pi_R(\exp(tX)) \psi) - \psi}{t}\in H_\lambda$ due to (1), also the limit of it belongs to $H_\lambda$, that is $X^R\psi \in H_\lambda$. We have proved (a) and (b).
Identity (c) immediately arises from (b), (\ref{OPLUSP}) and the fact that in finite dimensional vector spaces $e^{itA}x= \sum_{n=0}^{+\infty} \frac{(it)^nA^n}{n!}x$.
The last statement is an obvious consequence of (c).
This concludes the proof of (2).\\
(3) From (\ref{DECH2}) and standard constructions of spectral theory (see, e.g. \cite{Moretti-17}), we have that, if $Re(z)\geq 0$,
$$e^{-z \overline{\Delta_{\boldsymbol{g}}}} \psi = \sum_{n=0}^{+\infty} e^{-z\lambda_n} P_{n}\psi \:,\quad \forall \psi \in L^2(G, \mu_G)\:,$$
where the series converges in the topology of the Hilbert space. This decomposition immediately produces the thesis. 
\hfill $\Box$\\

\noindent {\bf Proof of proposition \ref{PROPALG}}.
Evidently, the constantly unit function $1$ belongs to the set $Span\{H_\lambda\:|\: \lambda \in \sigma(-\overline{\Delta_{\boldsymbol{g}}}) \}$ because $\Delta_{\boldsymbol{g}}1 =0$ so that $1\in H_0$. Furthermore, the considered linear span is closed with respect to the complex conjugation since $\Delta_{\boldsymbol{g}}$ is real.
To conclude it is sufficient to prove that $Span\{H_\lambda\:|\: \lambda \in \sigma(-\overline{\Delta_{\boldsymbol{g}}}) \}$ is closed with respect to the product and that it is dense in $C(G)$ in the norm $\|\cdot\|_\infty$. This is consequence of  the fact that $Span\{H_\lambda\:|\: \lambda \in \sigma(-\overline{\Delta_{\boldsymbol{g}}}) \}$ coincides with  the space ${\cal B}$ of \textit{trigonometric polynomials} on $G$ (Definition 7.6.1 in \cite{RT}) which is more generally given for topological compact groups.
 ${\cal B}$ is per definition made of the  finite  linear combination of \textit{continuous} functions $\phi_{ij}: G \to \mathbb{C}$ where $\phi$ is a representative for each equivalence class of topologically irreducible (finite dimensional according to the Peter-Weyl  theorem) unitary, strongly continuous, representations of $G$.
For every given $i,j \in \{1,\ldots, \dim(\phi)\}$, the functions $\{\phi_{ij}\}_j$, span a subspace $H_{\phi, i}\subset L^2(G;\mu_G)$
which is  invariant and irreducible under  $\pi_R$.
As a matter of fact, $\phi = \pi_R|_{H_{\phi, i}}$  up to unitary equivalence. 
 If $G$ is Lie and equipped with a bi-invariant Riemannian metric, for $\phi$ and $i$ fixed, all the functions $\phi_{ij}$ must belong to some space $H_\lambda$ and, as a consequence,  they are smooth for Proposition \ref{PROPSPECD}. Indeed,  
 $H_{\phi,i} \cap H_\lambda \neq \emptyset$ for some $\lambda$, since the direct orthogonal (Hilbert) sum of the $H_\lambda$ exhausts the whole Hilbert space in view of Proposition \ref{PROPSPECD}.
 On the other hand $H_{\phi,i} \cap H_\lambda$ must be invariant under $\pi_R$ by construction and thus it is an invariant (closed because we are working in finite dimensions)  subspace of $H_{\phi,i}$. Since this space is (topologically) irreducible, the only possibility is that $H_{\phi,i} \subset H_\lambda$ (the two spaces might coincide). This inclusion immediately implies that ${\cal B} \subset Span\{H_\lambda\:|\: \lambda \in \sigma(-\overline{\Delta_{\boldsymbol{g}}}) \}$.
 However, also the converse inclusion holds true.
 In fact, also the direct orthogonal (Hilbert) sum of the finite dimensional subspaces $H_{\phi,i}$ exhaust the whole Hilbert space for the Peter-Weyl theorem.
 Therefore  each $H_\lambda$ must intersect some $H_{\phi,i}$ and thus it includes $H_{\phi,i}$.
The space $H_{\phi,i}^{\perp_{H_\lambda}} \subset H_\lambda$ is still $\pi_R$ invariant (the proof is elementary since $\pi_R$ is unitary)  and it must intersect, and thus include, another space $H_{\phi',i'}$, and so on.
The process must end after a finite number of steps in view of the finite dimensionality of $H_\lambda$.
At the end of the game $H_\lambda$ turns out to be a finite orthogonal sum of suitable spaces $H_{\phi,i}$.
This fact implies the wanted remaining inclusion  ${\cal B} \supset Span\{H_\lambda\:|\: \lambda \in \sigma(-\overline{\Delta_{\boldsymbol{g}}}) \}$. As ${\cal B}$ is a dense unital subalgebra of $C(G)$ (Theorem 7.6.2 of \cite{RT} which is more generally valid for  compact topological groups),  the proof is over because the remaining statements have been already established or are obvious.  $\hfill \Box$\\

\end{document}